\documentclass[11pt]{article}
\usepackage{macros}

\newcommand*\samethanks[1][\value{footnote}]{\footnotemark[#1]}

\title{Understanding Robust Catalytic Computing}
\author{ Michal Kouck\'y\thanks{Partially supported by the Grant Agency of the Czech Republic under the grant agreement no. 24-10306S and by Charles Univ. project UNCE 24/SCI/008.} \\ Charles University \\ \texttt{koucky@iuuk.mff.cuni.cz} \and Ian Mertz\samethanks[1] \\ Charles University \\ \texttt{iwmertz@iuuk.mff.cuni.cz} \and Sasha Sami\samethanks[1] \\ Charles University \\ \texttt{sashasami@iuuk.mff.cuni.cz} }

\date{}

\begin{document}


\maketitle

\vspace{-3em}
\begin{abstract}
    Catalytic computing concerns space bounded computation which starts with a memory full of data that must be restored by the end of the computation. Lossy catalytic computing, defined by Gupta et al. (CoRR 2024) and fully characterized by Folkertsma et al. (Algorithmica 2026), is the study of allowing a small number of errors when resetting the catalytic tape at the end of a computation. Such a notion is useful when considering the robust use of catalytic techniques in the study of ordinary space-bounded algorithms. To that end, however, defining and characterizing less strict notions of error was left open by Folkertsma et al. (Algorithmica 2026) and other works (Mertz, BEATCS 2023).

    We expand the definition of possible resetting error in three natural ways: (1) randomized catalytic computation, which can \emph{completely destroy} the catalytic tape with some probability over the randomness; (2) randomized catalytic computation, which makes a bounded number of errors \emph{in expectation over the randomness}; (3) deterministic catalytic computation, which makes a bounded number of errors \emph{in expectation over the initial catalytic tape itself}. We show a near complete characterization of the above models, both in the general case and in the logspace polynomial-time regime, by showing equivalences either between one another, to errorless catalytic space models, or to standard time or space complexity classes. Under a derandomization assumption, we show a near full collapse of all existing catalytic classes in the logspace regime.
\end{abstract}

\thispagestyle{empty}
\setcounter{tocdepth}{2}
\tableofcontents
\thispagestyle{empty}
\setcounter{page}{0}
\clearpage

\section{Introduction}
\label{sec:intro}

\subsection{Catalytic computation}
\label{sec:intro:catalytic}
When designing space-bounded algorithms, a typical assumption has been that memory being used for storage is locked away and should not be touched while running unrelated subroutines. A recent paradigm called \emph{catalytic computing} has challenged this assumption, showing that given memory that is full of \emph{arbitrary data}, a machine can use this memory to aid computations while still \emph{perfectly resetting it} at the end of the computation.

Since the formulation of the model by Buhrman, Cleve, Kouck{\'y}, Loff, and Speelman~\cite{BuhrmanCleveKouckyLoffSpeelman14}, many works have studied the properties of computing with the assistance of a large full hard drive known as \emph{catalytic space}.
The original result of \cite{BuhrmanCleveKouckyLoffSpeelman14} showed that \emph{catalytic logspace} ($\CL$) is seemingly strictly more powerful than $\NL$, $\BPL$, or even stronger classes; this was later strengthened to include important problems such as bipartite matching~\cite{AgarwalaMertz25}, linear matroid intersection~\cite{AgarwalaAlekseevVinciguerra26}, and matrix powering~\cite{AlekseevFilmusMertzSmalVinciguerra25}.
A structural theory of catalytic space has emerged regarding the power of randomness~\cite{DattaGuptaJainSharmaTewari20,CookLiMertzPyne25,Pyne25,KouckyMertzPyneSami25}, non-determinism~\cite{BuhrmanKouckyLoffSpeelman18,GuptaJainSharmaTewari19,CookLiMertzPyne25,KouckyMertzPyneSami25}, and non-uniformity~\cite{Potechin17,RobereZuiddam21,CookMertz22,CookMertz24}.
Finally, there have been attempts to situate catalytic space in other settings, such as time-space efficient algorithms~\cite{CookPyne26}, alternative models~\cite{BisoyiDineshSarma22,BuhrmanFolkertsmaMertzSpeelmanStrelchukSubramanianTupker25,CookGhentiyalaMertzPyneSheffield25}, and, most importantly, the study of non-catalytic space.

With regards to this final point, the techniques of using full memory were also used in the context of ordinary space-bounded computation, in particular in the realms of space-bounded derandomization~\cite{DoronPyneTell24,LiPyneTell24,Pyne25,DoronPyneTellWilliams25,GuptaSharmaTewari25} and recently on the relationship of space and time \cite{CookMertz20,CookMertz21,CookMertz24,Williams25,Shalunov25}. These works used the aforementioned framework to repeatedly compute subroutines without a commensurate blowup in space, thus validating a key motivation for the introduction of catalytic space in the first place.

\subsection{Robust catalytic computation}
The ultimate goal of studying catalytic computation in the context of non-catalytic space is to utilize memory during a computation for both storage and other subcomputations. To that end, we may eject some of the restrictions of fully general catalytic computing by considering the problem and stored data at hand. For example, Bisoyi et al.~\cite{BisoyiDineshRaiSarma25} study the question of resetting the catalytic tape only when it contains some worthwhile information, or alternatively of ending with the tape in a string which shares some property of interest with its initial configuration.

Gupta et al.~\cite{GuptaJainSharmaTewari24} took a different approach, asking whether the model is \emph{robust} to allowing a small number of errors when resetting the catalytic tape, a variant they called \emph{lossy catalytic computing}. This latter question was answered by Folkertsma et al.~\cite{FolkertsmaMertzSpeelmanTupker26} when the number of errors is bounded for any run of the machine, showing that $e$ errors on a catalytic tape of length $c$ is equivalent in power to having $\Theta(e \log c)$ bits of additional free memory.

The more general question of failing to reset the catalytic tape, such as having a bounded number of errors in \emph{expectation} over a generic run of the machine, were posed by Mertz~\cite{Mertz23} and remain open. Understanding the richer landscape of potential robustness conditions on the catalytic tape will allow us to understand when memory can and cannot be reused outside of the strict catalytic context, where averaging arguments and other structural results are known.

\subsection{Our results}

\subsubsection{New robust classes}
In this paper we take a systematic approach to understand lossy catalytic computation.
To do so we introduce a host of different lossy catalytic models and show how the resulting complexity classes group around one another and around established classes. Fixing $s$, $c$, and (sometimes) $e$ to be our free space, catalytic space, and error respectively, our three variants are
\begin{enumerate}
    \item $\BPepsCdeltSpace{s}{c}{\delta}{\epsilon}$: randomized catalytic computation which, in addition to succeeding with probability $\frac{1}{2} + \epsilon$, has a probability $\delta$ to \emph{arbitrarily fail} to reset the catalytic tape, and which faithfully resets the catalytic tape otherwise
    \item $\AVGrBPLCSPACE{s}{c}{e}$: randomized catalytic computation which makes at most $e$ resetting errors \emph{in expectation over the randomness}
    \item $\AVGtauLCSPACE{s}{c}{e}$: \emph{deterministic} catalytic computation which makes at most $e$ resetting errors \emph{in expectation over the initial contents of the catalytic tape itself}
\end{enumerate}

\noindent
We focus on three restrictions: 1) the \emph{logspace} versions of these classes, which we denote the same way as $\CL$; 2) the \emph{constant-error} regime (i.e. $e = O(1)$), which we explicitly denote; and 3) we can also study the \emph{polynomial-time} bounded variants of the catalytic logspace classes, which we denote by appending $\PT$. See Figure \ref{fig:definitions} for a summary of our new classes and their definitional containments therein. We note that all theorems apply to more general $s$, $c$, and $e$, and will be stated as such.

The classes $\BPepsCdeltSpace{s}{c}{\delta}{\epsilon}$ and $\AVGrBPLCSPACE{s}{c}{e}$ capture two seemingly incomparable notions of randomized error. The former makes $\Omega(c)$ errors in expectation
but does so in a highly structured way, whereas the latter has more leeway about how the errors appear but only permits $e$ (typically $O(1)$) on average. Our third class $\AVGtauLCSPACE{s}{c}{e}$ appears to be much different, lacking randomness per se but being allowed to behave in a less careful way when the catalytic tape is ``far from random''.

With regard to selecting such classes, the latter two expectation classes are (in our view) the most general but plausible ``small-error'' classes — indeed, they are the natural instantiations of questions raised by \cite{Mertz23,FolkertsmaMertzSpeelmanTupker26} — while the former class captures the most reasonable notion of structured/controlled ``large-error'' catalytic computing. We also note that $\BPepsCdeltSpace{s}{c}{\delta}{\epsilon}$ corresponds to how low-probability errors compose over long computations; no matter how small or unlikely an early mistake may be, in all existing catalytic algorithms these will ripple forward to create arbitrary resetting issues by the end.

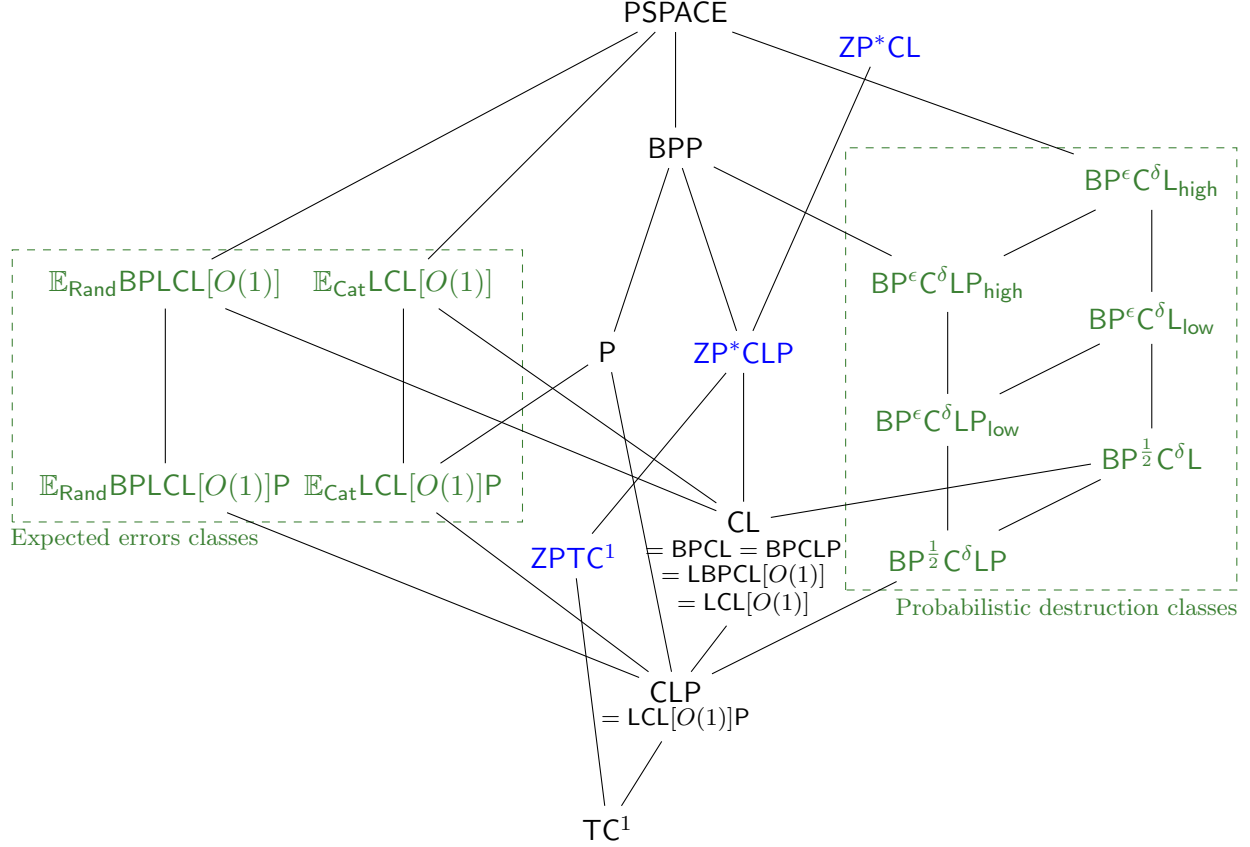
\begin{figure}[!thb]
\begin{tikzpicture}[scale=0.9
]

\node(pspace) at (10,10) {$\PSPACE$};
\node(bpp) at (10,8) {$\BPP$};
\node(p) at (9,5) {$\PT$};
\node(cl) at (11,2.5) {$\CL$};
\node(clmore1) at (11,2.1) {\footnotesize $= \BPCL = \BPCLP$};
\node(clmore2) at (11,1.7) {\footnotesize $= \BPLCL{O(1)}$};
\node(clmore3) at (11,1.3) {\footnotesize $= \LCL{O(1)}$};
\node(clp) at (10,0) {$\CL\PT$};
\node(clpmore) at (10,-0.4) {\footnotesize $= \LCLP{O(1)}$};
\node(tc1) at (9,-2) {$\TCo$};

\node(zpstarcl) at (13,9.5) {\color{blue} $\ZPstartCL$};
\node(zpstarclp) at (11,5) {\color{blue} $\ZPstartCLP$};
\node(zptc1) at (8.5,2) {\color{blue} $\ZPTCo$};

\node(bplclb) at (17,7.5) {\color{OliveGreen} $\BPepsCdeltLB$};
\node(bplclbp) at (14,6) {\color{OliveGreen} $\BPepsCdeltLBP$};
\node(bplclg) at (17,5.5) {\color{OliveGreen} $\BPepsCdeltLG$};
\node(bplclgp) at (14,4) {\color{OliveGreen} $\BPepsCdeltLGP$};
\node(bplclone) at (17,3.5) {\color{OliveGreen} $\BPepsCdeltL{\delta}{\frac{1}{2}}$};
\node(bplclonep) at (14,2) {\color{OliveGreen} $\BPepsCdeltL{\delta}{\frac{1}{2}}\PT$};

\draw[dashed,draw = OliveGreen] (12.5,1.5) rectangle (18.25,8);
\node(label1) at (15.75,1.25) {\color{OliveGreen} \footnotesize Probabilistic destruction classes};

\node(erlbpcl) at (2,6) {\color{OliveGreen} $\AVGrBPLCL{O(1)}$};
\node(erlbpclp) at (2,3) {\color{OliveGreen} $\AVGrBPLCL{O(1)}\PT$};
\node(etlcl) at (6,6) {\color{OliveGreen} $\AVGtauLCL{O(1)}$};
\node(etlclp) at (6,3) {\color{OliveGreen} $\AVGtauLCL{O(1)}\PT$};

\draw[dashed,draw = OliveGreen] (-0.25,2.5) rectangle (7.75,6.5);
\node(label2) at (1.55,2.25) {\color{OliveGreen} \footnotesize Expected errors classes};


\draw (pspace) -- (bpp);
\draw (bpp) -- (p);
\draw (p) -- (clp);
\draw (clmore3) -- (clp);

\draw (bpp) -- (zpstarclp);
\draw (zpstarcl) -- (zpstarclp);
\draw (zpstarclp) -- (zptc1);
\draw (zpstarclp) -- (cl);
\draw (zptc1) -- (tc1);
\draw (clpmore) -- (tc1);

\draw (pspace) -- (bplclb);
\draw (bpp) -- (bplclbp);
\draw (bplclb) -- (bplclbp);
\draw (bplclb) -- (bplclg);
\draw (bplclg) -- (bplclgp);
\draw (bplclbp) -- (bplclgp);
\draw (bplclg) -- (bplclone);
\draw (bplclone) -- (bplclonep);
\draw (bplclgp) -- (bplclonep);
\draw (bplclone) -- (cl);
\draw (bplclonep) -- (clp);

\draw (pspace) -- (erlbpcl);
\draw (pspace) -- (etlcl);
\draw (erlbpcl) -- (erlbpclp);
\draw (etlcl) -- (etlclp);
\draw (p) -- (etlclp);
\draw (erlbpcl) -- (cl);
\draw (etlcl) -- (cl);
\draw (erlbpclp) -- (clp);
\draw (etlclp) -- (clp);

\end{tikzpicture}
\caption{Basic class structure for logspace catalytic classes. Color coding: black = old classes; green + rectangle = robust catalytic classes introduced in this work; blue = reference classes which were implicitly previously defined but which have not been previously studied.  Upwards lines indicate known or definitional inclusions.} 
\label{fig:definitions}
\end{figure}

\subsubsection{Relationships}
A priori, there is no relationship for the aforementioned robust classes nor any extension of \cite{FolkertsmaMertzSpeelmanTupker26} to collapse them to e.g. $\CL$. Nevertheless, using a variety of techniques we show a near full characterization of the above models. Focusing on the logspace setting, we obtain the following:
\begin{itemize}
    \item we show three regimes for $\BPepsCdeltL{\delta}{\epsilon}$: 1) $\BPepsCdeltL{\delta}{\epsilon}_{\high}$ when $\delta > 2\epsilon$; 2) $\BPepsCdeltL{\delta}{\epsilon}_{\low}$ when $\delta < 2\epsilon < 1$; and 3) $\BPepsCdeltL{\delta}{\frac{1}{2}}$, i.e. $\delta < 2\epsilon = 1$. While the third is an intermediate class with no clear connections, the first can capture the full power of $\PSPACE$, while the second is equivalent to the $\AVGrBPLCL{e}$ model for $e = O(1)$ (up to some parameter slackness).
    \item with the additional poly-time restriction, the three aforementioned regimes of $\BPepsCdeltL{\delta}{\epsilon}\PT$ correspond exactly to $\BPP$, $\CL$ (and $\AVGrBPLCL{O(1)}\PT$), and poly-time $\CL$ respectively.
    \item $\AVGtauLCL{e}$ captures (seemingly) more power than $\AVGrBPLCL{e}$ when no time bound is imposed; we show that it exactly corresponds to (zero-error) randomized catalytic computing with \emph{two-way access} to its randomness, a class which, if shown to be equal to $\CL$, would imply major non-catalytic derandomizations as well.
    \item with the additional poly-time restriction, $\AVGtauLCL{e}\PT$ and $\AVGrBPLCL{e}\PT$ become incomparable, with the former being in $\PT$ for all $e$ and capturing the intersection of $\AVGtauLCL{e}$ and $\PT$ in the case of $e=O(1)$ (note that $\AVGrBPLCL{O(1)}\PT = \CL$ is not known to be in $\PT$).
    \item under a standard derandomization assumption about space (\autoref{conj:derandom_assumption}), all of the aforementioned models collapse to $\CL$ and even poly-time $\CL$ (besides $\BPepsCdeltL{\delta}{\epsilon}_{\high} = \PSPACE$ and $\BPepsCdeltL{\delta}{\epsilon}\PT_{\high} = \BPP = \PT$).
\end{itemize}


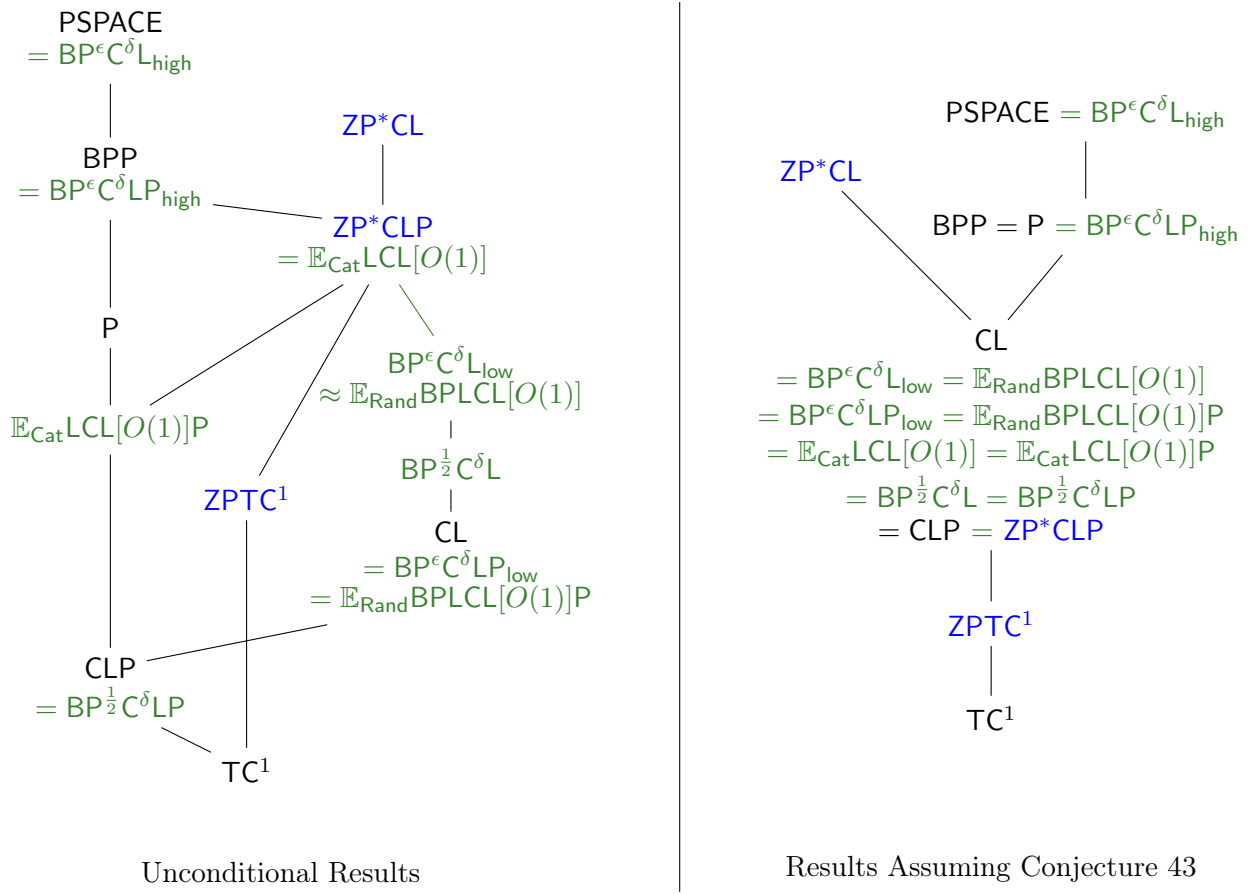
\begin{figure}[!thb]
\begin{tikzpicture}[scale=0.9
]

\node(pspace1) at (9,10) {$\PSPACE$};
\node(pspace2) at (9,9.5) {\color{OliveGreen} $= \BPepsCdeltLB$};
\node(bpp1) at (9,8) {$\BPP$};
\node(bpp2) at (9,7.5) {\color{OliveGreen} $= \BPepsCdeltLBP$};
\node(p) at (9,5.5) {$\PT$};


\node(zpstarcl) at (13,8.5) {\color{blue} $\ZPstartCL$};
\node(zpstarclp) at (13,7) {\color{blue} $\ZPstartCLP$};
\node(zptc1) at (11,3) {\color{blue} $\ZPTCo$};
\node(tc1) at (11,-1) {$\TCo$};

\node(etlclp) at (9,4) {\color{OliveGreen} $\AVGtauLCL{O(1)}\PT$};

\node(etlcl2) at (13,6.5) {\color{OliveGreen} $= \AVGtauLCL{O(1)}$};

\node(erlbpcl1) at (14,5) {\color{OliveGreen} $\BPepsCdeltLG$};
\node(erlbpcl2) at (14,4.5) {\color{OliveGreen} $\approx \AVGrBPLCL{O(1)}$};

\node(bplclone) at (14,3.5) {\color{OliveGreen} $\BPepsCdeltL{\delta}{\frac{1}{2}}$};

\node(cl1) at (14,2.5) {$\CL$};
\node(cl2) at (14,2) {\color{OliveGreen} $= \BPepsCdeltLGP$};
\node(cl2) at (14,1.5) {\color{OliveGreen} $=\AVGrBPLCL{O(1)}\PT$};

\node(clp1) at (9,0.5) {$\CL\PT$};
\node(clp2) at (9,0) {\color{OliveGreen} $= \BPepsCdeltL{\delta}{\frac{1}{2}}\PT$};

\node (cap) at (11.5,-2.5) {Unconditional Results};

\draw (pspace2) -- (bpp1);
\draw (bpp2) -- (p);
\draw (p) -- (etlclp);
\draw (etlclp) -- (clp1);

\draw (bpp2) -- (zpstarclp);
\draw[draw=OliveGreen] (etlcl2) -- (erlbpcl1);
\draw (erlbpcl2) -- (bplclone);
\draw (bplclone) -- (cl1);
\draw (cl2) -- (clp1);
\draw (clp2) -- (tc1);

\draw (etlcl2) -- (zptc1);
\draw (etlcl2) -- (etlclp);
\draw (zptc1) -- (tc1);
\draw (zpstarcl) -- (zpstarclp);

\end{tikzpicture}
\hspace{2em}
\vline
\hspace{2em}
%
\begin{tikzpicture}[
]

\node(pspace) at (11.25,8) {$\PSPACE$ {\color{OliveGreen} $= \BPepsCdeltLB$}};

\node(bpp1) at (11.25,6.5) {$\BPP = \PT$ {\color{OliveGreen} $= \BPepsCdeltLBP$}};

\node(cl1) at (10,5) {$\CL$};
\node(cl2) at (10,4.5) {\color{OliveGreen} $=\BPepsCdeltLG = \AVGrBPLCL{O(1)}$};
\node(cl3) at (10,4) {\color{OliveGreen} $= \BPepsCdeltLGP =\AVGrBPLCL{O(1)}\PT$};
\node(cl4) at (10,3.5) {\color{OliveGreen} $= \AVGtauLCL{O(1)} = \AVGtauLCL{O(1)}\PT$};
\node(cl5) at (10,3) {{\color{OliveGreen} $ = \BPepsCdeltL{\delta}{\frac{1}{2}} = \BPepsCdeltL{\delta}{\frac{1}{2}}\PT$}};
\node(cl6) at (10,2.5) {$= \CL\PT$ {\color{OliveGreen} $ =$} {\color{blue} $\ZPstartCLP$}};

\node(zptc1) at (10,1.25) {\color{blue} $\ZPTCo$};
\node(tc1) at (10,0) {$\TCo$};

\node(zpstarcl) at (7.75,7.25) {\color{blue} $\ZPstartCL$};

\node (cap) at (10,-2) {Results Assuming \autoref{conj:derandom_assumption}};

\draw (pspace) -- (bpp1);
\draw (bpp1) -- (cl1);
\draw (zpstarcl) -- (cl1);
\draw (cl6) -- (zptc1);
\draw (zptc1) -- (tc1);


\end{tikzpicture}

\caption{Summary of results (all green equalities and containments). In the unconditional setting, all new classes are either 1) equal to a preexisting class; 2) find themselves in two new clusters (note that $\approx$ means the parameters are not tight enough for full equality); or 3) are one of two stray classes. In the conditional setting, almost all new classes collapse to $\CL$ and all classes are comparable with the sole exception of $\ZPstartCL$.} 
\label{fig:results}
\end{figure}

\noindent
See Figure \ref{fig:results} for a summary of our results.


\subsubsection{Open questions}
The main gap in our results is to show that $\BPepsCdeltL{\delta}{\epsilon}_{\low}$ and $\AVGrBPLCL{O(1)}$ both collapse to $\CL$ unconditionally; this would seem most doable for $\BPepsCdeltL{\delta}{\frac{1}{2}}$ as a warm-up. The only other stray class is $\AVGtauLCL{O(1)}\PT$, whose complexity remains a mystery which would be useful to solve but may require new ideas in catalytic computing.

Our results seem to capture all remaining natural models of resetting error.
One could always take further combinations of the aforementioned relaxations, such as allowing expected error over both the initial tape $\tau$ as well as randomness, but we expect these to collapse to one of our clusters without much difficulty. Alternatively, one could use this framework to study \emph{average-case catalytic computation}, i.e., where the resetting error is in expectation over the \emph{input}; as a setting, this seems orthogonal to catalytic robustness, but perhaps similar techniques can be of use.

\subsection{Our techniques}
Because we have many results related to a host of different catalytic classes, there are various techniques at play.
However, a key throughline in many of our proofs is an analysis of the \emph{configuration graph} of our robust catalytic machines.
Thus, we will briefly discuss the common factors of these results at a high level.

\paragraph{Configuration graphs.}
Configuration graphs of catalytic machines are one of the key tools in understanding catalytic machines, underlying results such as $\CL = \BPCL$ \cite{CookLiMertzPyne25}, $\CL=\CNL$ \cite{KouckyMertzPyneSami25} and $\CL \cap \PT = \CL\PT$ \cite{CookLiMertzPyne25}.
The collapses to $\CL$ of both $\BPCL$ \cite{CookLiMertzPyne25,KouckyMertzPyneSami25} and $\CNL$ \cite{KouckyMertzPyneSami25}
are facilitated by efficient procedures that explore and analyze the configuration graph on a given input.
The configuration graph of a catalytic machine that uses work space $s$ and catalytic space $c$ on a given input $x$ might have up-to $2^{s+c}$ configurations, one for each $\conf{\pi}{u}$ where the machine has  $\pi$ on its catalytic tape and $u$ on its work tape.\footnote{Here we assume that the head positions and internal state are automatically recorded as part of the workspace.}
When $s=O(\log n)$ and $c=n^{O(1)}$, the graph is exponentially large, which presents a challenge to analyze.

For a machine $M$ and input $x$, we denote the resulting configuration graph by $\gmx$. 
For a given initial setting $\tau$ of the catalytic tape, we let $\gmxtau$ be the induced subgraph of $\gmx$ on configurations reachable from the starting configuration $\start$.
The computation of any catalytic machine is always required to reach a final configuration with $\tau$ back on the catalytic tape.
Without loss of generality, we can assume that $M$ has two final configurations, $\acc$ and $\rej$, which are easy to recognize; these configurations are sinks of $\gmxtau$.

For deterministic machines $M$, it follows that $\gmxtau$ and $\gmxtau[\tau']$ are disjoint for distinct $\tau$ and $\tau'$.
Thus, on average (over $\tau$), $\gmxtau$ is of size at most $2^s = \poly(n)$.
This fact is heavily used in many previous results on catalytic computation.
Most important for our algorithms is that any $\tau$ for which $\gmxtau$ has a much larger size is far from average, i.e., it cannot be very random.

Using this logic, Cook et al.~\cite{CookLiMertzPyne25} coined a technique called \emph{compress-or-random}.
Their approach is to analyze $\gmxtau$ and compress $\tau$
if the algorithm detects that $\gmxtau$ has size $\gg 2^s$, and then restart the analysis with a new section of the catalytic tape.
This eventually leads to either finding a $\tau$ that has a small configuration graph $\gmxtau$, or freeing up so much space on the catalytic tape 
that we can run an ordinary space-unbounded algorithm for our original problem.
Crucially, we are required to design efficient compression and decompression methods for $\tau$ in order to carry out this argument.

In our setting, the configuration graphs are more complicated.
Since the machine does not always have to restore the content of the catalytic tape, two different configuration graphs $\gmxtau$ and $\gmxtau[\tau']$ might intersect for $\tau \neq \tau'$.
This presents a challenge for analyzing $\gmx$ as the above averaging argument does not work.

\paragraph{Compression for lossy configuration graphs.}
Our key technical lemmas \autoref{lemmvianisansprggeneral} and \autoref{lemmvianisansprgforcl} address this challenge.
For a randomized machine $M$, we call a configuration of $\gmxtau$ a $\taubeta$-node if, from that configuration, we reset $\tau$ correctly with a probability of at least $\beta$.
Clearly, the same configuration can be a $\taubeta$-node for at most $1/\beta$ distinct values of $\tau$;
if the machine $M$ restores the catalytic tape with a probability of at least $1-\delta$, then the probability of reaching a node that is {\em not} a $\taubeta$-node is at most $\delta/(1-\beta)$.
Thus, for say $\beta<1/2$, if the randomized machine resets the catalytic tape correctly with high probability,
then it must be visiting only $\taubeta$-nodes with high probability.
Therefore, to understand the behavior of $M$ on $x$, we only need to analyze the subgraph of $\gmxtau$ consisting of $\taubeta$-nodes, and so our goal while exploring $\gmxtau$ will be to identify all $\taubeta$-nodes, of which we expect to have at most $2^s/\beta$, and analyze the subgraph consisting of them.

As observed in \cite{KouckyMertzPyneSami25}, it is beneficial to explore $\gmxtau$ \emph{backwards}, starting from the final configurations $\acc$ and $\rej$; this allows us to explore the graph while always ensuring that we can return to the correct $\acc$ or $\rej$.
Since we do not have space to store the explored portion of the graph and cannot use randomness for the exploration, as that could lead us to a place where we cannot reset, 
we will explore the graph using a \emph{pseudo-random generator}.
We will use a modification of Nisan's pseudo-random generator \cite{Nisan92} based on hash functions, where we will use as many hash functions in a sequence as we are generating bits.
The sequence of hash functions will be taken from a portion of the catalytic tape, and we will use the compression technique of $\cite{Pyne25}$ to select good hash functions.

We will use the bits from the generator to guide us backward from the final nodes $\acc$ and $\rej$, adding more hash functions to extend our horizon.
Keeping the same hash functions for the later portion of the walk is useful because it fixes the part of the graph we have already explored.
The pseudo-random generator also allows us to estimate well the probability of reaching either $\acc$ or $\rej$, hence determining whether a node is a $\taubeta$-node or not.

Our ultimate goal is to find the node $\start$, classify it as a $\taubeta$-node, and estimate the probability of reaching $\acc$.
As we extend the horizon, we will always check how many nodes we see in total using the pseudo-random walks.
If we encounter too many nodes, we will be able to use the large graph compression of \cite{CookLiMertzPyne25} to compress our current catalytic tape and restart the exploration.
Similarly, if we see sufficiently many $\taubeta$-nodes, we will also be able to compress the catalytic tape.

The only issue that will arise is if the start node $\start$ is very far from the final nodes on average.
In that case, we will run out of bits produced by our pseudo-random generator, and yet we will not be able to compress to the best of our knowledge.
In \autoref{lemmvianisansprgforcl}, we avoid this case altogether as the computation time is bounded by some fixed polynomial,
but in the general case, we will give up, as this will be a rare occurrence over the choice of $\tau$; this will determine the equivalences of our various error-prone catalytic classes.

For our conditional results, we can overcome the issue of such $\tau$'s by using strong pseudo-random generators and XORing $\tau$ with their output to obtain a more ``typical'' starting tape similar to the technique of \cite{BuhrmanKouckyLoffSpeelman18}.
Under a standard derandomization assumption, the Nisan-Wigderson pseudo-random generator is strong enough to guarantee that most of its output XOR-ed with $\tau$
avoids the aforementioned bad case.
\section{Preliminaries}
\label{sec:prelims}

Let $[n] = \{1,\ldots,n\}$. For a graph $G$, we denote its vertex set by $V(G)$. Unless stated otherwise, by the size of a graph $G$ we mean $|V(G)|$. We denote the length of a string $y$ by $|y|$. The notation $x \circ y$ stands for the concatenation of the strings $x$ and $y$. For a string $y \in \{0,1\}^n$, we write $y_i$ for the $i$-th bit of $y$ (where $i \in [n]$),  $y_{\le i}$ for the prefix of $y$ consisting of its first $i$ bits, and $y_{\ge i}$ for the suffix of $y$ starting at $y_i$.

\subsection{Catalytic computation}

For the remainder of the paper, we will define functions \( s \coloneqq s(n) \), \( c \coloneqq c(n) \), and sometimes \(e \coloneqq e(n)\), such that all are non-decreasing logspace-constructible
functions satisfying 1) \( s \geq \log{n} \); 2) \(s \leq c \leq 2^s \); and 3) \(e \leq c\). Note that putting all the above bounds together implies that all functions can be computed in space $O(s)$.

Our paper concerns the \emph{catalytic computing} model, as defined by Buhrman et al.~\cite{BuhrmanCleveKouckyLoffSpeelman14}:

\begin{definition}[Catalytic computation]
    A \emph{catalytic} Turing machine is a machine $M$ with four tapes: 1) a read-only input tape; 2) a write-only output tape; 3) a read-write work tape; and 4) a read-write \emph{catalytic} tape. $M$ has the property that for any initializations $x$ to the input tape and $\tau$ to the catalytic tape, $M$ halts with $\tau$ in the catalytic tape.
\end{definition}

\noindent
Such catalytic machines naturally give rise to complexity classes based on the available resource parameters.

\begin{definition}
    We define $\CSPACE{s}{c}$ to be the class of languages $L$ that can be decided by a catalytic machine with $s$ bits of work space and $c$ bits of catalytic space on inputs of length $n$, i.e., on input $x \in \{0,1\}^n$ and initial catalytic tape $\tau \in \{0,1\}^c$, $M$ halts with $L(x)$ in the output tape and $\tau$ in the catalytic tape.
\end{definition}

\noindent
The most important instantiation of the catalytic space class is \emph{catalytic logspace}. In this context, we can also talk about an additional \emph{polynomial time} restriction.

\begin{definition}[Catalytic logspace (CL) and polytime CL]
    We define the class \emph{catalytic logspace} as
    $$\CL \coloneqq \bigcup_{k \in \mathbb{N}} \CSPACE{k \log n}{n^k}$$
    We will also use the suffix $\P$ to impose an additional constraint to run in polynomial time, i.e.,
    $$\CL\PT \coloneqq \bigcup_{k \in \mathbb{N}} \CTIMESPACE{n^k}{k \log n}{n^k}$$
    where $\CTIMESPACE{t}{s}{c}$ is defined as $\CSPACE{s}{c}$ machines which always halt in time $t$.
\end{definition}

\noindent
An important relaxation of the base catalytic model is allowing additional access to \emph{randomness}.

\begin{definition}[Randomized catalytic computation]
    We define $\BPCSPACE{s}{c}$ to be the class of languages $L$ for which there exists a randomized catalytic machine $M$ with $s$ bits of work space, $c$ bits of catalytic space, and access to $2^c$ random bits in a read-only read-once fashion, such that $M(x) = L(x)$ with probability $\geq \frac{2}{3}$ over the choice of the randomness of $M$.
\end{definition}

\noindent
The focus of this paper will be on \emph{lossy} catalytic computing, another relaxation defined by Gupta et al.~\cite{GuptaJainSharmaTewari24} where we may make errors on the catalytic tape.

\begin{definition}[Lossy catalytic computation]
    A \emph{lossy catalytic Turing machine} with $e$ errors is a Turing machine $M$ with the same four tapes as a catalytic Turing machine, but where on any input $x$ and initial catalytic tape $\tau$, the machine must halt with some $\tau'$ in the catalytic tape such that $\tau$ and $\tau'$ differ in at most $e$ locations.

    Let $n \in \mathbb{N}$ and let $s \coloneqq s(n)$, $c \coloneqq c(n)$, $e \coloneqq e(n)$. We say a language is in $\LCSPACE{s}{c}{e}$ if it can be decided by a lossy catalytic Turing machine with free space $s$, catalytic space $c$, and making at most $e$ errors on inputs of length $n$.
\end{definition}

\noindent
We also define $\BPCL$, $\BPCL\PT$, $\LCL{e}$, and $\LCL{e}\PT$ analogously to $\CL$, as well as $\BPLCSPACE{s}{c}{e}$ and $\BPLCL{e}$ as the intersection of both restrictions. In this context, we have a number of collapses to $\CL$ and $\CL\PT$ by connecting results of Cook et al.~\cite{CookLiMertzPyne25} and Kouck\'y et al.~\cite{KouckyMertzPyneSami25} for $\BPCL$ with those of Folkertsma et al.~\cite{FolkertsmaMertzSpeelmanTupker26} for $\LCL{e}$ (reflected in \autoref{fig:definitions}).

\begin{lemma}[\cite{CookLiMertzPyne25,KouckyMertzPyneSami25,FolkertsmaMertzSpeelmanTupker26}]\label{thm:old_collapses}
    $$\CL = \BPCL = \BPCL\PT = \LCL{O(1)} = \BPLCL{O(1)} = \BPLCL{O(1)}\PT$$
    $$\CL\PT = \LCL{O(1)}\PT$$
\end{lemma}

In the case of lossy catalytic computation, we will use a slightly more flexible fact due to Folkertsma et al.~\cite{FolkertsmaMertzSpeelmanTupker26}, which in fact implies the corresponding results in \autoref{thm:old_collapses}:

\begin{lemma}[\cite{FolkertsmaMertzSpeelmanTupker26}]\label{thm:lcl_vs_cl}
    If $e \leq c^{1-\Omega(1)}$, then
    $$\LCSPACE{\Theta(s)}{\Theta(c)}{e} = \CSPACE{\Theta(s +e\log{c})}{\Theta(c)}$$
    and analogous results hold for $\BPCSPACE{s}{c}$, $\CSPACE{s}{c}\PT$, and $\BPCSPACE{s}{c}\PT$.
\end{lemma}

\noindent
Our paper focuses on the case of two-sided randomness, i.e., $\BPCL$, but we also define catalytic computing with \emph{zero-error randomness}, including a non-standard extension to the \emph{two-way randomness access} setting. The latter definition has implicitly existed in the literature but has seen no study thus far; however, while these classes are not on their face relevant to robust catalytic computing, they find their way into our characterizations.
\begin{definition}\label{def:zpstarcspace}
    We define $\ZPCSPACE{s}{c}$ to be the same as $\BPCSPACE{s}{c}$, but where the machine outputs $L(x)$ with probability $\ge \frac{1}{2}$ and otherwise outputs a special symbol $\perp$; in particular, it never outputs $\overline{L(x)}$.
    
     We define $\ZPstarCSPACE{s}{c}$ to be the same as $\ZPCSPACE{s}{c}$, but with \emph{two-way} access to its randomness. We restrict the amount of randomness to $2^c$ bits and require that our machine always halts.
\end{definition}

\subsection{Configuration graphs}

As in most models of space-bounded computation, we view catalytic space through the syntactic characterization of the \emph{configuration graph} of a machine $M$.
\begin{definition}
    Let $M$ be any catalytic machine with work space $s$ and catalytic space $c$.
    We denote by $\conf{\pi}{u}$ the configuration of $M$ where $\pi \in \{0,1\}^c$ is contained on the catalytic tape and $u \in \{0,1\}^s$ is on its work tape.     
\end{definition}

Here, we assumed without loss of generality that all auxiliary information
about the current configuration of $M$, i.e., the state of $M$'s internal DFA,
and the current positions of the tape heads for the input, work, and catalytic tapes
are all automatically recorded in a designated part of the work tape.\footnote{Altogether, this additional information
technically requires additional space $\log n + \log s + \log c + O(1)\le 3s$. 
 We can handle this by replacing $s$ with $4s$ throughout the proofs,
which we omit for clarity.} 

Consider the execution of $M$ on some fixed input $x$ and initial catalytic tape
contents $\tau$. Each configuration of $M$ can be uniquely represented by $\conf{\pi}{u}$
for some $\pi \in \{0,1\}^c$ and $u \in \{0,1\}^s$. Without loss of generality, we define
$\start := \conf{\tau}{0^s}$ to be the start configuration, $\acc := \conf{\tau}{1 \circ 1 \circ 1 \circ 0^{s-3}}$  
to be the unique accepting halt configuration, and $\rej := \conf{\tau}{1 \circ 1 \circ 0 \circ 0^{s-3}}$ to be the unique rejecting halt configuration. In the case that $M$ is a zero-error machine (see \autoref{def:zpstarcspace}), we also define the unique \emphdef{don't-know} halt configuration as $\dontknow \coloneqq \conf{\tau}{1 \circ 0 \circ 0^{s-2}}$. Thus, starting and halting configurations can be easily recognized. It will often be useful to discuss the configuration graph defined by such executions.

\begin{definition}[Configuration graphs]\label{definition_config_graphs}
    The configuration graph \( \gmx \) is a directed graph where each node corresponds to a configuration of the machine \( M \) on input \( x \). There is a directed edge from \( \conf{\pi}{u} \) to \( \conf{\pi'}{u'} \) if and only if \( \conf{\pi'}{u'} \) can be reached from \( \conf{\pi}{u} \) in one execution step of \( M \). If \( M \) is a probabilistic machine, the out-degree of every vertex in \( \gmx \) is at most 2; if \( M \) is deterministic, the out-degree is at most 1. A halting configuration has no outgoing edges in \( \gmx \). Furthermore, there exists a fixed constant $d_M$, which depends solely on $M$, such that each vertex in $\gmx$ has at most $d_M$ incident edges.
\end{definition}

An edge is marked with \( b \in \{0, 1\} \) if it corresponds to a randomized \( b \)-choice of the machine. For deterministic transitions, we label the edge with both 0 and 1.

For every catalytic tape \( \tau \), let \( \gmxtau \) represent the subgraph of \( \gmx \) that is induced by the configurations reachable from the starting configuration \( \start \). Because the machines we consider always halt, \( \gmxtau \) is acyclic and has a single source node, \( \start \), along with as many as two sink nodes, \( \acc \) and \( \rej \). If \( M \) is a zero-error machine, there can be a third sink node, \( \dontknow \). Observe that \( \gmxtau \) forms a path when \( M \) is a deterministic machine.

\subsubsection{Traversal of configuration graphs}\label{sc:grtrv}
Let \( M \) be a deterministic catalytic machine, and \( v \in V(\gmxtau) \), for some $\tau$. Then, the connected (in the undirected sense) component  of \( \gmx \) that contains \( v \) is a tree, which we will denote by $\gmx(v)$. 
This tree has a unique sink: the halting configuration in which \( M \) halts when run from $v$ (or from $\start$). 
This tree can be traversed, listing vertices in a pre-order fashion, where we think of the traversal as cyclic (upon completing a traversal, we start a new one). 
Then, the following procedures give us a way to traverse $\gmx(v)$:

\begin{lemma}\label{reversiblewalk}
Let \( M \) be a deterministic catalytic machine that utilizes \( s \) bits of work space and \( c \) bits of catalytic space. We have the following catalytic subroutines, which use an extra \( O(s) \) bit of work space.
\begin{itemize}
\item $\textbf{\nextstep}$. Given $\conf{\pi}{u}$ on catalytic tape and work tape, this subroutine replaces $\conf{\pi}{u}$ with $\conf{\pi'}{u'}$, which is the next vertex in  (cyclic) pre-order traversal of $\gmx(\conf{\pi}{u})$.
\item $\textbf{\stepback}$. Given $\conf{\pi}{u}$ on catalytic tape and work tape, this subroutine replaces $\conf{\pi}{u}$ with $\conf{\pi'}{u'}$, which is the previous vertex in pre-order traversal of $\gmx(\conf{\pi}{u})$.
\item $\textbf{\countsize}$. Given $\conf{\pi}{u}$ on catalytic tape and work tape, and an integer $S \le 2^{O(s)}$, this subroutine checks whether $\gmx(\conf{\pi}{u})$ is of size at most $S$. It returns the content of $\conf{\pi}{u}$ unchanged.
\end{itemize}
The procedure $\textbf{\countsize}$ runs in time $2^{O(s)}$, and the procedures $\textbf{\nextstep}$ and $\textbf{\stepback}$ run in time $2^{O(s)} \cdot D$ where $D$ is the depth of $\gmx(\conf{\pi}{u})$.
\end{lemma}

We will usually use  $\textbf{\nextstep}$ and $\textbf{\stepback}$  only when   $\gmx(\conf{\pi}{u})$ is of depth $2^{O(s)}$. In that case, both procedures take time $2^{O(s)}$.

Our subroutines can be easily obtained from those in \cite{KouckyMertzPyneSami25}. 
Similar procedures in \cite{KouckyMertzPyneSami25} take an additional parameter $h \in \{1,\dots,d_M\}$ which is a label of an incident edge.
For both of our procedures $\textbf{\nextstep}$ and $\textbf{\stepback}$, we can run their procedures repeatedly, starting with an edge label 1, and wait until we get into another vertex with the edge label 1.
When 1 corresponds to the incoming edge in the directed version of  $\gmx(\conf{\pi}{u})$, the process will list configurations in a pre-order traversal, 
and the number of steps we have to take to reach the next vertex in the pre-order fashion is at most the depth of the graph.

The $\textbf{\countsize}$ procedure runs the Eulerian tour traversal \cite{KouckyMertzPyneSami25} of $\gmx(\conf{\pi}{u})$ starting from $\conf{\pi}{u}$ for $6S$ steps. 
As $\gmx(\conf{\pi}{u})$ always contains an easy to recognize unique final configuration of $M$ if the tour does not visit it at least twice with an edge label 1 we know $\conf{\pi}{u}$ is larger that $S$.
Otherwise, we count the number of edges labeled 1 that we traverse between successive visits to the final configuration.
This number corresponds to the size of  $\gmx(\conf{\pi}{u})$, so we compare it with $S$ and decide on the answer. 
Using the reverse traversal of the Eulerian tour for $6S$ steps, we return to the initial $\conf{\pi}{u}$.

Since the tour given by $\nextstep$/$\stepback$ follows a depth-first search (DFS) traversal, we will commonly refer to it as such. 

We will need all of the above procedures to apply to graphs that are restrictions of configuration graphs of {\em randomized} machines.
In those cases, a sequence $y$ of bits will be specified, which determines the random choices of a machine $M$.
We will make a leveled version of the graph (level $i$ leading into level $i-1$). In level $i$, we will use the $i$-th bit of $y$ from the end to resolve the random choice of $M$ at that time step.
The leveled graph will be restricted only to have levels $|y|,\dots, 0$, and a node in this graph will be specified by a configuration of $\gmx$ and a level index.
If the bits of $y$ can be accessed via some efficiently computable procedure, then  
$\textbf{\nextstep}$, $\textbf{\stepback}$, and $\textbf{\countsize}$ on such a graph can be guaranteed to behave as described above,
where the space and time used need to account also for the space and time needed to calculate individual bits of $y$ (keeping track of the current level is easy based on the direction of the traversed edge in the directed graph $\gmx$).
Typically, the space needed to calculate individual bits of $y$ will be $O(s)$ and time $2^{O(s)}$. 
So, it will not affect the asymptotics of our space or running time.
For a binary sequence $y$, we will refer to the traversal of such a graph by the three procedures as $y$-DFS traversal.

\section{New Robust Catalytic Classes}
\label{sec:new}

Our chief object of study in this paper is a number of expansions of the basic $\LCSPACE{s}{c}{e}$ definition of lossy catalytic space. Thus, we begin with all definitions of interest. Every machine in this section will use $s$ bits of work space, $c$ bits of catalytic space, and any randomized machines have access to a read-only read-once random string $r \in \{0,1\}^{2^c}$ (recall our conditions on $s$, $c$, and $e$ from \autoref{sec:prelims}).

In our first model, a randomized catalytic machine may arbitrarily destroy its catalytic tape with some probability over the randomness, but must perfectly reset otherwise:
\begin{definition}
    Let $0 \leq \delta < 1$ and $0 < \epsilon \le \frac{1}{2}$. A language $L$ is in \( \BPepsCdeltSpace{s}{c}{\delta}{\epsilon} \) if there is a randomized machine $M$ such that on any input \( x \in \{0,1\}^n \) and initial setting of the catalytic tape \( \tau \in \{0,1\}^c \), $M$ always halts, outputs $L(x)$ with probability at least $\frac{1}{2} + \epsilon$, and resets the catalytic tape to \( \tau \) with probability at least $\ge  1 - \delta $, where both probabilities are over the machine's randomness $r$.
    
    We further use the notation \( \BPepsCdeltSpace{s}{c}{\delta}{\epsilon}_{\high} \) to denote $\delta > 2 \epsilon$ and \( \BPepsCdeltSpace{s}{c}{\delta}{\epsilon}_{\low} \) to denote $\delta < 2 \epsilon$ $(< 1)$.
\end{definition}

\noindent
Our second and third models have to do with a bounded expected number of errors on the catalytic tape. The first version does so in expectation over the randomness of the machine:
\begin{definition}
    A language $L$ is in $\AVGrBPLCSPACE{s}{c}{e}$ if there is a randomized machine $M$ such that on any input \( x \in \{0,1\}^n \) and initial setting of the catalytic tape \( \tau \in \{0,1\}^c \), $M$ always halts, outputs $L(x)$ with probability at least $\frac{2}{3}$, and resets the catalytic tape to $\tau'$ such that $\mathbb{E}_{r}[d(\tau, \tau')] \leq e$, i.e., we make at most $e$ errors in expectation over the machine's randomness $r$.
\end{definition}

\noindent
The last model again makes a bounded number of errors in expectation, but now the expectation is over the initial catalytic tape, with the machine acting deterministically:
\begin{definition}
    A language $L$ is in $\AVGtauLCSPACE{s}{c}{e}$ if there is a deterministic machine $M$ such that $M$ outputs $L(x)$ for every input \( x \in \{0,1\}^n \) and makes at most $e$ errors in expectation over the initial catalytic tape $\tau$, i.e., $\mathbb{E}_{\tau}[d(\tau, \tau')] \leq e$ where \( \tau, \tau' \in \{0,1\}^c \) are the initial and corresponding final contents of the catalytic tape, respectively.
\end{definition}

\noindent
In this paper, we focus on the logspace variant of all the above classes, both with and without the polynomial-time restriction, and we use notation in accordance with all previous classes.

A note on the remainder of the paper: while some results will follow from short proofs and/or minor tweaks to ones that appear earlier, for the sake of highlighting where collapses occur, we will use Theorem or Corollary exactly for those results that relate complexity classes together.

\section{Model 1: $\BPepsCdeltSpace{s}{c}{\delta}{\epsilon}$}
\label{sec:deltaepsilon}

We first address the case of $\BPepsCdeltSpace{s}{c}{\delta}{\epsilon}$, which we stratify into the \emph{high error case} when $\delta > 2\epsilon$ (i.e. $\BPepsCdeltL{\delta}{\epsilon}_{\high}$), the \emph{low error case} when $\delta < 2 \epsilon$ (i.e. $\BPepsCdeltL{\delta}{\epsilon}_{\low}$), and the \emph{zero error case} when $\epsilon = \frac{1}{2}$ (i.e. $\BPepsCdeltL{\delta}{\frac{1}{2}}$).

\subsection{High error case: $\BPepsCdeltL{\delta}{\epsilon}_{\high}$}

In the high error regime, our chance of destroying the catalytic tape is greater than our advantage of finding the correct answer over random chance. We can exploit this fact to randomly choose whether to erase the tape and brute force the function or simply output a random answer.

\begin{theorem}\label{trivialpspace}
Let $\delta, \epsilon$ be constants such that $\delta > 2\epsilon$. Then
$$\BPepsCdeltSpace{O(\log{S})}{\poly(S)}{\delta}{\epsilon} = \SPACE(\poly(S))$$ where $S\coloneqq S(n) \ge n$.
\end{theorem}
\begin{proof}
For the forward inclusion, $\BPepsCdeltSpace{s}{c}{\delta}{\epsilon} \subseteq \BP\SPACE(s+c) \subseteq \SPACE((s+c)^{3/2})$, where the first inclusion follows from using a ``normal'' read-write tape instead of the catalytic tape, and the second from the derandomization of bounded space due to Saks and Zhou~\cite{SaksZhou99}.

Now let $L \in \SPACE(S)$ and let $N$ be the space-$S$ machine deciding $L$. As dyadic rationals are dense in the reals \cite{263128}, there exists a rational number $\alpha$ such that $2\epsilon < \alpha < \delta$, and it has the form $\frac{p}{2^q}$ for $p,q \in \mathbb{N}$; since $\delta$ and $\epsilon$ are constants, $\alpha$ can be described with a constantly many bits. Our $\BPepsCdeltSpace{\log S}{S}{\delta}{\epsilon}$ machine $M$ will first destroy the tape with probability $\alpha$, and if it does, then we run $N$ in the catalytic memory and output the answer; if not, we flip a fair coin and output the answer.

$M$ does not reset the catalytic tape with a probability of at most $\alpha < \delta$, while the probability that it correctly decides the input is 
$$
\frac{1}{2}(1-\alpha) + \alpha = \frac{1}{2} + \frac{\alpha}{2} \ge \frac{1}{2} + \epsilon
$$
as required.
\end{proof}

\noindent
Our main consequence will be for catalytic logspace, where we exactly capture $\PSPACE$:

\begin{corollary}\label{trivialpspace2}
Let $\delta, \epsilon$ be constants such that $\delta > 2\epsilon$. Then
$$\BPepsCdeltL{\delta}{\epsilon} = \PSPACE$$
\end{corollary}

\noindent
The same logic applies when we have a polynomial time bound, with the only change being the power of the machine when we destroy the catalytic tape:

\begin{theorem}\label{trivialp}
Let $\delta, \epsilon$ be constants such that $\delta > 2\epsilon$. Then
$$\BPepsCdeltL{\delta}{\epsilon}\PT = \BPP$$
\end{theorem}
\begin{proof}
For the forward inclusion, $\BPepsCdeltL{\delta}{\epsilon}\PT \subseteq \BP\TIME(\poly(n))$ since the latter uses a ``normal'' read-write tape instead of the catalytic tape. Now let $L \in \BP\TIME(T)$ and let $N$ be the randomized time-$T$ machine deciding $L$. Again choose $\alpha$ such that $2\epsilon < \alpha < \delta$ of the form $\frac{p}{2^q}$ for $p,q \in \mathbb{N}$. Our $\BPepsCdeltL{\delta}{\epsilon}\PT$ machine $M$ uses a catalytic memory of size $T$ and is identical to \autoref{trivialpspace}, with the note that we can run $N$ in our catalytic memory in time $T = \poly(n)$. Here, we assumed without loss of generality that $N$ outputs correctly with probability $\ge 1-\eta$ such that $\alpha(1/2 - \eta) > \epsilon$ (this is possible since $\alpha/2 > \epsilon$).
\end{proof}

\subsection{Low error case: $\BPepsCdeltL{\delta}{\epsilon}_{\low}$}
\label{sc:lowerr}

For the case of $\BPepsCdeltL{\delta}{\epsilon}_{\low}$, our main result will be the following:

\begin{theorem}\label{corroepsdeltagoodpiscl}
Let $\epsilon,\delta$ be constants such that $\delta < 2\epsilon <1$. Then
$$\BPepsCdeltL{\delta}{\epsilon}_{\low}\PT = \CL$$
\end{theorem}

\noindent
We will not discuss the general time-unbounded $\BPepsCdeltSpace{s}{c}{\delta}{\epsilon}_{\low}$, but will return to this question in great detail in \autoref{sec:avgr}.

To prove \autoref{corroepsdeltagoodpiscl}, we will need to analyze the configuration graph of a $\BPepsCdeltL{\delta}{\epsilon}_{\low}\PT$ graph.
Buhrman et al. \cite{BuhrmanCleveKouckyLoffSpeelman14} demonstrated that for a $\CSPACE{s}{c}$ machine, the average size of the configuration graph over the initial catalytic tape is $2^{O(s)}$, but this relies on the fact that such a machine always resets the catalytic tape.
In the case of a $\BPepsCdeltSpace{s}{c}{\delta}{\epsilon}$ machine $M$, the relevant concept to consider is that of a $\taubeta$-graph.

\begin{definition}[$\taubeta$-graph]\label{taunodedefn}
Consider a $\BPepsCdeltSpace{s}{c}{\delta}{\epsilon}$ machine $M$ and its configuration graph $\gmxtau$, for input $x$ and initial catalytic tape $\tau$.  We define a configuration $v$ in this graph to be a $\taubeta$-node if the probability of reaching $\acc$ or $\rej$ from $v$ is at least $\beta$. The $\taubeta$-graph is then defined as follows:

\begin{itemize}
    \item If $\start$ is not a $\taubeta$-node, we define the $\taubeta$-graph to be empty.  Otherwise, the vertex set consists of all the $\taubeta$-nodes in $\gmxtau$ that can be reached from $\start$ via \textit{only} $\taubeta$-nodes.
    \item The transitions/edges are the same as $\gmxtau$, except that all transitions going to a non-$\taubeta$-node are routed to a (null) $\perp$ node.
\end{itemize}
\end{definition}

Let $\delta, \epsilon$ be constants such that $\delta < 2\epsilon < 1$, and
consider a suitable constant $\beta$, $0 < \beta < 1-\delta$.
Note that as $M$ resets the catalytic tape correctly with probability $\ge 1-\delta$, by  \autoref{taunodedefn} and $\beta < 1-\delta$, we have that $\start$ is in the $\taubeta$-graph. Since $M$ always halts by definition, and given that $\beta > 0$, it follows from a simple averaging argument that at least one of $\acc$ or $\rej$ is in the vertex set of the $\taubeta$-graph. Moreover, it follows from  \autoref{taunodedefn} and $\beta > 0$ that the $\taubeta$-graph does not contain $\acc[\tau']$ or $\rej[\tau']$ for $\tau \not = \tau'$.

\begin{lemma}\label{disjointobsv}
Let $v$ be a configuration in $\gmx$, then there are at most $\frac{1}{\beta}$ different values for the initial catalytic setting $\tau$ such that $v$ is in the $\taubeta$-graph.
\end{lemma}
\begin{proof}
If $v$ is in the $\taubeta$-graph for some $\tau$, then by \autoref{taunodedefn} it is a $\taubeta$-node and, with a probability of $\ge \beta$, reaches a halt state with catalytic contents $\tau$. Since we assume $\beta>0$, the lemma follows.
\end{proof}

\begin{lemma}\label{avgsize}
Let the size of a $\taubeta$-graph be defined as the number of vertices in it (excluding $\perp$). Then the average (over $\tau$) size of a $\taubeta$-graph is at most $2^{4s}$.
\end{lemma}
\begin{proof}
The number of bits required to describe any configuration is:
$$
s+c + (\log{n} + \log{c} + \log{s} + O(1)) \le 3s +c + 2\log{s}
$$ 
which includes the bits required to store the internal state and the tape head locations. Since by \autoref{disjointobsv}, any configuration can be a part of a $\taubeta$-graph for at most $\frac{1}{\beta}$ different values of $\tau$, the average size is at most $\frac{1}{\beta}\cdot \frac{s^2\cdot 2^{3s+c}}{2^c} \le \frac{s^2\cdot 2^{3s}}{\beta} \le 2^{4s}$ (for large enough $s$); where we used the fact that $\beta$ is a constant.
\end{proof}

\begin{lemma}\label{tauepsp}
The probability $q$ of reaching $\acc$ or $\rej$ from $\start$ within the $\taubeta$-graph is at least $1-\frac{\delta}{(1-\beta)}$. In other words, the probability of reaching $\perp$ from $\start$ in the $\taubeta$-graph is at most $\frac{\delta}{(1-\beta)}$.
\end{lemma}
\begin{proof}
Firstly, note that by our assumption $0 \le \delta < 1-\beta < 1$, and thus the probabilities make sense. We have that
\begin{align*}
q&=1-\Pr[\text{Going from } \start \text{ to} \perp \text{ in $\taubeta$-graph}]\\
&= 1- \Pr[\text{Reaching a non-$\taubeta$-node from } \start \text{ in } G_{M,x,\tau} ]
\end{align*}
By the definition of a non-$\taubeta$-node, we get
\begin{align*}
\big(1-\beta\big)\cdot& \Pr[\text{Reaching a non-$\taubeta$-node from } \start  \text{ in } G_{M,x,\tau}] \\
& \le \Pr[\text{Not resetting catalytic tape to $\tau$ } \text{from } \start \text{ in }G_{M,x,\tau} ] \le \delta \qedhere
\end{align*}
\end{proof}

If we take $\beta$ to be a constant value such that $0< \beta < 1-\frac{\delta}{2\epsilon} \le 1-\delta$ (such a value exists as $\delta < 2\epsilon \le 1$), then the aforementioned observations still hold. We set $\beta =\frac{1}{2} \big (1-\frac{\delta}{2\epsilon} \big)$. Then, by  \autoref{tauepsp}, the probability of reaching $\perp$ from $\start$ (in the $\taubeta$-graph) is at most $\gamma = \frac{2\delta}{1+\frac{\delta}{2\epsilon}} < 2\epsilon$. In other words, the probability of exiting the $\taubeta$-graph via $\perp$ is at most $\gamma$. When \( x \) is not in the language recognized by \( M \), by definition, \( M \) accepts it with a probability of at most \(\frac{1}{2} - \epsilon\). On the other hand, consider the case when \( x \) is in the language. Even if we restrict ourselves to the \(\taubeta\)-graph, we know that the probability of going from \(\start\) to \(\acc\) is at least \(\frac{1}{2} + \epsilon - \gamma\), which is a constant that is bounded away from \(\frac{1}{2} - \epsilon\).

In summary, since the \( \taubeta \)-graph has no halt configurations other than \( \acc \)  and \( \rej \), the machine \( M \) behaves like a purely catalytic machine under this restriction. Moreover, even with this restriction, we can still decide the input. Although we do not yet know how to handle the case where the \( \taubeta \)-graph is ``large''—which would help us de-randomize the class \( \BPepsCdeltSpace{s}{c}{\delta}{\epsilon} \)—we do know how to deal with the case when the graph is small. This is often the case, as indicated in \autoref{avgsize}.

When the $\taubeta$-graph is small, we can either decide the input or compress the catalytic contents. This serves as the basis for the primary lemma in this work, and we defer the proof to \autoref{nisansprgproofs}:

\begin{lemmaletter}{a}\label{lemmvianisansprggeneral} Let $\epsilon,\delta$ be constants such that $\delta < 2\epsilon$. Then, for every $L \in \BPepsCdeltSpace{s}{c}{\delta}{\epsilon}$ there exist \textbf{deterministic} catalytic subroutines $\mathcal{F}^L$ and $\mathcal{R}^L$, which given input $x$ behave as follows:
\begin{enumerate}
    \item $\mathcal{F}^L$ and $\mathcal{R}^L$ use $O(s)$ work space, $2^{O(s)}$ catalytic space and run in time $2^{O(s)}$.
    \item For initial catalytic contents $\tau$, $\mathcal{F}^L(x,\tau)$ either outputs (don't-know) $\perp$ or correctly decides if $x$ is in $L$, or changes the catalytic tape to $\tau'\circ 0$. In the first two cases, it does not change the contents of the catalytic tape.
    \item For $\tau'$ from the previous step, $\mathcal{R}^L(x,\tau'\circ 0)$ resets the catalytic tape back to $\tau$.
    \item The fraction of initial catalytic contents $\tau$, for which $\mathcal{F}^L$ outputs $\perp$ is  at most $\frac{1}{2^{4s}}$.
\end{enumerate}
\end{lemmaletter}

\noindent
For $\delta = 0$, $M$ always resets the catalytic tape. In this case, the proof of the lemma above provides a de-randomization of $\BPCSPACE{s}{c}$, which is already known due to Cook et al. \cite{CookLiMertzPyne25}; in this case, $\mathcal{F}^L$ never outputs $\perp$. Under the polynomial time restriction as well, $\mathcal{F}^L$ never outputs $\perp$ (we again defer the proof to \autoref{nisansprgproofs}):

\addtocounter{lemma}{-1}

\begin{lemmaletter}{b}\label{lemmvianisansprgforcl}
Let $\epsilon,\delta$ be constants such that $\delta < 2\epsilon$. Then, for every $L \in \BPepsCdeltL{\delta}{\epsilon}\PT$, there exist \textbf{deterministic} catalytic subroutines $\mathcal{F}^L$ and $\mathcal{R}^L$ that behave exactly the same as those in \autoref{lemmvianisansprggeneral}; \textit{except} that now $\mathcal{F}^L$ never outputs $\perp$.
\end{lemmaletter}

\noindent
Note that \autoref{lemmvianisansprggeneral} and \autoref{lemmvianisansprgforcl} hold trivially when $\delta=0$ and $\epsilon = \frac{1}{2}$, as by definition $\BPepsCdeltSpace{s}{c}{0}{\frac{1}{2}} = \CSPACE{s}{c}$ and $\BPepsCdeltL{0}{\frac{1}{2}}\PT = \CL\PT$.

We now have all the lemmas needed to prove our main result for $\BPepsCdeltL{\delta}{\epsilon}_{\low}$.

\begin{proof}[Proof of \autoref{corroepsdeltagoodpiscl}]
For the reverse direction, it is known that $\CL = \ZPCL\PT$ \cite{CookLiMertzPyne25}. Given a $\ZPCL\PT$ machine, we can run it $\big\lceil \log{\frac{1}{\frac{1}{2}-\epsilon}} \big\rceil$ times, and we output any non-$\perp$ answer we see; if none exists, we output 0. Clearly, we always reset the catalytic tape and answer incorrectly with a probability of at most $\frac{1}{2}-\epsilon$.

Now we prove the forward direction. Let $L \in \BPepsCdeltL{\delta}{\epsilon}\PT$, and let the corresponding machine be $N$ which runs in time $n^d$ (and hence uses catalytic space at most $n^d$) for $d \in \mathbb{N}$. Given input $x$, we run subroutine $\mathcal{F}^L$ from \autoref{lemmvianisansprgforcl} on $n^{10d}$ different chunks of polynomial-sized catalytic tapes. If $\mathcal{F}^L$ correctly decides $x$ for any chunk we output the result; otherwise, $\mathcal{F}^L$ frees up the last bit of each chunk, and hence $n^{10d}$ bits, on the catalytic tape. In this case, we do a brute force simulation of $N$ in the freed-up space, using the fact that $N \in \BP\TIME(s+c) \subseteq \SPACE(O(s+c))$ by iterating over all random strings. Finally, before outputting the result, we run $\mathcal{R}^L$ from \autoref{lemmvianisansprgforcl} to reset all the catalytic tape chunks. Our space bound follows because $\mathcal{F}^L$, $\mathcal{R}^L$, and $N$ each use $O(\log{n})$ work space.
%
%
\end{proof}


\subsection{Zero error case: $\BPepsCdeltL{\delta}{\frac{1}{2}}$}

A $\BPepsCdeltL{\frac{1}{poly(n)}}{\frac{1}{2}}$ machine always outputs the correct answer but has a non-negligible probability of destroying the catalytic tape; specifically, this probability is an inverse polynomial in the input length (for a sufficiently large polynomial). Currently, we do not know how to even show unconditionally that $\CL = \BPepsCdeltL{\frac{1}{poly(n)}}{\frac{1}{2}}$. However, if the probability of not resetting the tape is exponentially small (for a sufficiently large exponential), then the class collapses to $\CL$.

\begin{observation}\label{exponetiallysmall}
For any constant $0 < \epsilon \le \frac{1}{2}$,
$$\cup_{d \in \mathbb{N}} \BPepsCdeltSpace{d\log{n}}{n^d}{\frac{1}{2^{4n^d+3}}}{\epsilon}  = \CL$$
\end{observation}
\begin{proof}
The reverse direction follows by definition, while we defer the forward direction to \autoref{appexposmall}.
\end{proof}

Although, for the case of $\BPepsCdeltL{\delta}{\frac{1}{2}}$, we cannot say anything in the time-unrestricted case (it falls between $\CL$ and $\BPepsCdeltL{\delta}{\epsilon}_{\low}$ by definition), in the polynomial-time world, we can show a loss of power even beyond $\BPepsCdeltL{\delta}{\epsilon}\PT_{\low}$:

\begin{theorem}\label{epsilonhalfcorrolary}
Let $\delta < 1$ be a constant. Then
$$\BPepsCdeltL{\delta}{\frac{1}{2}}\PT=\CL\PT$$
\end{theorem}
\begin{proof}
Clearly $\BPepsCdeltL{\delta}{\frac{1}{2}}\PT \supseteq \CL\PT$ by definition. For the forward direction, we closely follow the proof of \autoref{corroepsdeltagoodpiscl}, except for the case when $\mathcal{F}^L$ frees up the last bit of each chunk. In this case, we simulate the $\BPepsCdeltL{\delta}{\frac{1}{2}}\PT$ machine using the newly available space on the all-zero random string. This works because the $\BPepsCdeltL{\delta}{\frac{1}{2}}\PT$ machine runs in polynomial time and is always correct.
\end{proof}

\subsection{Afterword: tradeoffs between $\epsilon$ and $\delta$}

It may seem strange that our choice of $\epsilon$ and $\delta$ is largely irrelevant as long as we understand where they fit in the above trichotomy. To close, we note that for the \emph{one-sided error} regime, we can precisely scale both $\epsilon$ and $\delta$ in tandem.

\begin{lemma}\label{lem:tradeoff_eps-delta}
    Define $\RepsCdeltSpace{s}{c}{\delta}{\epsilon}$ to be the class of languages $L$ for which there exists a randomized catalytic machine $M$ with $s$ bits of work space, $c$ bits of catalytic space, and access to a string $r$ of $2^c$ random bits in a read-only read-once fashion, such that 1) if $L(x) = 1$ then $M(x)$ outputs 1; 2) if $L(x) = 0$ then $M(x)$ outputs 0 with probability at least $\epsilon$ over $r$; 3) $M$ resets the catalytic tape with probability at least $1-\delta$ over $r$.

    Then for all $\epsilon, \delta$,
    $$\RepsCdeltSpace{s}{c}{\delta}{\epsilon} \subseteq \RepsCdeltSpace{s}{c}{\delta/2}{\epsilon/2}$$
    $$\RepsCdeltSpace{s}{c}{\delta}{\epsilon} \subseteq \RepsCdeltSpace{s}{c}{2\delta - \delta^2}{2\epsilon-\epsilon^2}$$
\end{lemma}
\begin{proof}
    Let $M$ be a $\RepsCdeltSpace{s}{c}{\delta}{\epsilon}$ machine. First, our $\RepsCdeltSpace{s}{c}{\delta/2}{\epsilon/2}$ machine $M_{\frac{1}{2}}$ flips a fair coin, and if it returns heads, we output 1; otherwise, we run $M$. Clearly, 1) remains fulfilled since we never output 0 when $M$ outputs 1. Since with probability $\frac{1}{2}$, we output 1 and do not destroy the tape, for 2) we have $\Pr[M_{\frac{1}{2}}(x) = 0] = (\frac{1}{2}) \cdot \epsilon$, while for 3) our probability of destroying the tape is at most $(\frac{1}{2}) \cdot \delta$.

    Second, our $\RepsCdeltSpace{s}{c}{2\delta - \delta^2}{2\epsilon - \epsilon^2}$ machine $M_2$ runs $M$ twice and outputs 1 if both runs output 1; otherwise, we output 0. Again, 1) remains fulfilled by definition. For 2), given $x$ such that $L(x) = 0$, we have
    $$\Pr[M_2(x) = 0] = 2\epsilon-\epsilon^2$$
    while for 3) we have
    $$\Pr[M_2(x) \mbox{ destroys } \tau] \le \delta + (1-\delta)\delta = 2\delta - \delta^2$$
    which completes the lemma.
\end{proof}

While the role of $\epsilon$ is different in $\BPepsCdeltSpace{s}{c}{\delta}{\epsilon}$, and hence we do not have the same tradeoff, intuitively \autoref{lem:tradeoff_eps-delta} shows that we have flexibility in choosing $\delta$ and $\epsilon$, but we cannot cross the boundary between $\delta < \epsilon$ and $\delta > \epsilon$ (which is analogous to the boundary between $\delta < 2\epsilon$ and $\delta > 2\epsilon$ for $\BPepsCdeltSpace{s}{c}{\delta}{\epsilon}$).
\section{Model 2: $\AVGrBPLCSPACE{s}{c}{e}$}
\label{sec:avgr}

In this section, we analyze the class $\AVGrBPLCSPACE{s}{c}{e}$ and ultimately show a loose relationship to the class $\BPepsCdeltSpace{s}{c}{\delta}{\epsilon}$. 

Recall that an $\AVGrBPLCSPACE{s}{c}{e}$ machine decides the input correctly with a probability of $\ge \frac{2}{3}$ and makes $e$ errors on average on the catalytic tape over the random bits of the machine, and this holds for \textit{every} input and initial catalytic setting. We can trivially boost the factor of $\frac{2}{3}$ to an arbitrary constant $< 1$, at the expense of worsening the work space and the parameter $e$ by a constant factor:

\begin{lemma}\label{boostingavgr}
Let $L \in \AVGrBPLCSPACE{s}{c}{e}$. Then, for every constant $\frac{2}{3} \le q < 1$, there exists an $\AVGrBPLCSPACE{O(s)}{c}{O(e)}$ machine $N$, for $L$, such that for \textit{every} input and initial catalytic tape $N$ correctly decides the input with probability $\ge q$.
\end{lemma}
\begin{proof}
Let $M$ be the $\AVGrBPLCSPACE{s}{c}{e}$ machine for $L$. Our machine $N$ runs $M$ $k = \lceil 300 \log{\frac{1}{(1-q)}} \rceil$  times independently, using the same $s$ bits of free memory and $c$ bits of catalytic tape, and takes the majority of the outputs. By a Chernoff bound, $N$ incorrectly decides the input with a probability of $\le 1-q$, and by linearity of expectation and the triangle inequality, $N$ makes at most $k \cdot e$ errors when resetting the catalytic tape.
\end{proof}

\noindent
We will utilize the following error-correction scheme used by Folkertsma et al.~\cite{FolkertsmaMertzSpeelmanTupker26} in the context of lossy catalytic computation.

\begin{lemma}[Folkertsma et al. \cite{FolkertsmaMertzSpeelmanTupker26}]\label{lem:ecc}
    There exists a compression scheme $(\text{Enc}_{\text{BCH}}, \text{Dec}_{\text{BCH}})$ such that:
    \begin{itemize}
        \item \textbf{Encoding:} \(\text{Enc}_{\text{BCH}}\) takes as input a string \( S \) of length \( c \), plus an additional \((2e + 1) \lceil \log(c + e) \rceil\) bits initialized to zero, and outputs a codeword \( S_{\text{enc}} \):
        \[
        S + [0]^{(2e+1)\lceil \log(c+e) \rceil} 
        \xrightarrow{\text{Enc}_{\text{BCH}}} 
        S_{\text{enc}}.
        \]
        Furthermore, all outputs \( S_{\text{enc}} \) generated this way have minimum distance \(\delta := 2e + 1\) from one another.
        
        \item \textbf{Decoding:} \(\text{Dec}_{\text{BCH}}\) takes as input a string \( S'_{\text{enc}} \) of length 
        \[
        c + (2e+1) \lceil \log(c+e) \rceil,
        \]
        with the promise that there exists a string \( S \) of length \( c \) such that 
        \[
        \text{Enc}_{\text{BCH}}\bigl(S + [0]^{(2e+1)\lceil \log(c+e)\rceil}\bigr)
        \]
        differs from \( S'_{\text{enc}} \) in at most \(\delta/2 - 1 = e\) locations, and outputs this string \( S \):
        \[
        S'_{\text{enc}} \xrightarrow{\text{Dec}_{\text{BCH}}} S + [0]^{(2e+1)\lceil \log(c+e)\rceil}.
        \]
        Furthermore, both \(\text{Enc}_{\text{BCH}}\) and \(\text{Dec}_{\text{BCH}}\) can be computed in space \( O(e \log c) \).
    \end{itemize}
\end{lemma}

\subsection{Relationship to $\BPepsCdeltSpace{s}{c}{\delta}{\epsilon}_{\low}$}

We now establish a connection between $\AVGrBPLCSPACE{s}{c}{e}$ and $\BPepsCdeltSpace{O(s+e\log{c})}{c}{\delta}{\epsilon}_{\low}$. In one direction, we can use any $\delta$ and $\epsilon$ in the low error regime.

\begin{theorem}\label{averageforallinepsdeltaclass}
Let $\delta,\epsilon$ be constants such that $\delta < 2\epsilon < 1$. Then
$$\AVGrBPLCSPACE{s}{c}{e} \subseteq  \BPepsCdeltSpace{O(s+e\log{c})}{c}{\delta}{\epsilon}$$
\end{theorem}
\begin{proof}
Let \( L \in \AVGrBPLCSPACE{s}{c}{e} \). By assumption, \( \frac{1}{2} + \epsilon < 1 \), and thus using \autoref{boostingavgr}, there exists a \( \AVGrBPLCSPACE{O(s)}{c}{O(e)} \) machine \( N \) for \( L \) that, for any input \( x \) and initial catalytic setting \( \tau \), decides the input correctly with probability \( \ge \frac{1}{2} + \epsilon \). Moreover, \( N \) makes at most \( e' = O(e) \) errors on the catalytic tape on average (where the average is over its random bits), for any \( \tau \). Therefore, by Markov's inequality, the probability that \( N \) makes more than \( \frac{e'}{\delta} \) errors is \( \le \delta \).

We will simulate $N$ by reusing the idea from \cite{FolkertsmaMertzSpeelmanTupker26}. Using $c$ bits of catalytic space and $O(\frac{e'}{\delta}\log{c}) = O(e\log{c})$ work space, we run the encoding algorithm $\text{Enc}_{\text{BCH}}$ from \autoref{lem:ecc}. This gives us a codeword $\tau' \circ w$, where $\tau'$ is written on the catalytic tape, and $w$ (which takes $O(\frac{e'}{\delta}\log{c}$ bits) is stored in the work space. We simulate $N$ for input $x$ and catalytic contents $\tau'$ using an additional $O(s)$ work space (and read-once randomness). Before outputting the result, we run the algorithm $\text{Dec}_{\text{BCH}}$ from  \autoref{lem:ecc} on our final catalytic tape $\tau''\circ w$, and then we output the result of the computation.

Since, with probability $\ge 1-\delta$, $N$ makes $\le \frac{e'}{\delta}$ errors; with probability $\ge 1-\delta$, the Hamming distance is $|\tau'' - \tau'| \le \frac{e'}{\delta}$. \autoref{lem:ecc} tells us that we can correct up to $\frac{e'}{\delta}$ errors. Thus, with a probability of $\ge 1-\delta$, running $\text{Dec}_{\text{BCH}}$ on $\tau''\circ w$ resets the tape back to $\tau$. As we use $O(s+e\log{c})$ work space in total and $c$ bits of catalytic space, the theorem follows.
\end{proof}

\begin{corollary}\label{everyavginepsdelcorr}
Let $\delta,\epsilon$ be constants such that $\delta < 2\epsilon < 1$. Then
$$\AVGrBPLCL{O(1)} \subseteq \BPepsCdeltL{\delta}{\epsilon}$$
\end{corollary}

\noindent
For the reverse direction, we can only capture extremely small resetting errors $\delta$, thus leading to a gap in our characterization but still establishing a connection.

\begin{theorem}
For any constant $\epsilon$
$$\BPepsCdeltSpace{s}{c}{\frac{1}{c}}{\epsilon} \subseteq \AVGrBPLCSPACE{O(s)}{c}{O(1)}$$
\end{theorem}
\begin{proof}
Let $L \in \BPepsCdeltSpace{s}{c}{\frac{1}{c}}{\epsilon}$, and let $N$ be the machine that decides the input correctly with probability $ \frac{1}{2}+\epsilon$ using $s$ bits of work space and $c$ bits of catalytic space, such that $N$ does not reset the catalytic tape with probability $\frac{1}{c}$. As in \autoref{boostingavgr}, we run $N$ a constant  number of times to increase the probability of deciding correctly to at least $\frac{2}{3}$, re-using the same $c$ bits of catalytic space each time. Let's say this constant (which depends only on $\epsilon$) is $k_{\epsilon}$. By a union bound, the probability that we do not reset the catalytic tape is $\le \frac{k_{\epsilon}}{c}$, and so on average over the read-once randomness we make at most $c\cdot \frac{k_{\epsilon}}{c} \le k_{\epsilon}$ errors on the catalytic tape.
\end{proof}

\begin{corollary}\label{inversepolyinacgeverytimecorr}
Let $\epsilon$ be a constant. Then
$$\BPepsCdeltL{\frac{1}{poly(n)}}{\epsilon} \subseteq \AVGrBPLCL{O(1)}$$
\end{corollary}

\noindent
According to \autoref{inversepolyinacgeverytimecorr}, the class \(\BPepsCdeltL{\frac{1}{\text{poly}(n)}}{\frac{1}{2}}\) is contained within \(\AVGrBPLCL{O(1)}\), which means we currently do not have a method to de-randomize it. However, with the additional restriction of polynomial time, we can show that, similar to the class \(\BPepsCdeltL{\delta}{\epsilon}\) (as stated in \autoref{corroepsdeltagoodpiscl}), \(\AVGrBPLCL{O(1)}\) is indeed equivalent to \(\CL\).

\begin{theorem}\label{obsvfromcharts1}
Let $\delta,\epsilon$ be constants such that $\delta < 2\epsilon < 1$. Then,
$$\CL = \AVGrBPLCL{O(1)}\PT = \BPepsCdeltL{\delta}{\epsilon}\PT$$
\end{theorem}
\begin{proof}
By \autoref{corroepsdeltagoodpiscl} we have $\CL = \BPepsCdeltL{\delta}{\epsilon}\PT$, while by~\cite{CookLiMertzPyne25} we have $\CL = \BPCL\PT \subseteq \AVGrBPLCL{O(1)}\PT$. Finally, $\AVGrBPLCL{O(1)}\PT \subseteq \BPepsCdeltL{\delta}{\epsilon}\PT$ by the same proof as \autoref{averageforallinepsdeltaclass}, as the proof only requires re-running the relevant machine a constant number of times.
\end{proof}

\section{Model 3: $\AVGtauLCSPACE{s}{c}{e}$}
\label{sec:avgtau}

In this section, we first explore the relationship between $\AVGtauLCSPACE{s}{c}{e}$ and other classes.

\subsection{Relationship to randomized error}
Our first set of results connects $\AVGtauLCSPACE{s}{c}{e}$ to $\BPepsCdeltSpace{s}{c}{\delta}{\epsilon}$, and by extension to $\AVGrBPLCSPACE{s}{c}{e}$. We specifically show that, in the low-error regime, $\BPepsCdeltL{\delta}{\epsilon}_{\low}$ is a subset of $\AVGtauLCL{o(1)}$. 

\begin{theorem}\label{epsdeltaclassinavgclass}
Let $\delta,\epsilon$ be constants such that $\delta < 2\epsilon$. Then
$$\BPepsCdeltSpace{s}{c}{\delta}{\epsilon} \subseteq \AVGtauLCSPACE{O(s)}{2^{O(s)}}{o(1)}$$
\end{theorem}
\begin{proof}
Let $L \in \BPepsCdeltSpace{s}{c}{\delta}{\epsilon}$, and let $N$ be the corresponding machine for $L$. By \autoref{lemmvianisansprggeneral}, both the subroutines $\mathcal{F}^L$ and $\mathcal{R}^L$ use catalytic space $2^{d\cdot s}$ for some constant $d$. Our machine $M$ will use a catalytic tape consisting of $r\cdot t$ chunks, where $r=2^{2s}$, $t=s$, and each chunk is of size $2^{ds}$ bits.

We invoke $\mathcal{F}^L$ on all the chunks one by one, resetting each chunk after its call using $\mathcal{R}^L$ if it was compressed. If $\mathcal{F}^L$ correctly decides the input for any chunk, we output the result. Otherwise, assume there exists a layer where $\mathcal{F}^L$ frees up the last bit of each chunk in that layer. Using $\mathcal{F}^L$ and $\mathcal{R}^L$, we figure out such a layer and free up the last bit of each chunk in the layer. Using the freed-up $2^{2s}$ bits, we run a space-inefficient algorithm to simulate $N$ deterministically and find the answer, resetting all the chunks in the layer afterwards.

If no such layer exists, then every layer has at least one chunk where the subroutine outputted $\perp$. For a random catalytic tape, this occurs with a probability of at most $\big (\frac{r}{2^{4s}} \big )^t = \frac{1}{2^{2s^2}}$, so we can destroy the catalytic tape and use a space-inefficient algorithm to find the correct answer. Thus, we always find the right answer, and the average number of errors we make over all initial catalytic tapes is 
$$
\le \frac{t.r.2^{ds}}{2^{2s^2}}  \le \frac{s2^{(d+2)s}}{2^{2s^2}} = o(1)$$
Our catalytic usage is at most $r\cdot t\cdot 2^{d\cdot s} = s \cdot 2^{(d+2)s}$, while we use at most $O(s)$ work space to run $\mathcal{F}^L$ and $\mathcal{R}^L$ by \autoref{lemmvianisansprggeneral}.
\end{proof}

\begin{corollary}\label{epsdelinavgcorr}
Let $\delta,\epsilon$ be constants such that $\delta < 2\epsilon$. Then
$$\BPepsCdeltL{\delta}{\epsilon} \subseteq \AVGtauLCL{o(1)}$$
\end{corollary}


\subsection{Equivalence to read-multiple randomness}
We now show a connection between $\AVGtauLCSPACE{s}{c}{e}$ and read-multiple errorless randomized catalytic computing. To do this, we first introduce another new catalytic class, which will be convenient for facilitating our results.

\begin{definition}
    A language $L$ is in $\ZPtauCSPACE{s}{c}$ if there is a deterministic catalytic machine $M$ with free space $s$ and catalytic space $c$ that always resets the catalytic tape and outputs $L(x)$ with probability $\frac{1}{2}$ over the initial catalytic tape and $\perp$ otherwise.
\end{definition}

\noindent
We begin by connecting $\ZPtauCSPACE{s}{c}$ to $\AVGtauLCSPACE{s}{c}{e}$.

\begin{theorem}\label{obavgclasszptau}
$$\AVGtauLCSPACE{s}{c}{e} \subseteq \ZPtauCSPACE{O(s+e\log{c})}{c}$$
\end{theorem}
\begin{proof}
Let $M$ be the $\AVGtauLCSPACE{s}{c}{e}$ (deterministic) machine, $x$ the input, and $\tau$ the initial catalytic contents. Then, by Markov's inequality, there are at most a $\frac{1}{10}$ fraction of initial catalytic tapes for which $M$ makes more than $10e$ errors.

Using $O(e\log{c})$ space, we enumerate all the $\sum_{i=0}^{10e} \binom{c}{i} \le c^{10e+1}$ possibilities where $\le 10e$ errors can be made on the catalytic tape. Let each possibility be represented by its characteristic error vector, denoted by $z$. For each $z$, we perform a DFS starting from the halt states $\acc[\tau \oplus z]$ and $\rej[\tau \oplus z]$, using the subroutine $\nextstep$ in \autoref{reversiblewalk}, running until we return to the unique halt state from which we started. If we encounter \textit{any} start state during the DFS from $\acc[\tau \oplus z]$, we accept, and likewise, we reject if we encounter any start state during the DFS from $\rej[\tau \oplus z]$. If we never see any start state, we output $\perp$.

Each DFS can be performed using $O(s)$ work space and $c$ bits of catalytic space, and hence we use  $O(s + e\log{c})$ work space and $c$ catalytic space
If $\tau$ is such that $M$ makes at most $10e$ errors with tape contents $\tau$, then some error vector $z'$ would correspond to exactly where these errors are made. Therefore, $\start$ will be encountered in either DFS from $\acc[\tau \oplus z']$ or from $\rej[\tau \oplus z']$ (depending on whether $x$ is in the language or not); furthermore, it cannot happen that we encounter start states in both the DFS from $\acc[\tau \oplus z]$ and $\rej[\tau \oplus z]$. Thus, we output the correct answer for such an error vector $z'$. Since $M$ makes more than $10e$ errors only for $\frac{1}{10}$ initial tapes, we output $\perp$ for at most this fraction of initial tapes.
\end{proof}

\noindent
We note that the proof of \autoref{obavgclasszptau} essentially uses the same idea as that used by Folkertsma et al. \cite{FolkertsmaMertzSpeelmanTupker26} in one of their proofs for $\LCSPACE{s}{c}{e} \subseteq \CSPACE{O(s+e\log{c})}{c}$.

\begin{theorem}\label{avgzptclotherside}
$$\ZPtauCSPACE{s}{c} \subseteq \AVGtauLCSPACE{O(s)}{O(c\log{c})}{o(1)}$$
\end{theorem}
\begin{proof}
    Let \( M \) be the \(\ZPtauCSPACE{s}{c}\) deterministic machine that outputs the correct answer on at least $\frac{1}{2}$ of all possible initial tapes and outputs \(\perp\) on the remaining tapes. We run \( M \) for \(\log{c^2}\) iterations on the given input, using a different segment of length $c$ from the catalytic tape for each iteration. If any run returns an output besides $\perp$, we return that output; otherwise, we erase the first $c$ bits of the catalytic tape and perform a brute-force search to find an initial tape $\tau$ for which $M$ does not output $\perp$. Clearly, we use $O(s)$ work space and $O(c\log c)$ bits on the catalytic tape. We always output the correct answer according to the definition of $M$, and we destroy the tape with a probability of at most $\frac{1}{2^{\log{c^2}}} = \frac{1}{c^2}$. Thus, we make at most $\frac{c}{c^2}= o(1)$ errors on average (over the tape). 
\end{proof}

\noindent
In order to connect $\ZPtauCL$ to $\ZPstartCL$, we need to move to the polynomial time variant, which, as it turns out, can be done for the former class without loss of generality:

\begin{theorem}\label{obsvzptauclsamewithptime}
   $$\ZPtauCL = \ZPtauCL\PT$$
\end{theorem}
\begin{proof}
We need only prove the forward direction, as the reverse holds by definition. Let \( M \) be the (deterministic) \(\ZPtauCL\) machine, and let \( x \) be the given input. Also, let $M$ use $s$ bits of work space and $c$ bits of catalytic space. As discussed in Section \ref{sc:grtrv}, we know that for any \(\tau\), the graph \(\gmx(\start)\) forms a tree, with its sink being one of the three halt configurations: \(\acc\), \(\rej\), or \(\dontknow\), depending on whether \( M \) outputs accept, reject, or \(\perp\) on input \( x \). 

Since \( M \) always restores the catalytic tape to its initial contents, the trees corresponding to different starting catalytic tapes are pairwise vertex-disjoint. By applying an averaging argument similar to that used by Buhrman et al. \cite{BuhrmanCleveKouckyLoffSpeelman14}, we find that the average size of a tree (over \(\tau\)) is at most \( 2^{4s} \). Using Markov's inequality, the probability (over \(\tau\)) that a tree has a size of at least \( 10\cdot2^{4s} \) is at most \( \frac{1}{10} \). Additionally, since \( M \) outputs a non-\(\perp\) answer with a probability (over \(\tau\)) of at least \( \frac{1}{2} \), we conclude that for at least \( \frac{1}{2} - \frac{1}{10} = 0.4 \) fraction of the \(\tau\), \( M \) outputs a non-\(\perp\) answer while the tree \(\gmx(\start)\) has a size smaller than \( 10\cdot2^{4s} \).

Our $\ZPtauCL\PT$ machine operates as follows: Given a catalytic tape $\tau$, it first utilizes the subroutine $\countsize$ from \autoref{reversiblewalk} to determine if the size of the tree $\gmx(\start)$ is less than $10\cdot2^{4s}$. If it is, the machine employs the subroutine $\nextstep$ (also from \autoref{reversiblewalk}) to perform a walk starting from $\start$ over $\gmx(\start)$ for $10\cdot2^{4s}$ steps. The output is determined based on the halting configuration observed during this walk: we accept if we see $\acc$, reject if we see $\rej$, and output $\perp$ if we see $\dontknow$. Before producing the final answer, the machine uses the subroutine $\stepback$ to return to $\start$, thereby resetting the catalytic tape. If the size of the tree $\gmx(\start)$ is $\ge 10\cdot2^{4s}$, the machine outputs $\perp$.

Using \autoref{reversiblewalk}, we know our machine can use the mentioned subroutines using \(O(s)\) work space, and \(c\) bits of catalytic space (which is used to simulate that of \(M\)). Moreover, our machine runs in time $2^{O(s)} = \text{poly}(n)$. Furthermore, our machine never outputs an incorrect answer (which follows from the definition of \(M\)). We know that for at least \(0.4\) fraction of the initial catalytic tapes \(\tau\), the tree \(\gmx(\start)\) has size smaller than \(10\cdot 2^{4s}\), such that the unique sink/halting configuration in the tree is either \(\acc\) or \(\rej\). Thus, our machine outputs a non-\(\perp\) answer for at least \(0.4\) fraction of \(\tau\).

Finally, we note that we can boost this probability to any constant by simply running the machine a constant number of times, each time using a different section of the catalytic tape.
\end{proof}

\noindent
Finally, we can connect the poly-time variants of $\ZPtauCL$ and $\ZPstartCL$, thus successfully characterizing $\AVGtauLCL{O(1)}$:

\begin{theorem} \label{lem:avgtau_zpstarclp}
    $$\AVGtauLCL{O(1)} = \ZPtauCL\PT = \ZPstartCL\PT$$
\end{theorem}
\begin{proof}
The first equality is a corollary of \autoref{obavgclasszptau}, \autoref{avgzptclotherside}, and \autoref{obsvzptauclsamewithptime}. Thus, we show both directions of the second equality.

($\Rightarrow$) Let $M$ be a $\ZPtauCL\PT$ machine that uses $c$ bits of catalytic tape. For a given catalytic tape $\tau$ (of length $c$), we read $c$ two-way random bits, which we denote by $r$. We then simulate $M$ on the given input using catalytic tape $\tau \oplus r$ and output the result of this simulation. Since $M$ always restores the catalytic tape to its initial contents, we can likewise restore our tape by re-reading the bits of $r$. Furthermore, we know $M$ outputs a non-$\perp$ (and correct) answer with probability $\frac{1}{2}$ over the initial tape. Since $\tau\oplus r$ gives us a uniformly random tape, we output a non-$\perp$ answer with the same probability over the randomness $r$.

($\Leftarrow$) Let the $\ZPstartCL\PT$ machine $N$ use $m = poly(n)$ random bits on inputs of length $n$, and $c$ bits of catalytic tape. We construct a machine that uses a catalytic tape of size $c + m$ bits: the first $c$ bits are used to emulate the catalytic tape of $N$, and the remaining $m$ bits are for the randomness used by $N$. We output the answer we obtain by simulating $N$. Since, for \textit{every} initial tape configuration, $N$ produces a non-$\perp$ output with probability $\ge \frac{1}{2}$ over its randomness, our machine will likewise produce a non-$\perp$ output with the same probability over our catalytic tape.
\end{proof}

\noindent
We close by noting that this connection to $\ZPstartCL\PT$ gives us a number of results that do not obviously follow for $\AVGtauLCSPACE{s}{c}{O(1)}$. First, there is a containment in randomized polynomial time for the non-polynomial time variant of $\AVGtauLCL{O(1)}$:

\begin{theorem} \label{lamezpprmk}
    $$\AVGtauLCL{O(1)} \subseteq \ZPP$$
\end{theorem}
\begin{proof}
    By \autoref{lem:avgtau_zpstarclp}, $\AVGtauLCL{O(1)} \subseteq \ZPstartCL\PT$. Note that for \textit{every} initial catalytic tape, a $\ZPstartCL\PT$ machine runs in polynomial time, uses polynomial space between its work and catalytic tapes, and outputs a non-$\perp$ answer with a probability of at least $\frac{1}{2}$. Thus $\ZPstartCL\PT$ is trivially contained in $\ZPP$.
\end{proof}

\noindent
Second is a characterization of $\AVGtauLCL{O(1)}\PT$ as exactly being the intersection of $\AVGtauLCL{O(1)}$ with no time restriction and $\PT$ with no catalytic restriction; this was previously proven for $\CL\PT$ and $\CL \cap \PT$ by Cook et al.~\cite{CookLiMertzPyne25}:
\begin{theorem}
$$\AVGtauLCL{O(1)}\PT = \AVGtauLCL{O(1)} \cap \PT$$
\end{theorem}
\begin{proof}
The forward direction follows from the definitions of the classes. For the reverse direction, consider an arbitrary language $ L \in \AVGtauLCL{(O(1))} \cap \PT $; it follows from \autoref{lem:avgtau_zpstarclp} that \( L \in \ZPtauCL\PT \cap \PT\). Let \( M \) be the \(\ZPtauCL\PT\) (deterministic) machine for \( L \), such that \( M \) uses a catalytic tape of size \( n^d \) on input length \( n \) (for some constant \( d \)). Without loss of generality, we assume that \( L \) has a polynomial-time algorithm with a runtime of \( n^d \).

Our \(\AVGtauLCL{(O(1))}\PT\) machine \(N\) uses a catalytic tape of size \(n^{2d}\), which is broken into \(n^d\) chunks, each of size \(n^d\). \(N\) works as follows: given input \(x\), it simulates \(M\) for input \(x\) with each chunk serving as the starting catalytic tape for \(M\). Note that, by the definition of \(M\), after each simulation of \(M\), the contents of the respective chunk are reset to their original state. If \(N\) receives a non-\(\perp\) answer from even one chunk during the simulation of \(M\), it outputs that answer. Otherwise, it destroys its catalytic tape and uses the space to run the polynomial-time algorithm for \(L\).

Thus, \( N \) always decides the input correctly. Since \( M \) runs in polynomial time, \( N \) does as well. Lastly, since \( M \) outputs a non-\(\perp\) answer for at least \(\frac{1}{2}\) of its initial tapes, the probability (over the uniform choice of its starting catalytic tape) that \( N \) never sees \( M \) output a non-\(\perp\) answer for any of the chunks is \( \le \frac{1}{2^{n^d}}\). Thus, the average number of errors made by \( N \) (averaged over its starting tape) is \(\frac{n^{2d}}{2^{n^d}} = o(1)\).
\end{proof}

\noindent
Lastly, the equivalence to $\ZPstartCL\PT$ provides evidence that proving $\AVGtauLCL{(O(1))} = \CL$ is \emph{much harder} than proving $\AVGrBPLCL{O(1)}=\CL$, as it would show novel derandomizations, such as reducing $\ZPTCo$ to the \emph{lossy coding problem} (see \cite{CookLiMertzPyne25} for further discussion).
On the flip side, an appropriate derandomization assumption would remove this barrier and give reason to believe that the collapse of $\AVGtauLCL{O(1)}$, and hence all previously discussed classes, could indeed occur; we validate this idea in the upcoming section.

\section{Further Results Assuming Derandomization}
\label{sec:assume}

We now move to conditional results, in particular assuming a fairly standard hardness assumption:
\begin{conjecture}\label{conj:derandom_assumption}
    There exists a constant $\gamma > 0$ such that $\Deterministic\SPACE(n) \not \subseteq \SIZE(2^{\gamma n})$.
\end{conjecture}

\noindent
\autoref{conj:derandom_assumption} is sufficient to construct \emph{pseudorandom generators} (PRGs); we will take the following construction due to Klivans and van Melkebeek~\cite{KlivansM02}: 
\begin{lemma}[Derandomization assumption]\label{ass:derandom}
    If \autoref{conj:derandom_assumption} holds, then for all constants $d$, there exists a constant $d'$ and a function $G : \{0,1\}^{d' \log{n}} \rightarrow \{0,1\}^n$ such that for
     any circuit $C$ of size $n^d$, $G$ $\frac{1}{n}$-fools the circuit $C$, i.e.,
    $$
    \big|\Pr_{r\in \{0,1\}^n}[C(r)=1]-\Pr_{s \in \{0,1\}^{d'\log{n}}}[C(G(s)=1]\big| < \frac{1}{n}
    $$
    and $G$ is computable in space logarithmic in $n$. 
\end{lemma}

\noindent
Note that \autoref{ass:derandom} is sufficient to show $\BPP = \PT$, which extends to $\BPepsCdeltL{\delta}{\epsilon}_{\high}\PT$ by \autoref{trivialp}:
\begin{theorem}
    Let $\delta, \epsilon$ be constants such that $\delta > 2\epsilon$.
    If \autoref{conj:derandom_assumption} holds, then
    $$\BPP = \BPepsCdeltL{\delta}{\epsilon}\PT = \PT$$
\end{theorem}

\noindent
It is not surprising that one can strengthen \autoref{thm:lcl_vs_cl}, particularly the fact that $\CL = \LCL{O(1)}$, by allowing at most $e$ errors for every initial catalytic setting to allowing at most $e$ errors on average (averaged over the initial catalytic tape), under the de-randomization assumption.

\begin{lemma}\label{lossyderandlemma}
If $e \leq c^{1-\Omega(1)}$ and \autoref{conj:derandom_assumption} hold, then
$$\AVGtauLCSPACE{\Theta(s)}{\Theta(c)}{e} = \CSPACE{\Theta(s +e\log{c})}{\Theta(c)}$$
\end{lemma}
\begin{proof}
The reverse direction holds by \autoref{thm:lcl_vs_cl}, and so we focus on the forward direction. Let \( L \in \AVGtauLCSPACE{s}{c}{e} \) and hence \( L \in \ZPtauCSPACE{O(s+e\log{c})}{c} \) by \autoref{obavgclasszptau}. Let \( M \) be the deterministic machine for \( L \), as per the definition of this class, and let \( x \) be the input and $\tau$ the initial catalytic setting. Also, let $s'=O(s+e\log{c})$ be the work space used by $M$. As argued in the proof of \autoref{obsvzptauclsamewithptime}, one can show that the probability (over $\tau$) that $M$ outputs a non-$\perp$ answer and that the size of the configuration graph $G_{M,x,\tau}$ is at most $T=10\cdot 2^{4s'}$ is at least $0.4$, in which case we say the configuration graph is \emph{good}. Consider the circuit $C_{x,\tau}$, which takes another catalytic setting $w$ as input and outputs $1$ if and only if the configuration graph $\gmxtau[\tau \oplus w]$ is good. As $M$ is a deterministic machine, $\gmxtau[\tau \oplus w]$ is simply a path. Thus, the circuit can check the configuration after $T$ simulation steps of $M$ and output $1$ if it is either $\acc[\tau \oplus w]$ or $\rej[\tau \oplus w]$. Therefore, using the fact that $e \le c\le 2^s$, the size of the circuit can be bounded by a sufficiently large polynomial in $T$.  Thus, we get that $C_{x,\tau}$ outputs $1$ for at least $0.4$ fraction of inputs $w$. Note that we can give the circuit a longer input, say $w \circ r$, where $r$ is a redundant part that the circuit ignores; such that the length of input $w \circ r$ is, say, $O(T)$. The circuit size is now polynomial in its input size. Under the derandomization assumption in \autoref{ass:derandom}, there exists a PRG $P$ that has a seed length of $O(\log(T)) = O(s')$ and is computable in space $O(\log(T)) = O(s')$, such that $P$ $\frac{1}{10}$-fools the circuit $C_{x,\tau}$. 
Thus, for any $\tau$, the fraction of seeds of $P$ making the circuit $C_{x,\tau}$ output $1$ is at least  $0.3$. Finally, we describe the $\CSPACE{O(s')}{c}$ machine $N$ for $L$. $N$ simply does the following: For every seed $q$ of the PRG $P$, it computes $P(q)$ and simulates $M$ with the catalytic tape $\tau \oplus P(q)$. If, for any seed, we see $M$ give a non-$\perp$ answer, $N$ outputs that; otherwise, it outputs ERROR. The machine \( N \) can simulate the machine \( M \) using a work space of \( O(s') \) and a catalytic space \( c \). Additionally, since the PRG \( P \) has a seed length of \( O(s') \) and is computable in space \( O(s') \), \( N \) can exhaustively try all possible seeds. Furthermore, because \( M \) always restores the catalytic tape to its initial state, \( N \) can also restore the catalytic tape by XORing it back with \( P(q) \) after running the simulation for seed \( q \). Observe that there exists a seed \( q \) such that when \( M \) is run with input \( x \) and catalytic tape \( \tau \oplus P(q) \), it outputs a non-\(\perp\) answer. By the definition of \( M \), this will be the correct output. Therefore, \( N \) always provides the correct answer and never outputs ERROR. 
\end{proof}

\noindent
\autoref{lossyderandlemma} is enough to show that all the catalytic classes discussed in this paper, with the sole exceptions of $\ZPstartCL$ and $\BPepsCdeltL{\delta}{\epsilon}_{\high}(\PT)$, are equivalent, since $\AVGtauLCL{e}$ is the highest class in the chain of containments. Since $\CL \subseteq \ZPP$, this also immediately implies that $\CL \subseteq \PT$. Combining this with a result from Cook et al.~\cite{CookLiMertzPyne25}, which shows $\CL\PT = \CL \cap \PT$, we can thus extend this to $\CL = \CL\PT$ under \autoref{conj:derandom_assumption}, giving us:

\begin{corollary}\label{chcollpse}
If \autoref{conj:derandom_assumption} holds, then
$$\AVGtauLCL{O(1)} = \CL\PT$$
\end{corollary}

\noindent
Besides giving a full characterization of all classes discussed in this paper, another interesting way of interpreting \autoref{chcollpse} is that it ``scales down'' a result showing e.g. $\PT = \ZPP$ to show $\CL\PT = \ZPstartCL\PT$.

We also make a note about more general parameters. When derandomizing $\BPepsCdeltSpace{s}{c}{\delta}{\epsilon}$ into $\AVGtauLCSPACE{s}{c}{o(1)}$ using \autoref{epsdeltaclassinavgclass}, we end up using catalytic space $2^{O(s)}$ regardless of the value of $c$. However, in \autoref{proofepsdelderan} we show that \autoref{conj:derandom_assumption} allows us to avoid $\AVGtauLCSPACE{s}{c}{o(1)}$ and derandomize directly to $\CSPACE{s}{c}$ without incurring this overhead:

\begin{theorem}\label{epdeltaderantgh}
Let $\delta, \epsilon$ be constants such that $\delta < 2\epsilon$. Then, assuming \autoref{conj:derandom_assumption} holds, we have
$$\BPepsCdeltSpace{s}{c}{\delta}{\epsilon} \subseteq \CSPACE{O(s)}{O(c)}$$
\end{theorem}

\noindent
This is a ``scaled up'' version of a result due to \cite{KouckyMertzPyneSami25}, which shows that $\BPL = \Logspace$ (which is true under \autoref{conj:derandom_assumption}) implies $\BPCSPACE{s}{c} \subseteq \CSPACE{O(s)}{O(c)}$ (the former class being equivalent to $\BPepsCdeltSpace{s}{c}{0}{\epsilon}$)\footnote{A recent result \cite{OpenAI26LBPL} claims $\BPL=\Logspace$. If true, then \cite{KouckyMertzPyneSami25} unconditionally gives us $\BPCSPACE{s}{c} \subseteq \CSPACE{O(s)}{O(c)}$. However, we do not know if $\BPL=\Logspace$, instead of the stronger assumption \autoref{conj:derandom_assumption}, gives \autoref{epdeltaderantgh}.}.

\addcontentsline{toc}{section}{References}
\bibliographystyle{alpha}
\bibliography{refs.bib}

\appendix

\section{Proof of  \autoref{epdeltaderantgh}}\label{proofepsdelderan}


For constants $0 \le \delta < 2\epsilon$, we consider an arbitrary language $L \in \BPepsCdeltSpace{s}{c}{\delta}{\epsilon}$, with $M$ being the corresponding machine according to the definition of the class. Going forward, we will work with a fixed input $x$. Without loss of generality\footnote{Let $\delta',\epsilon'$ be rationals such that $\delta < \delta' < 2\epsilon' < 2\epsilon$. Then, by definition $ \BPepsCdeltSpace{s}{c}{\delta}{\epsilon} \subseteq \BPepsCdeltSpace{s}{c}{\delta'}{\epsilon'}$.}, we can assume $\delta$ and $\epsilon$ are rational.

The proof of the following lemma demonstrates that the class $\BPepsCdeltSpace{s}{c}{\delta}{\epsilon}$ (low-error regime) is contained within the class of languages that can be solved by a bounded-probabilistic machine running in time $2^{O(s)}$. This will be useful for proving \autoref{epdeltaderantgh}. Recall our discussion from \autoref{sc:lowerr} about the $\taubeta$-graph, in which we had set $\beta$ to be the constant $\frac{1}{2} \left(1 - \frac{\delta}{2\epsilon}\right)$. It followed from \autoref{tauepsp} that, for this value of $\beta$, the probability of exiting the $\taubeta$-graph via $\perp$ is at most the constant $\gamma = \frac{2\delta}{1 + \frac{\delta}{2\epsilon}}$, which is less than $2\epsilon$.

\begin{lemma}\label{epsdelclassinbpp}
Let $\delta,\epsilon$ be constants such that $\delta < 2\epsilon$. Then, we have that $\BPepsCdeltSpace{s}{c}{\delta}{\epsilon} \subseteq \mathrm{BPTIME[\frac{1}{2}-\epsilon,\frac{1}{2}-\epsilon + \eta](2^{10s})}$, for some constant $\eta > 0$.
\end{lemma}
\begin{proof}
Consider the algorithm $\mathcal{A}$ which works as follows (given input $x$): it first selects $\tau$, an initial setting for the catalytic contents, uniformly at random. Then, it takes a random walk of length $T \cdot2^{4s}$ from $\start$ in $G_{M,x,\tau}$. If it reaches $\acc$, it accepts; otherwise, it rejects. Here, $T$ is a constant that we set to $
\lceil 10(1+\frac{\frac{1}{2}-\epsilon}{(2\epsilon-\gamma)}) \rceil$.
As the size of $\tau$ is $c \le 2^s$, being generous, it follows that $\mathcal{A}$ is a probabilistic algorithm that runs in time $2^{10s}$.

If the input $x \not \in L$, then regardless of $\tau$, the probability of reaching $\acc$ from $\start$ is $\le \frac{1}{2}-\epsilon$; thus $\mathcal{A}$ accepts with a probability of $\le \frac{1}{2}-\epsilon$. Next, consider the case where $x \in L$. Using \autoref{avgsize} and the Markov inequality, we know that with a probability (over $\tau$) of at least \(1-\frac{1}{T}\), the size of the \(\tau^{\beta}\)-graph is less than \(T \cdot 2^{4s}\). In such a  scenario, since the length of the walk is greater than the size of the \(\tau^{\beta}\)-graph, the random walk will either exit the \(\tau^{\beta}\)-graph to a non-\(\tau^{\beta}\)-node or halt at either \(\acc\) or \(\rej\) within the \(\tau^{\beta}\)-graph. Using \autoref{tauepsp} we know that the probability the walk leaves the \(\tau^{\beta}\)-graph (to a non-\(\tau^{\beta}\)-node) is $\le \gamma< 2\epsilon$, for our choice of $\beta$. By the definition of $M$, the probability that the walk goes to $\rej$ (in the $\tau^{\beta}$-graph)  is $\le \frac{1}{2}-\epsilon$. Thus, $\mathcal{A}$ accepts with a probability of at least

\begin{align*}
\bigg(1-\frac{1}{T}\bigg)\bigg(1-\bigg(\frac{1}{2}-\epsilon\bigg) -\gamma\bigg) &\ge \bigg(1-\frac{2\epsilon-\gamma}{10(\frac{1}{2}+\epsilon-\gamma)}\bigg)\bigg(\frac{1}{2}+\epsilon -\gamma\bigg) \\
&\ge \frac{5+8\epsilon-9\gamma}{10} \\
&\ge \frac{1}{2} + \frac{8\epsilon-9\gamma}{10} \\
&= \frac{1}{2} -\epsilon + \frac{9(2\epsilon-\gamma)}{10}
\end{align*}
Since $(2\epsilon -\gamma) > 0$ is a constant, the probability that $\mathcal{A}$ accepts when $x \in L$ is bounded away from $\frac{1}{2}-\epsilon$ by a positive constant; hence, the lemma follows by setting $\eta = \frac{9(2\epsilon-\gamma)}{10}$.
\end{proof}


\begin{remark}
Using \autoref{epsdelinavgcorr} and \autoref{lamezpprmk} we already know that for constants $\delta < 2\epsilon$, $\BPepsCdeltL{\delta}{\epsilon} \subseteq \mathrm{ZPP}$.
\end{remark}

Consider the algorithm $\mathcal{A}'$ that takes as input $x$ and an initial catalytic setting $\tau$. It behaves exactly like the algorithm $\mathcal{A}$ mentioned in the proof above, but it performs the walk over $\gmxtau[\tau\oplus w]$, where $w$ is chosen uniformly at random. As we were generous with time estimates, $\mathcal{A}'$ also runs in probabilistic time $2^{10s}$. As for any $\tau$, $\tau\oplus w$ also gives a uniform distribution over all catalytic settings, $\mathcal{A}'$ decides $x$ with the same probability separation as $\mathcal{A}$. Let $C_{x,\tau}$ denote the circuit that performs the computation of $\mathcal{A}'$ for given random bits as input; i.e., the circuit takes an input of length $N= c+T
\cdot 2^{4s}$, where the first $c$ bits are for the random catalytic tape setting; and the latter bits are for the random walk. Since $c \le 2^s$, the size of the circuit can be bounded by a large enough polynomial in $2^{10s}$. Thus, the size of the circuit is polynomial in the size of the input; therefore, by \autoref{ass:derandom}, there exists a PRG $P: \{0,1\}^{d'\log{N}} \rightarrow \{0,1\}^N$ that $\frac{1}{N}$-fools the circuit $C_{x,\tau}$ and can be computed in the space $O(\log{N}) = O(s)$ (where $d'$ is some constant).

\begin{observation}\label{epsdeltcircuitargument}
The circuit $C_{x,\tau}$ outputs $1$ (accepts) on at most a $\frac{1}{2}-\epsilon$ fraction of inputs when $x \notin L$, and on at least a $\frac{1}{2}-\epsilon + \frac{9(2\epsilon - \gamma)}{10}$ fraction of inputs when $x \in L$.
\end{observation}
\begin{proof}
    Follows from the definition of the circuit $C_{x,\tau}$ and \autoref{epsdelclassinbpp}.
\end{proof}

Before giving the proof of \autoref{epdeltaderantgh}, we will need a few definitions.

\begin{definition}[Layered Configuration Graph]\label{layeredconfgrph}
Let \( k \in \mathbb{N} \). We define the layered configuration graph with \( k \) layers associated with \( \gmx \) as follows:
\begin{itemize}
\item Its vertex set consists of pairs $\conf{v}{i}$ for every \( v \in V(\gmx) \) and \( i \in [k] \cup \{0\}\).
\item For every directed edge \((v,v') \) in \(\gmx\) with label \( b \in \{0,1\} \) and \( i \in [k] \), the layered graph contains a directed edge from $\conf{v}{i}$ to $\conf{v'}{i-1}$, carrying the same label \( b \).
\item Additionally, for every halt configuration \( v \in V(\gmx) \) and \( i \in [k] \), the layered graph includes a directed edge from $\conf{v}{i}$ to $\conf{v}{i-1}$ that is labeled with both \(0\) and \(1\).
\end{itemize}
One can imagine the layers arranged from left to right, numbered from $k$ to $0$. Note that the vertices in the last layer (layer $0$) have no outgoing edges.
\end{definition}

Thus, every vertex in the layered configuration has at most two outgoing edges, with each edge label assumed to occur with a probability $\frac{1}{2}$. Notice that the probability of transitioning from $\conf{v}{i}$ to $\conf{v'}{0}$ in the layered graph, where \(v'\) is a halting configuration, is the same as the probability (over the read-once randomness of \(M\)) of reaching the configuration \(v'\) within $i$ steps when the machine \(M\) is executed starting from \(v\).

\begin{remark}[Notation]
The vertices of the layered configuration graph, which we will simply refer to as configurations, are denoted by \( \conf{v}{i} \). In this context, \( v \in V(\gmx) \) is represented as \( \conf{\pi}{u} \), where \( \pi \in \{0,1\}^c \) and \( u \in \{0,1\}^s \). To avoid nested brackets, we will also denote \( \conf{v}{i} \) as \( \confl{\pi}{u}{i} \). 
\end{remark}

\begin{definition}[$y$-tree]\label{ytree}
Let \( y \in \{0,1\}^m \). Consider the subgraph of the layered configuration graph of \(\gmx\), that uses the same vertex set but includes only specific edges. For each \( i \in [m] \) and for every vertex \( \conf{v}{i} \) in layer \( i \), we include only the outgoing edge labeled with \( y_{m-i+1} \).
For $i>m$ we include no outgoing edges.

For any vertex \( \conf{v'}{j} \), the connected component of this subgraph that contains \( \conf{v'}{j} \) is referred to as the \emph{\( y \)-tree} of \( \conf{v'}{j} \), which is denoted as \( y \)-tree\((\conf{v'}{j}) \).
\end{definition}

Note that by definition $y$-tree($\conf{v}{j}$) is the isolated vertex $\conf{v}{j}$ for $j>m$. Whereas for $j\le m$, $y$-tree($\conf{v}{j}$) is a directed tree whose unique sink is a vertex in layer $0$ that can be reached from $\conf{v}{j}$ by following the edge labels obtained by reading the bits $y_{m-j+1}$ to $y_m$.  

\begin{remark}\label{extendingthelemmatolayered}
    Assuming we have oracle access to \( y \), the subroutines outlined in \autoref{reversiblewalk} can be easily extended to traverse \( y \)-tree($\conf{v}{j}$) while maintaining the time and space guarantees provided in \autoref{reversiblewalk}. This is possible as long as \( |y| \le 2^{O(s)} \), which will be the case in our usage; since we only need an additional \( O(s) \) work space to keep track of the layer number of the current vertex/configuration we are at while performing a walk.  Furthermore, since a \( y \)-tree(\(\conf{v}{0}\)) has \(\conf{v}{0}\) as its sink, the guarantees from \autoref{reversiblewalk} for a halting configuration also hold for the sink \(\conf{v}{0}\). 
\end{remark}

 \begin{definition}[$y$-DFS]\label{ydfs}
 When we refer to performing a $y$-DFS starting from $\conf{v}{j}$, we mean that we execute the walk defined by $\nextstep$ (in \autoref{reversiblewalk}) over $y$-tree($\conf{v}{j}$), starting from $\conf{v}{j}$. Conversely, when we say we conduct a $y$-DFS from $\conf{v}{j}$ in a reversible manner, we mean that we perform the walk in reverse order using $\stepback$ (also in \autoref{reversiblewalk}).
 \end{definition}

\begin{figure}[h!]
\begin{center}
\begin{tikzpicture}[every node/.style={font=\small}]
  \def\boxwidth{2.5}
  \def\boxheight{1}
  \def\braceYOffset{0.2}  
  \def\arrowYOffset{0.4}  
  \foreach \i/\name in {0/$\tau$,1/\textbf{tar}$_1$,2/\textbf{tar}$_2$,3/\dots,4/\textbf{tar}$_r$} {
    \draw (\i*\boxwidth,0) rectangle ++(\boxwidth,\boxheight);
    \node at (\i*\boxwidth+0.5*\boxwidth,0.5*\boxheight) {\name};
  }

  \foreach \i in {1,2,3,4} {
    \draw (\i*\boxwidth,0) -- ++(0,\boxheight);
  }

\end{tikzpicture}
\caption{Catalytic tape}
\label{tape}
\end{center}
\end{figure}

To prove \autoref{epdeltaderantgh}, we describe a deterministic catalytic algorithm $\mathcal{B}$, which always resets the tape and decides the language $L$. $\mathcal{B}$ has the following catalytic tape structure: it consists of $c$ bits denoted by $\tau$ in \autoref{tape}, followed by $r$ strings $\textbf{tar}_1$ to $\textbf{tar}_r$, called targets. We set $r=\lceil \frac{c}{s} \rceil$, and the length of each string $\textbf{tar}_j$ is $10d\cdot s$, where $d$ is a constant such that $d \ge 5$ and $d\cdot s$ is an upper bound to the space required to compute the PRG $P$ (described above) for a given seed. This includes the space to hold a seed. At a high level, $\mathcal{B}$ employs a \emph{compress-or-compute} approach introduced by Cook et al. \cite{CookLiMertzPyne25}. Here, $\mathcal{B}$ either compresses the catalytic tape to create enough space for an inefficient brute-force algorithm or is able to decide the input quickly.

Assume that we are processing $\textbf{tar}_j$ and the first $c$ bits of the catalytic tape at this point are $\pi$. $\mathcal{B}$ works as follows: 

\begin{enumerate}
\item We iterate through all the seeds $q$. For each seed $q$, $P(q)$ gives us a string of length $N=c+T\cdot 2^{4s}$; where we treat the first $c$ bits as a catalytic setting $w$ and the latter bits as choices for a random walk, denoted by $y$.
\item For each $w,y$ and $i\in [T\cdot 2^{4s}]$, we call $\countsize$ from \autoref{reversiblewalk} on $y_{\le i}$-tree$(\acczp{\pi \oplus w})$ to see if any of the trees is larger than $2^{10d\cdot s}$.
\item \label{caseanyonebig} If for some $w,y,i$ the size of $y_{\le i}$-tree$(\acczp{\pi \oplus w})$ is larger than $2^{10d\cdot s}$, 
we compress the contents of the catalytic tape using time-step compression \cite{CookLiMertzPyne25} as follows. 
Let $val(\textbf{tar}_j)$ represent the integer value of the binary string $\textbf{tar}_j$ incremented by one. 
Thus, $val(\textbf{tar}_j)$ is an integer in $[2^{10d\cdot s}]$. 
We call the walk procedure $\nextstep$ for $val(\textbf{tar}_j)$ times to get from $\acczp{\pi \oplus w}$ to some configuration $\confl{\pi'}{u}{k}$ where $\pi'$ is the catalytic part, $u$ is work space part, and $k \le T\cdot 2^{4s} < 2^{5s}$ is the layer index. 
In place of $\textbf{tar}_j$ on the catalytic tape, we write $q \circ u \circ k \circ i \circ 0^{10d\cdot s - (ds+5s+5s + 5s)}$ where $q$ is the seed corresponding to $w,y$. 
Here, we used the fact that $q$ can be specified using $ds$ bits and $k,i$ using $5s$ bits each.
Also, the work space part of the configuration $u$ (which includes head locations, etc.) can be specified using $5s$ bits. 
Since $d \ge 5$, we managed to free up at least $6d\cdot s$ bits. 
The walk also replaced $\pi \oplus w$ with $\pi'$ on the catalytic tape. 
After compression, we move on to the next target $\textbf{tar}_{j+1}$.

This step is reversible, as we can always decompress as follows: we walk back using the subroutine \(\stepback\) in \autoref{reversiblewalk}, from \(\confl{\pi'}{u}{k}\) over the \(y_{\le i}\)-tree of \(\acczp{\pi \oplus w}\), where \(y_{\le i}\) can be computed using \(i\) and \(q\). We count the number of steps it takes for us to reach the unique halting/sink state $\acczp{\pi \oplus w}$. This count is precisely the value of $\textbf{tar}_j$. Thus, we recover $\textbf{tar}_j$. At this point, the first $c$ bits of the catalytic tape are $\pi \oplus w$. Finally, we can compute $w$ from $q$ and XOR it with the first $c$ bits of the catalytic tape, resetting it to $\pi$. 
\item \label{caseallsmall} If the sizes of all the $y_{\le i}$-tree$(\acczp{\pi \oplus w})$  are smaller than $2^{10d\cdot s}$, we count the number of pairs $(w,y)$ such that for some $i$, the $y_{\le i}$-DFS from $\acczp{\pi\oplus w}$ visits $\conf{\start[\pi\oplus w]}{i}$. We call such a pair $(w,y)$ \textit{accepting}. If the fraction of seeds $q$ for which the corresponding pair $w,y$ is accepting is at least $\frac{1}{2}-\epsilon +\frac{27(2\epsilon-\gamma)}{40}$, we accept; otherwise, we reject. In this step, we essentially do not change the catalytic contents $\pi$ (each time we XOR it with some $w$, we XOR it back before moving to the next seed). Thus, before outputting the final answer, we just decompress all the previous targets as described in step \ref{caseanyonebig}.
\item \label{alltargetscompressed} In the case where we compress all the targets without hitting step \ref{caseallsmall}, we free up the last \(6d \cdot s\) bits of each target, thus freeing up \(6d \cdot s \cdot r \ge 6d \cdot s \cdot \frac{c}{s} > c\) bits in total on the catalytic tape. We claim that this is sufficient to decide the input. Again, before we provide the final output, we decompress all targets as described in step \ref{caseanyonebig}.
\end{enumerate}

The seed length \( q = O(s) \) and \( P(q) \) can be computed using space \( O(s) \). Additionally, any DFS we conduct is for at most \( 2^{O(s)} \) steps. Therefore, using \autoref{reversiblewalk}, we can see that all steps utilize \( O(s) \) space, except for step \ref{alltargetscompressed}. The space considerations for this step are established in \autoref{correctnesstwo} below. The size of the catalytic tape is given by \( c + 10d \cdot s \cdot \lceil \frac{c}{s} \rceil = O(c) \). Furthermore, we always reset the catalytic tape since the compression described in step \ref{caseanyonebig} is reversible. The correctness of the algorithm \( \mathcal{B} \) can be established through the following two observations:

\begin{observation}\label{correctnessone}
In case $\mathcal{B}$ hits step \ref{caseallsmall}, it correctly decides the input $x$.    
\end{observation}
\begin{proof}
Recall that the PRG $P$ $\frac{1}{N}$-fools the circuit $C_{x,\pi}$ (for every $\pi$). Since $\epsilon$ and $\gamma$ are constants satisfying $(2\epsilon - \gamma) > 0$, for large enough $N$, we have that  $\frac{1}{N} < \frac{9(2\epsilon - \gamma)}{40}$. Consequently, by applying \autoref{epsdeltcircuitargument}, we conclude that if $x \in L$, then at least a $\frac{1}{2} - \epsilon + \frac{27(2\epsilon - \gamma)}{40}$ fraction of the seeds cause the circuit to output $1$; whereas if $x \notin L$, then at most a $\frac{1}{2} - \epsilon + \frac{9(2\epsilon - \gamma)}{40}$ fraction of the seeds make the circuit output $1$. 
Recall that the circuit accepts the input $P(q) \coloneqq w \circ y$ iff the walk as per $y$ from $\start[\pi\oplus w]$ reaches $\acc[\pi\oplus w]$ (in $\gmxtau[\pi \oplus w]$). Because all the $y_{\le i}$-DFS are small—specifically, they return to $\acczp{\pi\oplus w}$ within $2^{10d\cdot s}$ steps—\autoref{reversiblewalk} guarantees that for each $i\in [T\cdot 2^{4s}]$, every vertex in the $y_{\le i}$-tree($\acczp{\pi\oplus w}$) is visited at least once. Consequently, we can equivalently state that the circuit accepts the input $w \circ y$ if and only if the $y_{\le i}$-DFS starting from $\acczp{\pi\oplus w}$ reaches $\conf{\start[\pi\oplus w]}{i}$ for some $i\in [T\cdot 2^{4s}]$. Hence, it suffices to check the proportion of accepting pairs $(w,y)$ and accept if and only if this fraction is at least $\frac{1}{2}-\epsilon + \frac{27(2\epsilon - \gamma)}{40}$.



\end{proof}

\begin{observation}\label{correctnesstwo}
In case $\mathcal{B}$ reaches step \ref{alltargetscompressed} and frees up $c$ bits, this is enough to decide the input using $O(s)$ work space.
\end{observation}
\begin{proof}
Assume that there exists a catalytic setting $\pi$ such that the size of $y_{\le i}$-DFS from $\acczp{\pi \oplus w}$ is small for all $w,y,i$ (where $w \circ y$ denotes the output of $P$ for a seed and $i\in [T\cdot 2^{4s}]$). We can use the freed-up $c$ bits on the catalytic tape to find such a $\pi$ by brute-force, and then use step \ref{caseallsmall} to correctly decide the input; where the correctness follows from observation \ref{correctnessone}. This clearly takes $O(s)$ work space. 
Now we argue that such a $\pi$ exists.

From the proof of Observation \ref{avgsize}, we know that $\gmx$ has at most $2^{c+4s}$ configurations. Hence, its layered configuration graph (with $T\cdot 2^{4s}$ layers) contains at most $T\cdot 2^{c+8s}$ configurations in total. Fix a seed $q$ of $P$ such that $P(q) = w \circ y$ and fix some $i \in [T\cdot 2^{4s}]$. For any two distinct $\tau',\tau''$, the vertex sets of $y_{\le i}$-tree($\acczp{\tau' \oplus w}$) and $y_{\le i}$-tree($\acczp{\tau'' \oplus w}$) are disjoint. The reason is that every configuration in the first tree reaches the unique sink vertex $\acczp{\tau' \oplus w}$ under a $y_{\le i}$-DFS, while every configuration in the second tree reaches the unique sink $\acczp{\tau'' \oplus w}$. Consequently, the expected size (over a uniformly random choice of $\pi'$) of the vertex set of $y_{\le i}$-tree($\acczp{\pi' \oplus w}$) is at most 
$T\cdot \frac{2^{c+8s}}{2^c} = T\cdot 2^{8s}$.  Let us choose $\pi'$ uniformly at random. Then we have


\begin{align*}
&\Pr[\exists i \in [T\cdot 2^{4s}], q \text{ such that } P(q)\coloneqq w \circ y \text{ and } y_{\le i}\text{-DFS} \text{ from $\acczp{\pi' \oplus w}$ has size } \ge 2^{10d\cdot s}] \\
& \le  \sum_{q,i}\Pr[P(q)\coloneqq w \circ y \text{ and } y_{\le i}\text{-DFS} \text{ from $\acczp{\pi' \oplus w}$ has size } \ge 2^{10d\cdot s}] \\
&\le \sum_{q,i} \frac{T\cdot 2^{8s}}{2^{10d\cdot s}} \le \sum_{q,i} \frac{1}{2^{8d\cdot s}} \le \frac{2^{ds}\cdot T \cdot 2^{4s}}{2^{8d\cdot s}} \le \frac{1}{2^{6d\cdot s}} \le \frac{1}{2^{30s}} < 1
\end{align*}

where the first inequality follows from the union bound; the second from the Markov inequality; and the subsequent ones from the fact that $T$ is a constant, $d\ge 5$, and the size of the seed $q$ is at most $d\cdot s$ bits. Thus, there is a non-zero probability of ending up with a $\pi'$ such that all the corresponding $y_{\le i}$-DFS are small.
\end{proof}


\section{Proof of  
\autoref{lemmvianisansprggeneral} and \autoref{lemmvianisansprgforcl}}\label{nisansprgproofs}


We will begin by proving \autoref{lemmvianisansprggeneral}. The proof of \autoref{lemmvianisansprgforcl} will follow from this and will be discussed in \autoref{lemmapart2}. Let \(\delta\) and \(\epsilon\) be constants such that \(\delta < 2\epsilon\). Without loss of generality, we can assume that both \(\delta\) and \(\epsilon\) are rational numbers. Initially, we will assume that \(\delta > 0\) and will address the case where \(\delta = 0\) separately in \autoref{deltazero}. 

Let \( L \in \BPepsCdeltSpace{s}{c}{\delta}{\epsilon} \), and let \( M \) be the machine that decides it. Let \( x \) be an arbitrary fixed input. We will describe how the subroutine \( \mathcal{F}^L \) operates. There are various scenarios in which \( \mathcal{F}^L \) compresses the catalytic tape, freeing at least the last bit in the process (as we will demonstrate, we will free more than this, up to \( O(s) \) bits). For each of these scenarios, we will outline the method for decompressing the tape. Additionally, we will be able to determine which decompression method to apply based on a given compressed tape. This describes $\mathcal{R}^L$.

The structure of the catalytic tape used by $\mathcal{F}^L$ is depicted in \autoref{tapeappmab}. The tape starts with $c$ cells denoted by $\tau$ (which will be used to simulate the catalytic tape of $M$), followed by $\pmb{l}$ hash functions, each specified using $\pmb{2m}$ bits;  followed by a \textit{target} string $\pmb{\mathrm{tar}}$ of length $\pmb{3m}$.

\begin{figure}[h!]
\begin{center}
    \begin{tikzpicture}[every node/.style={font=\small}]
  \def\boxwidth{2.5}
  \def\boxheight{1}
  \def\braceYOffset{0.2}  
  \def\arrowYOffset{0.4}  
 \foreach \i/\name in {0/$\tau$} {
    \draw (\i*\boxwidth,0) rectangle ++(\boxwidth,\boxheight);
   \node at (\i*\boxwidth+0.5*\boxwidth,0.5*\boxheight) {\name};
  }




  \def\bottomy{0}
  \def\bottomwidth{1.5}
  \foreach \i/\name in {0/$h_l$,1/$h_{l-1}$,2/\dots,3/$h_1$,4/$\pmb{\mathrm{tar}}$} {
    \draw (\i*\bottomwidth+\boxwidth, \bottomy) rectangle ++(\bottomwidth,\boxheight);
    \node at (\i*\bottomwidth+\boxwidth+0.5*\bottomwidth, \bottomy+0.5*\boxheight) {\name};
  }

  \foreach \i in {1,2,3,4} {
    \draw (\i*\bottomwidth+\boxwidth, \bottomy) -- ++(0, \boxheight);
  }

\end{tikzpicture}
\caption{Catalytic tape}
\label{tapeappmab}
\end{center}
\end{figure}

The hash functions on the catalytic tape are derived from a 2-independent hash family \(\mathcal{H} = \{h: \{0,1\}^m \rightarrow \{0,1\}^m\}\) with a size of \(2^{2m}\), as in Nisan's PRG \cite{Nisan92}; where each hash function can be computed in \(O(m)\) space.

\subsection{The parameters}\label{paramters}

The parameters of interest are:

\begin{itemize}
\item $\pmb{m}$ the seed length of the PRGs that will be built, which we will set later. We can think of $m=d\cdot s$, for a large constant $d \ge 500$.
\item  $\pmb{l}$ the number of hash functions; with each specified by $2m$ bits. We set $l= 2^{20s}$. 
\item length of the string $\pmb{\mathrm{tar}}$, which we set to $3m$.
\item $\pmb{H}$ (called upper-threshold), which we set to $2^{3m}$.
\item $\pmb{T}$ (called lower-threshold), which we set to $2^{100s}$. 
\end{itemize}

Thus, the size of the tape used by $\mathcal{F}^L$ is \[c+l\cdot2m + 3m \le 2^s + O(s)\cdot 2^{20s} + O(s) = 2^{O(s)}\] as desired.

\subsection{Building PRGs}

The broad idea is as follows. $\mathcal{F}^L$ will build a sequence of pseudo-random generators: $\text{PRG}_0,\text{PRG}_1,\cdots,\text{PRG}_l$, where PRG$_0$ is defined to output the empty string for every seed. For $1 \le i \le l$, PRG$_i$ outputs a string of length $i$ and is defined as

\begin{equation}\label{nsprgdefn}
\text{PRG}_i(seed) \coloneqq h_i(seed) \circ h_{i-1}(seed) \circ \cdots \circ h_1(seed) 
\end{equation}

where the hash functions are read from the catalytic tape. Here, we abuse the notation and use $h_i(seed) \circ h_{i-1}(seed) \circ \cdots \circ h_1(seed)$ to denote the $i$-bit length string obtained by concatenating the first bits outputted by each of the $i$ hash functions. Note that since the hash functions can be computed in space $O(m)=O(s)$, we can compute these PRGs using $O(s)$ work space, since $l$, the number of hash functions, is $2^{O(s)}$. We would like PRG$_i$ to approximate the probability of going from \textit{any} configuration in layer $i$ in the layered configuration graph of $\gmx$ (recall \autoref{layeredconfgrph}) to $\accz$ and $\rejz$. To be precise, we want $\gamma_i^{acc}$ and $\gamma_i^{rej}$ to be ``small'', where

\begin{align}\label{errorpromise}
\begin{split}
&\gamma_i^{acc} \coloneqq \max_{\confl{\pi}{u}{i}}\bigg\{\big | \Pr[\text{Going from } \confl{\pi}{u}{i}  \text{ to } \accz] \;-  \\ 
& \Pr_{seed}[y\text{-DFS} \text{ from } \accz \text{ visits } \confl{\pi}{u}{i} \text{ for } y\leftarrow\text{PRG}_i(\text{seed})] \big|\bigg\} 
\end{split}
\end{align}
and $\gamma_i^{rej}$ is defined analogously.
Here, $\confl{\pi}{u}{i}$ denotes an \text{arbitrary} configuration in layer $i$. In the equation above, the first probability refers to the likelihood of transitioning from \(\confl{\pi}{u}{i}\) to \(\accz\) in the layered configuration graph of $\gmx$. The second probability is taken over a uniformly random seed for PRG$_i$, such that the $y$-DFS (see \autoref{ydfs}) starting from $\accz$ visits $\confl{\pi}{u}{i}$, where $y$ denotes the output of the PRG$_i$ on a seed. Note that, by definition, the $y$-trees (see \autoref{ytree}) $y$-tree($\accz$) and $y$-tree($\rejz$) for an empty string $y$ are simply the isolated vertices $\accz$ and $\rejz$. Additionally, using \autoref{layeredconfgrph}, we know that vertices within the same layer do not have edges between them. Therefore, $\gamma_0^{acc}=\gamma_0^{rej}=0$.

We define PRG$_0$ to be \emph{good}, and we inductively define PRG$_i$ to be good for $1 \le i \le l$ if PRG$_{i-1}$ is good and $h_i$ is \emph{good}. We will define what it means for a hash function to be good later. We will describe $\mathcal{F}^L$ in an iterative fashion by assuming that for every $j \le i$ (for some $0 \le i < l$), PRG$_j$ is good and that $\mathcal{F}^L$ has not yet compressed the catalytic tape. We will break down this description into various cases. As we will see, $\mathcal{F}^L$ may compress the catalytic tape if $h_{i+1}$ is considered bad, among other possible scenarios, and then stop. Otherwise, it will iterate to the next value of $i$ and repeat the same process. However, if $i = l$, $\mathcal{F}^L$ will either be able to determine the input $x$ or will output $\perp$. We assume at this point the first $c$ bits of the catalytic tape read $\tau$ (as in \autoref{tapeappmab}).

\phantomsection
\subsubsection*{Case $1$: Size $> H$.}\label{firstcase}

\begin{definition}[Size of $y$-DFS]\label{sizeofydfs}
For a $y$-DFS from $\accz$/$\rejz$, we define the size of the $y$-DFS to be the number of steps required to return to the halting configuration from which we initiated the $y$-DFS. In other words, it is the size of the $y$-tree\textup{(}$\accz$\textup{)}/$y$-tree\textup{(}$\rejz$\textup{)}.
\end{definition} 

The first scenario occurs when the $y$-DFS from either $\accz$ or $\rejz$ has a size that is greater than  $H=2^{3m}$ for some value of $y$ (the output of PRG$_i$ on a seed). In this situation, we will use timestamp compression \cite{CookLiMertzPyne25} to compress the target string. We will check the size of a \(y\)-DFS (starting from \(\accz\) or \(\rejz\)) using the subroutine 
$\countsize$  from \autoref{reversiblewalk} with the size parameter $H$.


    

\textbf{Compression}. By trying all the \(y\)-DFS (from \(\accz\) and \(\rejz\)), we can check if any of them has a size \(> H\), while reusing the work space required for \(\countsize\). Without loss of generality, let the seed for which the \(y\)-DFS is large be \(q\) (i.e., \(y \leftarrow \text{PRG}_i(q)\)), and assume it was performed from \(\accz\). Let \( val(\textbf{tar}) \) represent the integer value of the binary string \( \textbf{tar} \) incremented by one. Thus, \( val(\pmb{\mathrm{tar}}) \) is an integer in \([2^{3m}]\). 
Let \( \confl{\pi}{u}{k} \) be the configuration reached by performing a \( y \)-DFS from \( \accz \) for \( val(\pmb{\mathrm{tar}}) \) steps using $\nextstep$   from \autoref{reversiblewalk}. 
We compress by replacing the contents of the string \( \pmb{\mathrm{tar}} \) with \( 00 \circ q \circ k \circ i \circ u \circ 0^{2m-41s-2} \). Note that \( q \) takes \( m \) bits, \( 0 \le k \le i \le l \le 2^{20s} \), and each can be described using \( 20s \) bits, while \( u \) (work space description) takes \( s \) bits. We also replace \( \tau \) with \( \pi \) in the first \( c \) cells of the catalytic tape. The prefix \( 00 \) indicates that we applied compression under the scenario ``size \( > H \)'' (since there will be other compression scenarios). Thus, we compress, freeing at least the last $s$ bits of the catalytic tape, given that  \( m \ge \frac{42s+3}{2} \).

\textbf{Decompression}. To decompress, we perform a reverse $y$-DFS from $\confl{\pi}{u}{k}$, using $\stepback$ from \autoref{reversiblewalk}, counting the number of steps required to reach the unique halting configuration ($\accz$ or \(\rejz\)). This count is equal to $val(\pmb{\mathrm{tar}})$. Therefore, we recover $\pmb{\mathrm{tar}}$ and reset the first $c$ bits of our catalytic tape back to $\tau$ upon reaching $\accz$. 



\subsubsection*{Case $2$: Size $\le H$}\label{secondcasesubsection}
We now address the scenario in which every $y$-DFS according to PRG$_i$ (from $\accz$ or $\rejz$) has a size of at most $H$. 

\begin{definition}\label{videfn}
We define \( V_i \) as the set of vertices, i.e., configurations \( \conf{v}{i} \) in layer \( i \) such that the following condition holds: 

\[
\Pr_{seed}[y\text{-DFS from } \accz \text{ or } \rejz \text{ visits } \conf{v}{i} \text{ for } y \leftarrow \text{PRG}_i(seed)] > 0
\]

In other words, \( V_i \) is the collection of vertices in layer \( i \) that have a non-zero probability (over seed) of being reached via a \( y \)-DFS, as determined by \( \text{PRG}_i \), starting from \( \accz \) or \( \rejz \).
\end{definition}

\begin{definition}\label{sidefini}
We define \( S_i \) as the set of vertices, i.e., configurations \( \conf{v}{i} \) in layer \( i \) such that the following holds: 

\[
\Pr_{\text{seed}}[y\text{-DFS from } \accz \text{ or } \rejz \text{ visits } \conf{v}{i} \text{ for } y \leftarrow \text{PRG}_i(\text{seed})] \ge \frac{1}{2^{30s}}.
\]

In other words, \( S_i \) is the set of vertices in layer \( i \) that can be reached from \( \accz \) or \( \rejz \) for at least \( \frac{1}{2^{30s}} \) fraction of the \( y \)-DFS, as determined by PRG\(_i\).
\end{definition}

Note that by definition, $S_i\subseteq  V_i$. 

If \( |S_i| \ge T \), we will proceed to subcase $2.1$; otherwise, we will move to subcase $2.2$. Before diving into these subcases, we will describe some subroutines that will help us check this condition and will also be useful in subsequent steps. These subroutines use the first \( c \) bits of the catalytic tape to simulate that of \( M \).

\begin{itemize}
    \item \textbf{$\countvi$}. Computes $|V_i|$.
    \item \textbf{$\indexvi$}. Given \( j \in [\countvi] \), the procedure outputs a seed \( q \), an integer \( t \geq 0 \), and a binary value \( b \in \{0, 1\} \). This indicates that the \( j \)-th vertex in \( V_i \) (according to a specific numbering) can be reached by performing a \( y \)-DFS from \( \accz \) for \( t \) steps if \( b = 1 \). Conversely, if \( b = 0 \), the \( j \)-th vertex can be reached by performing the \( y \)-DFS from \( \rejz \) for \( t \) steps, where \( y \) is given by \( \text{PRG}_i(q) \). Moreover, $t$ is upper-bounded by the size of the corresponding \( y \)-DFS.
    \item\label{countsidescrptionmention} \textbf{$\countsi$}. Computes $|S_i|$. 
    \item\label{indexsidescriptionmention} \textbf{$\indexsi$}. Given \( j \in [\countsi] \), it works analogously to $\indexvi$. 
\end{itemize}

\begin{claim}\label{subroutineclaim}
Given that all the $y$-DFS from $\accz$ and $\rejz$, as per PRG$_i$, have sizes of at most $B=2^{O(s)}$\footnote{where $B$ is logspace-constructible.}, the subroutines $\countvi$, $\countsi$, $\indexvi$, and $\indexsi$ can be executed using $O(s)$ work space. Furthermore, these subroutines run in time $2^{O(s)}$. Additionally, they always restore the first $c$ bits of the catalytic tape they utilize, to their original contents.
\end{claim}
\begin{proof}
First, given $i$, a bit $b$, seed $g$, $t\le B$, and $k\in [c+O(s)]$, we can determine the $k$-bit of the vertex reached by procedure $\nextstep$ (from \autoref{reversiblewalk}) repeated $t$-times starting from $\accz$ if $b=1$ or $\rejz$ if $b=0$ in the $y$-DFS, where  \( y \leftarrow \text{PRG}_i(g) \).
Then, by running $\stepback$ $t$-times, we can also restore the catalytic tape.
This uses $O(s)$ work space and time $2^{O(s)}$. Hence, for any two given triples $(b,g,t)$ and $(b',g',t')$, we can determine whether they refer to the same vertex by comparing the two vertices bit by bit.

To count the size of $V_i$, we just need to determine how many triples $(b,g,t)$ give distinct vertices. 
To do so, we initialize a counter $\pmb{count}$ to $0$ and then
enumerate over triples $(b,g,t)$ in lexicographical order and check for each of them whether some lexicographically smaller triple gives the same vertex.
If not, and the vertex given by $(b,g,t)$ is in layer $i$ (which can be determined from its bit description), we increment $\pmb{count}$.
In either case, we proceed to the next triple $(b,g,t)$. Clearly, this counts the size of $V_i$. 
This can be implemented using a work space of size $O(s)$ via two nested for-loops over the triples.
The catalytic space is restored to its initial content, and there are no other modifications to the space other than the internal work space. To identify the triple $(b,g,t)$ corresponding to a given index $j \in [\countvi]$, 
we can run the same procedure as above but terminate it once  $\pmb{count}$ reaches $j$.
At that point, we have the triple $(b,g,t)$ on our tape. This implements the two procedures $\countvi$ and $\indexvi$.

To determine whether a given triple $(b,g,t)$ refers to a vertex from $S_i$ 
we need to check for how many seeds $g'$ there is some $t'\le B$ and $b' \in \{0,1\}$,
so that $(b,g,t)$ and $(b',g',t')$ refer to the same vertex.
Again, this condition can be checked using internal work space $O(s)$. By going over $j \in [\countvi]$, finding its corresponding triple $(b,g,t)$ using $\indexvi$
and testing whether it belongs to $S_i$ we can count the size of $S_i$ using internal work space $O(s)$. This implements $\countsi$. To implement $\indexsi$, we can output $(b,g,t)$ that is reached when the counter to compute $\countsi$ reaches the desired value.
\end{proof}

Since all the $y$-DFS from $\accz$ and $\rejz$ have sizes of at most $H = 2^{3m} = 2^{O(s)}$, using the previous claim, we can utilize $\countsi$ to check if \( |S_i| \ge T = 2^{100s} \) using \( O(s) \) work space and in \( 2^{O(s)} \) time without altering the contents of the catalytic tape.

\phantomsection\label{case2pointone} \paragraph{\textbf{Case $\pmb{2.1}$: $\pmb{|S_i| \ge T}$}.}

In this scenario, we again compress the target string. We use \(val(\pmb{\mathrm{tar}}_{\le \log{T}})\) to denote the integer value of the first \(\log{T}\) bits of \(\pmb{\mathrm{tar}}\), incremented by \(1\), which is an integer in $[T]$. Here, we assume that \(3m \ge \log T = 100s\). Let $\confl{\pi}{u}{i}$ be the vertex/configuration in $S_i$ with index $val(\pmb{\mathrm{tar}}_{\le \log{T}})$, as per the subroutine $\indexsi$.

Let us consider performing a $y$-DFS starting from $\confl{\pi}{u}{i}$ for each $y$ generated by PRG$_i(seed)$, using all possible seeds, with each run limited to $H$ steps. For every $y$-DFS, we will visit at most one (distinct) configuration in \textbf{layer $0$}, which will be the unique sink of the corresponding $y$-tree. Let TAU represent the set of all $c$-bit catalytic settings such that $\pi' \in \text{TAU}$ iff, for at least $\frac{1}{2^{30s}} 2^m$ seeds, the $y$-DFS from $\confl{\pi}{u}{i}$ visits either $\conf{\acc[\pi']}{0}$ or $\conf{\rej[\pi']}{0}$, within $H$ steps. Then, it follows that $|\text{TAU}| \leq 2^{30s}$. We know that for at least $\frac{1}{2^{30s}}$ fraction of the seeds, the $y$-DFS from $\accz$ or $\rejz$ visits $\confl{\pi}{u}{i}$, since this configuration is in $S_i$. Let us denote this set of seeds as $I$—the set of seeds for which the $y$-DFS from $\accz$ or $\rejz$ reaches $\confl{\pi}{u}{i}$. By definition, we have $|I| \geq \frac{1}{2^{30s}} 2^m$. Since we are in a scenario where all $y$-DFS from $\accz$ and $\rejz$ have a size no greater than $H$, by applying \autoref{reversiblewalk}, we can conclude that we will visit $\accz$ or $\rejz$ from $\confl{\pi}{u}{i}$ while performing a $y$-DFS within $H$ steps, with $y$ being the output of PRG$_i$ for seeds in $I$. This implies that $\tau \in \text{TAU}$. We will index the catalytic settings $\pi'$ in TAU according to the lexicographic order of the first seed for which $\conf{\acc[\pi']}{0}$ or $\conf{\rej[\pi']}{0}$ is visited while performing a $y$-DFS from $\confl{\pi}{u}{i}$ (for $H$ steps).

\textbf{Compression}. First, we run the subroutine $\indexsi$ with input $val(\pmb{\mathrm{tar}}_{\le \log{T}})$. This gives a seed $q$, an integer $t \ge 0$, and a binary value $b$ (which we, without loss of generality, assume to be $1$). Since all the $y$-DFS are of size at most $H$, it follows from \autoref{subroutineclaim} that we can save $q,t$ using $O(s)$ work space. This allows us to visit the configuration $\confl{\pi}{u}{i}$, which changes the contents $\tau$ on the catalytic tape to $\pi$. With $q$ and $t$ stored in our work space, we can always return to $\accz$ by performing a $y$-DFS from this point, corresponding to the seed $q$. Effectively, we are now “standing” at the configuration $\confl{\pi}{u}{i}$. Let the index of $\tau$ in TAU be denoted as $idx$; therefore, $idx$ can be specified using $30s$ bits. We will outline how to compute $idx$.

\begin{enumerate}[(a)]
    \item For any seed, we can check if the $y$-DFS from $\confl{\pi}{u}{i}$ visits a halting configuration in layer $0$ within $H$ steps and then walk back to $\confl{\pi}{u}{i}$ by performing the $y$-DFS in reverse order, all the while keeping track of the number of steps we walk in any direction. If we see a halting configuration in layer $0$, we will, for ease, call the corresponding seed \emph{halting}.
    \item For any two halting seeds \(seed_1\) and \(seed_2\) (which can be verified using $(a)$), we can determine if both the \(y_1\)-DFS and \(y_2\)-DFS—where \(y_1 \leftarrow \text{PRG}_i(seed_1)\) and \(y_2 \leftarrow \text{PRG}_i(seed_2)\)—encounter a halting configuration in layer 0 with the same catalytic contents. To do this, we first modify point $(a)$ so that, in addition to checking if a seed is halting, it also reports the \(j\)-th bit of the catalytic part of the halting configuration that is observed, for any index \(j \in [c]\). We can then perform a bit-by-bit comparison of the catalytic parts of the halting configurations from the two \(y\)-DFS, while continuously returning to \(\confl{\pi}{u}{i}\). If the catalytic contents match, we classify the two seeds as being \emph{similar}.
    \item For any halting seed, using $(b)$, we can count how many seeds are similar to it (including itself). If there are at least $\frac{1}{2^{30s}}2^m$, we call the seed \emph{relevant}.
    \item \label{figoureoutindx} We keep a counter $\pmb{count}$ initialized to $0$. For each seed, we check if it is halting, relevant, and not similar to any seed that is lexicographically smaller than it. If so, we increment $count$. Let us suppose we are at a seed for which the counter was increased. Since it's a halting seed, let the contents of the halt configuration seen be $\pi'$. Then, by definition, $\pi' \in$ TAU (since the seed is relevant), and the value of $count$ at this point is its index in TAU (since it is not similar to any lexicographically smaller seed).
    \item Finally, we compute \( idx \) as follows. Let \( seed_1 \) be any seed for which we have incremented $count$. Let the catalytic contents of the halt configuration observed during a \( y_1 \)-DFS from \(\confl{\pi}{u}{i}\), using \( y_1 \leftarrow \text{PRG}_i(seed_1) \), be  \( \pi' \). We can verify whether \( \pi' \) is identical to \( \tau \) again in a bit-by-bit manner, similar to the method described in part $(b)$. This verification is feasible because we have \( q \) and \( t \) stored, which allows us to move back and forth between \( \accz \) and \(\confl{\pi}{u}{i}\). If $\pi'=\tau$, we declare $idx$ as the value of $count$.
\end{enumerate}

 To compress, we replace the first $c$ cells of the catalytic tape with $\pi$ (recall we are effectively standing at $\confl{\pi}{u}{i}$); and we replace $\pmb{\mathrm{tar}}_{\le \log{T}}$ with $01 \circ idx \circ i \circ u \circ 0^{49s-2}$. This is feasible as $\log{T} = 100s$ and 
    $$2 + |idx| + |i| + |u| \le 2 + 30s + 20s + s = 51s +2$$
where we used that $i \le l = 2^{20s}$ and thus can be specified using $20s$ bits. The leading $01$ specifies this type of compression.
    
\textbf{Decompression}. Given the compressed tape, we find ourselves in the configuration $\confl{\pi}{u}{i}$. Knowing the value of $i$ provides us access to PRG$_i$. By using step \ref{figoureoutindx} of the compression process, we can determine a seed $q$ such that the $y$-DFS from $\confl{\pi}{u}{i}$, for $y \leftarrow \text{PRG}_i(q)$, will visit either $\accz$ or $\rejz$, within $H$ steps. Without loss of generality, we can assume it visits $\accz$. To be precise, $q$ is the seed for which the counter $count$ in step \ref{figoureoutindx} of the compression becomes equal to $idx$, which we know. As a result, we can effectively recover the $\tau$ part of the catalytic tape. Furthermore, we can move back and forth between the configurations $\confl{\pi}{u}{i}$ and $\accz$. We are now left with the task of recovering the contents of $\pmb{\mathrm{tar}}$. For this, we utilize the subroutines: $\countsi$ and $\indexsi$. For each \( j \in [\countsi] \), we invoke the subroutine \( \indexsi \), which provides us with a seed \( q' \), an integer \( t \ge 0 \), and \( b \in \{0, 1\} \) (without loss of generality, we can assume \( b = 1 \)). Consequently, we can move back and forth between \( \accz \) and the \( j \)-th vertex in \( S_i \), denoted as \( \confl{\pi'}{u'}{i} \). Thus, we can check in a bit-by-bit manner whether \( \confl{\pi'}{u'}{i} \) is the same as \( \confl{\pi}{u}{i} \)—that is, we verify if \( \pi = \pi' \) and \( u = u' \). Once we find such a \( j \), we know its value is \( val(\pmb{\mathrm{tar}}_{\le \log{T}}) \), which allows us to recover the compressed contents of \( \pmb{\mathrm{tar}} \) as well. 

We observe that, since we are in a situation where all the $y$-DFS from $\accz$ and $\rejz$ are small, using \autoref{subroutineclaim}, it can be verified that we can perform the compression and decompression processes using $O(s)$ work space and with time $2^{O(s)}$. Also note that we can always change the compression to free up the last bit(s) on the catalytic tape by swapping the freed-up $0^{49s-2}$ bits with the last $49s-2$ bits of $\pmb{\mathrm{tar}}$.


\paragraph{Case $2.2$: $\pmb{|S_i| < T}$.}\label{sigreatethantcompression}

In this scenario, we will check whether the hash function $h_{i+1}$ is good (defined shortly). If the hash function is bad, we will compress the catalytic tape again. Otherwise, we will move to the next iteration $i+1$. Before we define the goodness of a hash function, we will need the following definitions.


  




\begin{definition}\label{avbvdefn}
   Let $\conf{v}{i}$ and $\conf{z}{i+1}$ be two \textit{arbitrary} vertices/configurations in layers $i$ and $i+1$, respectively, of the layered configuration graph of $\gmx$. We introduce the following sets:
\begin{itemize}
    \item $A_{v}^{acc}=\{\text{seeds in } \{0,1\}^m \text{ for which $y$-DFS from $\accz$ visits $\conf{v}{i}$, for $y\leftarrow$ PRG$_i$(seed)}\}$. We similarly define $A_{v}^{rej}$, where the definition uses $\rejz$ instead of $\accz$.
    \item $B_{vz}= \{t\in \{0,1\}^m\ \mid \text{ there is an edge with label } t_1 \text { \textup{(}first bit of $t$\textup{)} from } \conf{z}{i+1} \text{ to } \conf{v}{i}\}$. Thus, $|B_{vz}|$ is either $0$, $2^{m-1}$, or $2^m$.
\end{itemize} 
\end{definition}


\begin{definition}[Neighborhood]
We define \( N_{i+1}(\cdot) \) as the function that identifies the neighborhood in layer \( i+1 \). Specifically, for any subset of vertices \( Y \) in layer \( i \), \( N_{i+1}(Y) \) denotes the set of vertices \( \conf{w}{i+1} \) in layer \( i+1 \) that have at least one outgoing edge directed toward a vertex in \( Y \). Moreover, we will use the notation \( N_{i+1}(\conf{v}{i}) \) to represent \( N_{i+1}(\{\conf{v}{i}\}) \), where \( \conf{v}{i} \) is a vertex in layer \( i \).
\end{definition}

Note that it follows from the definitions of $S_i,V_i$ (\autoref{videfn} and \autoref{sidefini}) that $N_{i+1}(S_i) \subseteq N_{i+1}(V_i)$.

\begin{definition}[$(A,B,\alpha)$-independence \cite{Nisan92}]\label{goodnesslemma}
Let $\mathcal{H}=\{h:\{0,1\}^m \rightarrow \{0,1\}^m\}$ be a pairwise independent hash family, and let $A,B \subseteq \{0,1\}^m$. A hash function $h:\{0,1\}^m \rightarrow \{0,1\}^m$ is called $(A,B,\alpha)$-independent if 
$$\big|\Pr_{x\in \{0,1\}^m}[x\in A \text{ and } h(x) \in B]-\rho(A)\rho(B)\big| \le \alpha$$ Here $\rho(A)=\frac{|A|}{2^m}$, and $\rho(B)=\frac{|B|}{2^m}$.
\end{definition}

\begin{definition}[Good hash function]\label{goodnessofhashfn}
A hash function from the hash family is considered \emph{good} if it is 
\begin{enumerate}
    \item $(A_v^{acc},B_{vz}, \alpha)$-independent for every $\conf{v}{i} \in S_i$ and $\conf{z}{i+1} \in N_{i+1}(\conf{v}{i})$.
    \item $(A_v^{rej},B_{vz}, \alpha)$-independent for every $\conf{v}{i} \in S_i$ and $\conf{z}{i+1} \in N_{i+1}(\conf{v}{i})$.
\end{enumerate}
where $\pmb{\alpha}$ is a parameter whose value will be set later.
\end{definition}

We now break down the analysis further into two sub-cases, depending on whether $h_{i+1}$ is good.

\subparagraph{Case $2.2.1$: $h_{i+1}$ is good.}
In this case, we increment $i$ and start over with \hyperref[firstcase]{Case $1$}. However, if $i$ becomes $l$, we will proceed to \hyperref[catch case]{Case $3$}. This assumes that we can determine whether \( h_{i+1} \) is good; we will discuss this in more detail in the next sub-case. Before we proceed, we need to analyze how the approximation errors \(\gamma_{i+1}^{acc}\) and \(\gamma_{i+1}^{rej}\), as described in \autoref{errorpromise}, grow. This analysis will help us determine the parameter \(\alpha\) and the seed length \(m\) later on. We will focus our analysis on \(\gamma_{i+1}^{acc}\), as the analysis for \(\gamma_{i+1}^{rej}\) is identical.

Consider an arbitrary vertex $\conf{z}{i+1}$ in layer $i+1$. Let the two outgoing transitions from $\conf{z}{i+1}$ into layer $i$ go to vertices $\conf{a}{i}$ and $\conf{b}{i}$ (for the ease of analysis, we assume $a \not =b$; though they can be the same). Then, we know that:

\begin{align*}
\begin{split}
&\Pr[\text{Going from } \conf{z}{i+1} \text{ to } \accz] = \\ &\frac{1}{2}\Pr[\text{Going from } \conf{a}{i} \text{ to } \accz] + \frac{1}{2}\Pr[\text{Going from } \conf{b}{i} \text{ to } \accz]
\end{split}
\end{align*}

whereby the first probability refers to the probability of transitioning from $\conf{z}{i+1}$ to $\accz$ in the layered configuration graph of $\gmx$ (see \autoref{layeredconfgrph}); the other probabilities are similarly defined. 

\begin{remark}[Notation]
When we write a quantity $q \pm \varepsilon$, we mean a value that differs from $q$ by at most $\varepsilon$.
\end{remark}

Using \autoref{errorpromise}, \autoref{avbvdefn}, and the definition of $\rho(\cdot)$ as per \autoref{goodnesslemma}, we have

\begin{align}\label{simpleprobexpanded}
\begin{split}
\Pr[\text{Going from } \conf{z}{i+1} \text{ to } \accz] &= \frac{1}{2}\bigg(\rho(A_a^{acc}) \pm \gamma_i^{acc}\bigg) + \frac{1}{2}\bigg(\rho(A_b^{acc})\pm \gamma_i^{acc}\bigg)\\
&\frac{1}{2}\rho(A_a^{acc}) + \frac{1}{2}\rho(A_b^{acc}) \pm \gamma_i^{acc}
\end{split}
\end{align}

Note, by definition $B_{az}$ consists of strings starting with $0$, and $B_{bz}$ consists of strings starting with $1$ (or the other way around); therefore $B_{az} \cap B_{bz} = \phi$. Using the fact that  PRG$_{i+1}(seed)=h_{i+1}(seed) \circ \text{PRG}_{i}(seed)$ it follows from the definition of $y$-DFS (\autoref{ydfs}) that

\begin{equation}\begin{aligned}\label{secondboundhelp1}
&\Pr_{seed}[y\text{-DFS} \text{ from } \accz \text{ visits } \conf{z}{i+1}  \text{ for } y\leftarrow\text{PRG}_{i+1}(\text{seed})]  \\
&=\Pr_{seed}[seed \in A_a^{acc} \text{ and } h_{i+1}(seed) \in B_{az} ] + \Pr_{seed}[seed \in A_b^{acc} \text{ and } h_{i+1}(seed) \in B_{bz} ] \\
& \le \Pr_{seed}[seed \in A_a^{acc}] + \Pr_{seed}[seed \in A_b^{acc}] \\
& = \rho(A_a^{acc}) + \rho(A_b^{acc})
\end{aligned}\end{equation}

To simplify our notation, we define $actual = \Pr[\text{Going from } \conf{z}{i+1} \text{ to } \accz]$ and $approx = \Pr_{seed}[y\text{-DFS from } \accz \text{ visits } \conf{z}{i+1} \text{ for } y \leftarrow \text{PRG}_{i+1}(seed)]$. Hence, we are interested in upper-bounding $|actual-approx|$, which provides an upper-bound on $\gamma_{i+1}^{acc}$ since $\conf{z}{i+1}$ is an arbitrary vertex in layer $i+1$. We consider an exhaustive list of possibilities:

\begin{itemize}
\item $\pmb{\conf{z}{i+1} \not \in N_{i+1}(V_i)}$. Then, $\conf{a}{i},\conf{b}{i} \not \in V_i$. Thus, it follows from \autoref{videfn} and \autoref{avbvdefn}  that $\rho(A_a^{acc})=\rho(A_b^{acc})=0$. Therefore, $actual \le \gamma_i^{acc}$ by \autoref{simpleprobexpanded}. Moreover, by \autoref{secondboundhelp1} we get that $approx = 0$. Therefore, we have:
\begin{equation}\label{error_one}
\begin{split}
& |  actual - approx| \le \gamma_i^{acc}
\end{split}
\end{equation}

\item $\pmb{\conf{z}{i+1} \in N_{i+1}(V_i) \backslash N_{i+1}(S_i)} \textbf{ and } \pmb{\conf{a}{i},\conf{b}{i} \in V_i \backslash S_i}$. 
Using \autoref{simpleprobexpanded} we know that
\begin{equation}\label{secondboundhelp2}
\mid actual - \frac{1}{2}\big(\rho(A_a^{acc})+\rho(A_b^{acc})\big) \mid \; \le \gamma_i^{acc}
\end{equation}

Thus, we have:

\begin{equation}\begin{aligned}\label{errorboundtwo}
|actual - approx| &= |actual - \frac{1}{2}\big(\rho(A_a^{acc})+\rho(A_b^{acc})\big) +  \frac{1}{2}\big(\rho(A_a^{acc})+\rho(A_b^{acc})\big) -approx| \\
&\le  |\frac{1}{2}\big(\rho(A_a^{acc})+\rho(A_b^{acc})\big)-approx| + |actual - \frac{1}{2}\big(\rho(A_a^{acc})+\rho(A_b^{acc})\big)|  \\
&\le \frac{1}{2}(\rho(A_a^{acc})+ \rho(A_b^{acc})) + \gamma_i^{acc} \\
&\le \frac{1}{2^{30s}} + \gamma_i^{acc}
\end{aligned}\end{equation}

The second-to-last inequality follows from \autoref{secondboundhelp1}, the fact that \( \text{approx} \ge 0 \), and \autoref{secondboundhelp2}. The final inequality is grounded in the observation that \( \rho(A_a^{acc}) \) and \( \rho(A_b^{acc}) \) are both less than \( \frac{1}{2^{30s}} \). This conclusion follows from the definitions of \( S_i \) and the fact that \( \conf{a}{i} \) and \( \conf{b}{i} \) belong to \( V_i \setminus S_i \).
\item $\pmb{\conf{z}{i+1} \in N_{i+1}(V_i) \backslash N_{i+1}(S_i)} \textbf{ such that wlog. } \pmb{\conf{a}{i} \in V_i \backslash S_i}$ \textbf{ and } $\pmb{\conf{b}{i} \not \in V_i}$. Since $\conf{b}{i}\not \in V_i$, $\rho(A_b^{acc})=0$. \autoref{simpleprobexpanded} gives us that
$$
|actual - \frac{1}{2}\rho(A_a^{acc})|\le \gamma_i^{acc} 
$$
and \autoref{secondboundhelp1} gives that
$$
approx \le \rho(A_a^{acc})
$$
Therefore, similar to the previous case, we get
\begin{equation}
\begin{aligned}
|actual-approx| &\le \frac{1}{2}\rho(A_a^{acc}) + \gamma_i^{acc} \\
& \le \frac{1}{2}\cdot \frac{1}{2^{30s}} + \gamma_i^{acc}
\end{aligned}
\end{equation}
where the last inequality follows from the fact that $\conf{a}{i} \in V_i \backslash S_i$ and thus $\rho(A_a^{acc}) < \frac{1}{2^{30s}}$.
\item $\pmb{\conf{z}{i+1} \in N_{i+1}(S_i)}\textbf{ and }\pmb{\conf{a}{i},\conf{b}{i} \in S_i}$. Since by our assumption $a$ and $b$ are different configurations, $|B_{az}|=|B_{bz}|=2^{m-1}$. Thus we can rewrite \autoref{simpleprobexpanded} as

\begin{equation}\label{rewrite}
actual = \rho(B_{az})\rho(A_a^{acc}) + \rho(B_{bz})\rho(A_b^{acc}) \pm \gamma_i^{acc}
\end{equation}

Since $\conf{z}{i+1}\in N_{i+1}(\conf{a}{i})\cap N_{i+1}(\conf{b}{i})$ and both $\conf{a}{i}$ and $\conf{b}{i}$ are in $S_i$, and since we are in the case where $h_{i+1}$ is good, using \autoref{goodnessofhashfn} we obtain that

\begin{equation}\label{2ee}
\begin{aligned}
&|\Pr_{seed}[seed \in A_a^{acc} \text{ and } h_{i+1}(seed) \in B_{az}] - \rho(B_{az})\rho(A_a^{acc})| \le \alpha \\
&|\Pr_{seed}[seed \in A_b^{acc} \text{ and } h_{i+1}(seed) \in B_{bz}] - \rho(B_{bz})\rho(A_b^{acc})| \le \alpha
\end{aligned}
\end{equation}
Using \autoref{rewrite} and \autoref{2ee}, we get

\begin{align}\label{gettingtoitone}
\begin{split}
actual =\\ &\Pr_{seed}[seed \in A_a^{acc} \text{ and } h_{i+1}(seed) \in B_{az} ] \\ &+ \Pr_{seed}[seed \in A_b^{acc} \text{ and } h_{i+1}(seed) \in B_{bz} ]  \pm (\gamma_i^{acc} + 2\alpha)
\end{split}
\end{align}
Using \autoref{secondboundhelp1} and \autoref{gettingtoitone}, we get
\begin{equation}
|actual-approx| \le \gamma_i^{acc} + 2\alpha
\end{equation}
\item $\pmb{\conf{z}{i+1} \in N_{i+1}(S_i)} \textbf{ such that wlog. } \pmb{\conf{a}{i} \in S_i} \textbf{ and } \pmb{\conf{b}{i} \not \in V_i}$. As $\conf{b}{i} \not \in V_i$, $\rho(A_b^{acc})= 0$, and thus \autoref{rewrite} gives us that
$$\
actual = \rho(A_a^{acc})\rho(B_{az}) \pm \gamma_i^{acc}
$$
and using \autoref{secondboundhelp1}, we have
$$
approx = \Pr_{seed}[seed \in A_a^{acc} \text{ and } h_{i+1}(seed) \in B_{az}]
$$
As $\conf{z}{i+1} \in  N_{i+1}(\conf{a}{i})$ and $\conf{a}{i} \in S_i$, by the goodness of $h_{i+1}$, we have
$$
|\Pr_{seed}[seed \in A_a^{acc} \text{ and } h_{i+1}(seed) \in B_{az}] - \rho(B_{az})\rho(A_a^{acc})| \le \alpha
$$
Therefore, using the above, we get
$$
actual =  \Pr_{seed}[seed \in A_a^{acc} \text{ and } h_{i+1}(seed) \in B_{az}] \pm (\gamma_i^{acc} + \alpha)
$$
i.e.,
$$
|actual - approx | \le \gamma_i^{acc} + \alpha
$$
\item $\pmb{\conf{z}{i+1}  \in N_{i+1}(S_i)} \textbf{ such that wlog } \pmb{\conf{a}{i} \in S_i} \textbf{ and } \pmb{\conf{b}{i} \in V_i \backslash S_i}$. 
Let $Y = \big(\rho(B_{az})\rho(A_a^{acc}) + \rho(B_{bz})\rho(A_b^{acc})\big)$. We have that
\begin{align*}
|actual -approx | &= |actual - Y + Y - approx| \\
& \le |approx - Y| + |actual - Y| \\
&\le |approx - Y| + \gamma_i^{acc} \\
& =  \mid\Pr[seed \in A_a^{acc} \text{ and } h_{i+1}(seed) \in B_{az}] + \\&\quad+ \Pr[seed \in A_b^{acc} \text{ and } h_{i+1}(seed) \in B_{bz}] - Y\mid + \gamma_i^{acc} \\
& \le |\Pr[seed \in A_a^{acc} \text{ and } h_{i+1}(seed) \in B_{az}] - \rho(B_{az})\rho(A_a^{acc})| \\
&\quad + |\Pr[seed \in A_b^{acc} \text{ and } h_{i+1}(seed) \in B_{bz}] - \rho(B_{bz})\rho(A_b^{acc})| + \gamma_i^{acc} \\
&\le \alpha + |\Pr[seed \in A_b^{acc} \text{ and } h_{i+1}(seed) \in B_{bz}] - \rho(B_{bz})\rho(A_b^{acc})| + \gamma_i^{acc} \\
& =|\Pr[seed \in A_b^{acc} \text{ and } h_{i+1}(seed) \in B_{bz}] - \frac{1}{2}\rho(A_b^{acc})| + \alpha + \gamma_i^{acc}\\
&\le \max\{\Pr[seed \in A_b^{acc} \text{ and } h_{i+1}(seed) \in B_{bz}], \frac{1}{2}\rho(A_b^{acc})\} + \alpha + \gamma_i^{acc}\\
&\le \max\{\rho(A_b^{acc}), \frac{1}{2}\rho(A_b^{acc})\} + \alpha + \gamma_i^{acc}\\
&= \rho(A_b^{acc})  + \alpha + \gamma_i^{acc}\\
&\le \frac{1}{2^{30s}}  + \alpha + \gamma_i^{acc}
\end{align*}
where the second inequality follows from \autoref{rewrite}; the second equality from \autoref{secondboundhelp1}; the fourth inequality from the goodness of $h_{i+1}$; the third equality from the fact that $|B_{bz}|= 2^{m-1}$; fifth inequality from the fact that the terms in absolute difference are non-negative quantities; and the last inequality from the fact that $\conf{b}{i}\in V_i \backslash S_i$, and thus $\rho(A_b^{acc}) < \frac{1}{2^{30s}}$.
\end{itemize}
Thus, in summary, if \( h_{i+1} \) is good, we find that \( \gamma_{i+1}^{acc} \le \frac{1}{2^{30s}} + \gamma_i^{acc} + 2\alpha \). By considering \( \rejz \) instead of \( \accz \) and repeating the same analysis as above, one can show that the same recurrence relation holds for \( \gamma_{i+1}^{rej} \) as well.

\; \\ 
\textbf{Setting parameters $\alpha$ and $m$}. We pause our description of further cases to first establish the values for the seed length \( m \) and the parameter \( \alpha \) in \autoref{goodnessofhashfn} with foresight. In the scenario where \( \mathcal{F}^L \) does not compress the catalytic tape, and all the hash functions \( h_1 \) to \( h_l \) are good in their respective iterations, we would like to ensure that \( \gamma_l = \max\{\gamma_l^{acc}, \gamma_l^{rej}\} \)—which we refer to as the \emph{error} of PRG\(_l\)—is small. Since in this case all hash functions are good, we can use the recurrence $
\gamma_{i+1}^{acc} \le \frac{1}{2^{30s}} + \gamma_i^{acc} + 2\alpha 
$ (and the same for \( \gamma_{i+1}^{rej} \)). Since \( \gamma_0^{acc} = \gamma_0^{rej} = 0 \), we derive that \( \gamma_l \) is at most 

\begin{equation}\label{prgerrorapp}
l \cdot \left(2\alpha + \frac{1}{2^{30s}}\right)
\end{equation}

By setting $\alpha = \frac{1}{2l^2} = \frac{1}{2 \cdot 2^{40s}}$, we find from \autoref{prgerrorapp} that $\gamma_l$ is at most $\frac{1}{2^{20s}} + \frac{2^{20s}}{2^{30s}}$. This is at most $\frac{1}{2^{5s}}$ for sufficiently large $s$. The following lemma assists in bounding the fraction of bad hash functions, which will aid in compressing the catalytic tape if \( h_{i+1} \) is bad.

\begin{lemma}[Nisan \cite{Nisan92}]\label{actualnisanlemmmmes}
Let $\mathcal{H}=\{h:\{0,1\}^m \rightarrow \{0,1\}^m\}$ be a pairwise independent hash family, and let $A,B \subseteq \{0,1\}^m$. Then, the probability $h$ chosen from $\mathcal{H}$ uniformly at random is not $(A, B,\alpha)$-independent is less than or equal to $ \frac{\rho(A)\rho(B)(1-\rho(B))}{2^m\alpha^2}$.
\end{lemma}

\begin{corollary}\label{mostaregood}
The probability that a hash function (chosen uniformly at random from the family) is bad is upper-bounded by $\frac{2d_M\cdot T}{2^m \alpha^2}$, where $d_M$ is a constant that depends on the machine $M$.
\end{corollary}
\begin{proof}
In the layered configuration graph of $\gmx$, the number of incoming edges to any vertex or configuration is bounded by the constant $d_M$, which represents the number of edges incident to any vertex. Thus, given that we are in the case where \( |S_i| \leq T \), we can apply \autoref{actualnisanlemmmmes} and use the union bound to upper-bound the probability that a randomly chosen hash function from the family is not good by:
\begin{equation}\label{badupperbound}
\begin{aligned}
\frac{1}{2^m \alpha^2} \big ( 2d_M|S_i|\big) \le \frac{2d_M\cdot T}{2^m \alpha^2} 
\end{aligned}
\end{equation}
\end{proof}

We set \(m = 500s\). Thus, substituting the values of \(T\), \(\alpha\), and \(m\) in \autoref{mostaregood}, we can upper bound the fraction of hash functions that are not good (in a specific iteration) by 

\begin{equation}\label{fractionbad}
 \frac{8 \cdot d_M \cdot 2^{180s}}{2^{500s}}   
\end{equation}

which is at most \(\frac{1}{2^{200s}}\) for sufficiently large \(s\).

Note that our choice of seed length satisfies the constraints we encountered previously, as \(m > \max\left\{\frac{42s + 3}{2}, \frac{100s}{3}\right\}\). Additionally, with this seed length, we have \(T= 2^{100s} < H = 2^{3m} = 2^{1500s}\), ensuring that the thresholds make sense. 

We now resume the description of the further cases.

\subparagraph{Case $2.2.2$: $h_{i+1}$ is bad.}
We now consider the scenario when the hash function \( h_{i+1} \) is bad. In this case, we compress the catalytic tape by reusing the idea from Pyne \cite{Pyne25}, which was used to compress bad hash functions in the context of de-randomizing \(\BPL\) in a catalytic setting. 

We first assume that we can efficiently check whether a hash function from the family is good or bad and describe the compression scheme.

\textbf{Compression}. We iterate over all the hash functions from the family up to \( h_{i+1} \) in lexicographical order and count how many bad hash functions exist. This allows us to determine the index \( idx \) of \( h_{i+1} \) among the set of bad hash functions, which corresponds to the total count at the end of this process. Each hash function is represented by \( 2m = 1000s \) bits; therefore, we can iterate through them using \( O(s) \) space and in \( 2^{O(s)} \) time. Based on our choice of seed length, we know that the fraction of bad functions, as indicated in \autoref{fractionbad}, is at most \( \frac{1}{2^{200s}} \). Consequently, the index \( idx \) can be represented using \( 800s \) bits, which frees up \( 200s \) bits. In other words, we replace \( h_{i+1} \) with \( idx \circ 0^{200s} \) on the catalytic tape. As we cannot remember which hash function we compressed, we copy the contents of the first \(20s + 2\) bits of \(\pmb{\mathrm{tar}}\) (recall that the length of \(\pmb{\mathrm{tar}}\) is \(3m = 1500s\)), given by \(\pmb{\mathrm{tar}}_{\leq 20s + 2}\), into the freed-up \(200s\) bits. That is, we change the contents (in place of \(h_{i+1}\)) to \(idx \circ \pmb{\mathrm{tar}}_{\leq 20s + 2} \circ 0^{180s - 2}\). In place of the original $\pmb{\mathrm{tar}}_{\le 20s+2}$, we write $10\circ i$. This is feasible because $i \le l$, so it can be specified using $20s$ bits. Here, the $10$ in the beginning indicates this compression type. We further copy the last $180s-2$ bits of $\pmb{\mathrm{tar}}$, given by $\pmb{\mathrm{tar}}_{\ge 1320s +3}$, into the remaining free $180s-2$ bits. Thus, finally changing the contents in place of $h_{i+1}$ to $idx \circ \pmb{\mathrm{tar}}_{\leq 20s + 2} \circ \pmb{\mathrm{tar}}_{\ge 1320s +3}$; which effectively frees up the last $180s-2$ bits of $\pmb{\mathrm{tar}}$ (and thus the catalytic tape). Notice that in this compression, we do not change the first $c$ bits of the catalytic tape.

\textbf{Decompression}. We decompress as follows: We read the first \(20s + 2\) bits of \(\pmb{\mathrm{tar}}\), which gives us \(i\) and thus indicates the location on the catalytic tape where \(h_{i+1}\) was written. Consequently, we now have access to the contents \(idx \circ \pmb{\mathrm{tar}}_{\le 20s + 2} \circ \pmb{\mathrm{tar}}_{\ge 1320s + 3}\). We then loop through all hash functions again in lexicographical order. When we find a bad one with index \(idx\), we know it is \(h_{i+1}\). We copy \(\pmb{\mathrm{tar}}_{\le 20s + 2}\) and \(\pmb{\mathrm{tar}}_{\ge 1320s + 3}\) to their original locations and replace \(idx \circ \pmb{\mathrm{tar}}_{\le 20s + 2} \circ \pmb{\mathrm{tar}}_{\ge 1320s + 3}\) with \(h_{i+1}\).  

We now describe, in a step-by-step manner, how we can check if a hash function is bad:

\begin{enumerate}[(a)]
    \item Since we are in the sub-case where all the $y$-DFS from $\accz$ and $\rejz$ have a size of at most \( H = 2^{1500s} \), and where additionally \( |S_i| \leq T = 2^{100s} \), we can utilize the subroutines \(\countsi\) and \(\indexsi\) from \autoref{subroutineclaim}. This enables us to visit each configuration \(\conf{v}{i} \in S_i \). Furthermore, for any \( j \in [c+s] \), we can output the \( j \)-th bit in the description of \( v \) while ensuring that the subroutines reset the first \( c \) bits of the catalytic tape back to \( \tau \). According to \autoref{subroutineclaim}, this entire process can be accomplished using \( O(s) \) work space and requires \( 2^{O(s)} \) time. 
 \item For any \(\conf{v}{i}\) in \(S_i\) (which can be visited using $(a)$), we know that there are at most \(d_M\) (a constant) elements in \(N_{i+1}(\conf{v}{i})\). While at \(\conf{v}{i}\), we can take a peek at any of its neighbors \(\conf{z}{i+1}\) in layer \(i+1\). We can keep track of the local changes made to move from \(\conf{v}{i}\) to \(\conf{z}{i+1}\) using \(O(s)\) space, which can be undone. This allows us to effectively visit \(\conf{z}{i+1}\) and determine the label(s) for the edge from \(\conf{z}{i+1}\) to \(\conf{v}{i}\), which can be \(0\), \(1\), or both. Consequently, we can figure out if \(B_{vz}\) is the set of seeds with the first bit as \(0\) or \(1\) or the set of all possible seeds, and thus we can compute \(\rho(B_{vz})\), which will be either \(\frac{1}{2}\) or \(1\).
  
\item Now we describe how to check the goodness of a hash function. It suffices to demonstrate how to determine if a given hash function \( h \) is \((A_v^{acc}, B_{vz})\)-independent for some \(\conf{v}{i} \in S_i\) and \(\conf{z}{i+1} \in N_{i+1}(\conf{v}{i})\). The same approach can be used to check \((A_v^{rej}, B_{vz})\)-independence. We can then iterate through all possible \(\conf{v}{i}\) and \(\conf{z}{i+1}\) using steps $(a)$ and $(b)$ to check the goodness of \(h\).

We will keep two counters, both initialized to zero. Given that all \(y\)-DFS from \(\accz\) are of size at most $H$, we do the following for every seed (of PRG$_i$): We run the \(y\)-DFS from \(\accz\) for $H$ steps. For each configuration visited, we move back and forth via \(\accz\) to compare it bit-by-bit with \(\conf{v}{i}\). This allows us to check if the $y$-DFS from $\accz$ visits \(\conf{v}{i}\). If it does, we increase the first counter. If, for such a seed, \(h(seed) \in B_{vz}\), we increase the second counter as well. Notice that after executing the $y$-DFS for each seed, the value of the first counter, as per \autoref{avbvdefn}, is \(|A_v^{acc}|\). Meanwhile, the value of the second counter is \(2^m \Pr_{seed}[seed \in A_v^{acc} \text{ and } h(seed) \in B_{vz}]\). Consequently, we can check if $h$ is \((A_v^{acc}, B_{vz}, \alpha)\)-independent (as per \autoref{goodnesslemma}) for our value of $\alpha=\frac{1}{2\cdot 2^{40s}}$; using the value of the second counter, $|A_v^{acc}|$, and $|B_{vz}|$ computed using $(b)$.
\end{enumerate}

The hash functions from the family can be computed using space \( O(s) \), and the subroutines we employed also use space \( O(s) \) while requiring time \( 2^{O(s)} \). This allows us to efficiently determine whether a hash function is bad. Additionally, after performing the check, we reset the first \( c \) bits of the catalytic tape (used by the subroutines) back to the value \( \tau \).


\phantomsection
\subsubsection*{Case $3$: Can build upto PRG$_l$}\label{catch case}
The final case is where $\mathcal{F}^L$ does not achieve compression anywhere, and all the hash functions \( h_1 \) to \( h_l \) are good in their respective iterations.  In this case, $\mathcal{F}^L$ computes two fractions:

\begin{equation}\label{fractions}
\begin{aligned}
&f_{acc} \coloneqq \Pr_{seed}[y\text{-DFS from } \accz \text{ visits } \conf{\start}{l} \text{ for } y \leftarrow \text{PRG}_l], \\
&f_{rej} \coloneqq \Pr_{seed}[y\text{-DFS from } \rejz \text{ visits } \conf{\start}{l} \text{ for } y \leftarrow \text{PRG}_l].
\end{aligned}
\end{equation}

Here, $f_{acc}$ denotes the fraction of seeds for which the $y$-DFS from $\accz$ according to PRG$_l$, visits $\conf{\start}{l}$, while $f_{rej}$ is similarly defined. Let the constant $\zeta \coloneqq \frac{2\delta}{1+\frac{\delta}{2\epsilon}} < 2\epsilon$ (follows from the assumption that $\delta < 2\epsilon$). If $f_{acc} \ge \frac{1}{2} + \epsilon - \zeta$, then $\mathcal{F}^L$ accepts the input $x$. Conversely, if $f_{rej} \ge \frac{1}{2} + \epsilon - \zeta$, then $\mathcal{F}^L$ rejects. If neither condition is satisfied, it outputs $\perp$.

Note that all the $y$-DFS from $\accz$ and $\rejz$, as per PRG$_{l-1}$, have a size of at most $H$ in this situation. Otherwise, $\mathcal{F}^L$ would have compressed the catalytic tape. It follows that the size of each $y$-DFS from $\accz$ and $\rejz$, as per PRG$_l$, is at most $d_M H$, where $d_M$ is a constant that denotes the maximum number of edges incident to any vertex in the layered configuration graph of $\gmx$. Thus, to compute the fraction $f_{acc}$, we can do the following: For each seed (of PRG$_l$), we perform the $y$-DFS from $\accz$ for $d_MH$ steps. For each step, we check if the configuration visited is a starting configuration in layer $l$. If yes, we can move back and forth between the visited configuration and $\accz$ to check if the catalytic contents of the configuration are $\tau$. This way, we can count the number of seeds for which $y$-DFS from $\accz$ (as per PRG$_l$) visits $\conf{\start}{l}$, and thus compute $f_{acc}$. This can be done using $O(s)$ work space and in time $2^{O(s)}$. The fraction $f_{rej}$ can be computed similarly.

\subsection{Correctness}\label{correctness}
It can be verified that $\mathcal{F}^L$ (and the decompression algorithms, i.e., $\mathcal{R}^L$) use only $O(s)$ work space and run in time $2^{O(s)}$. Moreover, in all cases, $\mathcal{F}^L$ frees up the last bit on the catalytic tape, except when it ends up in \hyperref[catch case]{Case $3$}, where it outputs accept, reject, or $\perp$. We will first argue the correctness of these outputs.

Consider the case when $f_{acc} \ge \frac{1}{2}+\epsilon  -\zeta$. Using \autoref{prgerrorapp}, we know the error by PRG$_l$ is at most $\frac{1}{2^{5s}}$. That is,

\begin{equation}\label{useless0}
\Big|f_{acc} - \Pr[\text{Going from } \conf{\start}{l} \text{ to }\accz \text { in the layered graph of $\gmx$}] \Big| \le \frac{1}{2^{5s}}
\end{equation}

Thus, it follows from \autoref{layeredconfgrph}

\begin{equation*}
\Big|f_{acc} - \Pr[\text{Going from } \start \text{ to } \acc \text{ within } l \text{ steps in } \gmx]\Big| \le \frac{1}{2^{5s}}
\end{equation*}
Therefore, we get
\begin{align*}
&\Pr[M \text{ accepts $x$, given initial catalytic contents $\tau$}] \\
&\ge \Pr[\text{Going from } \start \text{ to } \acc \text{ within } l \text{ steps in } \gmx]\\
&\ge f_{acc} - \frac{1}{2^{5s}}\\
&\ge \frac{1}{2}+\epsilon -\zeta - \frac{1}{2^{5s}}\\ &= \frac{1}{2} -\epsilon + (2\epsilon - \zeta) - \frac{1}{2^{5s}}
\end{align*}
Since the constant $(2\epsilon - \zeta) > 0$, it holds for large enough $s$ that $\frac{1}{2^{5s}} < \frac{(2\epsilon - \zeta)}{2}$; thus, the probability is lower-bounded by $\frac{1}{2}-\epsilon + \frac{(2\epsilon-\zeta)}{2}$. But this means the input $x$ is in the language. This is because if it weren't the case, by the definition of $M$, the probability $M$ that accepts $x$ is $\le \frac{1}{2}-\epsilon$, leading to a contradiction. Similarly, if $f_{rej} \ge\frac{1}{2}+\epsilon-\zeta$, it implies $x$ is not in the language. Thus, whenever $\mathcal{F}^L$ outputs accept/reject, it is correct.

We are left to show that the fraction of initial catalytic tapes for which $\mathcal{F}^L$ outputs $\perp$ is upper-bounded by $\frac{1}{2^{4s}}$. Let $\tau$ be the contents of the initial $c$ bits of the catalytic tape used by $\mathcal{F}^L$. Then, recalling \autoref{taunodedefn} and the discussion from \autoref{proofepsdelderan}, we know that for $0 < \beta  < 1-\delta$ (this is valid as $\delta < 2\epsilon \le 1$), the $\tau^{\beta}$-graph consists of $\start$, at least one of $\acc$ and $\rej$, and does not contain $\acc[\tau']$/$\rej[ \tau']$, for $\tau \not = \tau'$. Taking $\beta$ to be

$$0< \beta = \frac{1}{4}\bigg(1-\frac{\delta}{2\epsilon}\bigg) < 1-\frac{\delta}{2\epsilon} \le 1-\delta$$

and using \autoref{tauepsp}, we get that the probability of leaving the $\tau^{\beta}$-graph to (null vertex) $\perp$, starting the walk from $\start$  is 

$$\le  \frac{\delta}{(1-\beta)} = \frac{2\delta}{1.5 + \frac{\delta}{4\epsilon}} \coloneqq \eta $$

That is with probability $\ge 1-\eta$ we reach $\acc$ or $\rej$ within (i.e., without leaving it) the $\tau^{\beta}$-graph, starting the walk from $\start$. Consider the case when $x$ is in the language. Then, by definition, $M$ rejects it with probability $\le \frac{1}{2}-\epsilon$. Thus, the probability of reaching $\acc$ from $\start$ within the $\tau^{\beta}$-graph is at least 

\begin{equation}\label{useless}
\begin{aligned}
    1-\bigg(\frac{1}{2}-\epsilon\bigg) - \eta &= \frac{1}{2} + \epsilon  -\eta \\
    &= \frac{1}{2}+\epsilon -\zeta + (\zeta-\eta)
\end{aligned}
\end{equation}

Similarly, when $x  \not\in L$, the probability of reaching $\rej$ from $\start$ within the $\tau^{\beta}$-graph is at least $\frac{1}{2}+\epsilon -\zeta + (\zeta-\eta)$. Plugging in the values for $\eta,\zeta$ we get that

\begin{align*}
(\zeta -\eta) &= \frac{2\delta}{1+\frac{\delta}{2\epsilon}} - \frac{2\delta}{1.5+\frac{\delta}{4\epsilon}} \\
& =2\delta\cdot \frac{0.5 -\frac{\delta}{4\epsilon}}{\big(1+\frac{\delta}{2\epsilon}\big)\big(1.5+\frac{\delta}{4\epsilon}\big)}\\
& > 0
\end{align*}
where the last inequality follows from the assumption that $2\epsilon > \delta > 0$. Finally, using \autoref{avgsize} we know that the average size of $\tau^{\beta}$-graph is at most $2^{4s}$. Thus, by Markov's inequality, the probability (over the uniform choice of $\tau$) that the $\tau^{\beta}$-graph has size $\ge 2^{9s}$ is at most $\frac{1}{2^{5s}}$. In other words, the probability (over $\tau$) that there exists a path of length $\ge 2^{9s}$ from $\start$ to $\acc$/$\rej$ within the $\tau^{\beta}$-graph is at most $\frac{1}{2^{5s}}$. Since $l=2^{20s}$, using \autoref{useless} and \autoref{taunodedefn}  we have that with probability (over $\tau$) $\ge 1-\frac{1}{2^{5s}}$:

\begin{align*}
x \in L &\Rightarrow  \\
&\Pr[\text{Going from } \start \text{ to } \acc \text{ within } l \text{ steps in } \gmx] \ge \frac{1}{2}+\epsilon - \zeta + (\zeta-\eta)    \\
x \not \in L &\Rightarrow \\
&\Pr[\text{Going from } \start \text{ to } \rej \text{ within } l \text{ steps in } \gmx] \ge \frac{1}{2}+\epsilon - \zeta + (\zeta-\eta) 
\end{align*}

Since $(\zeta-\eta)> 0$, for large enough $s$, $\frac{1}{2^{5s}} \le \frac{(\zeta-\eta)}{2}$. Since PRG$_l$ has an error of $\le \frac{1}{2^{5s}}$, using \autoref{useless0} and the above, we have that, with a probability (over $\tau$) of at least $1-\frac{1}{2^{5s}}$,

\begin{align*}
x \in L &\Rightarrow  \\
&f_{acc} \ge \frac{1}{2}+\epsilon - \zeta + (\zeta-\eta) - \frac{1}{2^{5s}} \ge  \frac{1}{2}+\epsilon - \zeta + \frac{(\zeta-\eta)}{2} > \frac{1}{2}+\epsilon - \zeta \\
x \not \in L &\Rightarrow \\
&f_{rej} \ge \frac{1}{2}+\epsilon - \zeta + (\zeta-\eta) - \frac{1}{2^{5s}} \ge  \frac{1}{2}+\epsilon - \zeta + \frac{(\zeta-\eta)}{2} >    \frac{1}{2}+\epsilon - \zeta 
\end{align*}

In other words, we have that $\mathcal{F}^L$ outputs $\perp$ with probability (over $\tau$) $\le \frac{1}{2^{5s}}$. This completes the proof of  \autoref{lemmvianisansprggeneral} for the case $\delta > 0$.

\subsection{When \texorpdfstring{$\delta = 0$}
{delta = 0}.}\label{deltazero}


For \(\delta = 0\), \(\mathcal{F}^L\) behaves exactly as it does for \(\delta > 0\), except when it falls into \hyperref[catch case]{Case $3$}, where it behaves slightly differently. In this case, it outputs accept if \(f_{acc} \ge \frac{1}{2} + \frac{\epsilon}{2}\) and reject if \(f_{rej} \ge \frac{1}{2} + \frac{\epsilon}{2}\) (recall that, by definition, \(\epsilon > 0\)). If neither of these conditions is true, as we show, it will compress the catalytic tape. Thus, \(\mathcal{F}^L\) never outputs \(\perp\).

Firstly, recall that when \(\mathcal{F}^L\) ends up in Case $3$, all the hash functions from \(h_1\) to \(h_l\) are good, and thus the error by PRG\(_l\) is at most \(\frac{1}{2^{5s}}\) (using \autoref{prgerrorapp}). Therefore, it can be argued, as done in \autoref{correctness}, that when \(\mathcal{F}^L\) does output accept or reject, it is correct. Next, observe that when \(\delta = 0\), by \autoref{taunodedefn}, the \(\tau^{\beta}\)-graph (\(\beta > 0\)) is precisely the configuration graph \(G_{M,x,\tau}\). In other words, the \(\tau^{\beta}\)-graph has no null vertex (\(\perp\)), since it follows from \autoref{tauepsp} that the probability of leaving the \(\tau^{\beta}\)-graph (from $\start$) to the \(\perp\)-vertex is \(0\), for \(\delta = 0\). Moreover, repeating the argument from \autoref{correctness}, one can verify that if the longest path  from $\start$ to $\acc$ or $\rej$ within the $\taubeta$-graph (which is the same as $\gmxtau$ in this scenario) has length $<2^{9s}$, then $\mathcal{F}^L$ outputs accept or reject, as either $f_{acc}$ or $f_{rej}$ is at least $\frac{1}{2} + \frac{\epsilon}{2}$. We now describe how $\mathcal{F}^L$ compresses the catalytic tape in the scenario in which neither of these fractions is at least $\frac{1}{2}+\frac{\epsilon}{2}$.


Let $P$ be the longest path starting from $\start$ in $G_{M,x,\tau}$. Then, we know that the length of $P$ is at least $2^{9s}$. Focus on the last $2^{9s}$ configurations along $P$, which we label as $v_1, v_2, \ldots, v_{2^{9s}}$, where each $v_i$ is defined as $\conf{\pi_i}{u_i}$. Then, by the definition of $P$ the longest path from $v_i$ (in $\gmxtau$) to either of the halt configurations $\acc$/$\rej$ has a length of $\le 2^{9s}$. Furthermore, for every $i\in [2^{9s}]$, $M$ always resets the catalytic tape to $\tau$, starting from $v_i$ (follows from the fact that $\delta=0$). Consider an arbitrary $i\in [2^{9s}]$. Then, without loss of generality, the probability that $M$, when run from $v_i$, reaches $\acc$ within $2^{9s}$ steps is at least $\frac{1}{2}$. However, because PRG$_l$ has an error of $\le \frac{1}{2^{5s}}$ and $l=2^{20s}$, using \autoref{errorpromise}, it follows that the fraction of seeds for which $y$-DFS (as per PRG$_l$) from $\accz$ visits $\confl{\pi_i}{u_i}{l}$ is at least $\frac{1}{2} - \frac{1}{2^{5s}}$. In general, for every $i$, the fraction of seeds for which the $y$-DFS from $\accz$ or $\rejz$ visits $\confl{\pi_i}{u_i}{l}$ is at least $\frac{1}{2}-\frac{1}{2^{5s}} \ge \frac{1}{4}$, for large enough $s$.  

 Consider the following set $S'_l$, which is defined similarly to the set $S_l$ (\autoref{sidefini}), with the difference that we now consider a configuration $\conf{v}{l}$ to be in the set $S'_l$ iff for at least $\frac{1}{4}$ (instead of $\frac{1}{2^{30s}}$) fraction of the seeds, the $y$-DFS (as per PRG$_l$) from $\accz$ or $\rejz$ visits $\conf{v}{l}$. That is, a configuration $\conf{v}{l}$ (in layer $l$) is in $S'_l$ iff:

 \begin{equation*}
 \Pr_{\text{seed}}[y\text{-DFS from } \accz \text{ or } \rejz \text{ visits } \conf{v}{l} \text{ for } y \leftarrow \text{PRG}_l(\text{seed})] \ge \frac{1}{4}.
 \end{equation*}

 Since, for every $i\in [2^{9s}]$, the configuration $\confl{\pi_i}{u_i}{l}$ lies in $S'_l$, we have that $|S'_l| \ge 2^{9s}$. Similar to subroutines $\countsl$ and $\indexsl$ (introduced in \hyperref[secondcasesubsection]{Case $2$}), we can implement $\countslprime$ and $\indexslprime$ to compute the size of the set $S'_l$ and to index every configuration in it. Recall that all the $y$-DFS from $\accz$ and $\rejz$, as per PRG$_{l-1}$, have a size of at most $H$ in this situation. Otherwise, $\mathcal{F}^L$ would have compressed the catalytic tape (and we would not have ended up in \hyperref[catch case]{Case $3$}). Thus, the size of each $y$-DFS from $\accz$ and $\rejz$, according to PRG$_l$, is at most $d_M H$. Therefore, referring to \autoref{subroutineclaim}, we can implement $\countslprime$ and $\indexslprime$ using $O(s)$ work space and in time $2^{O(s)}$.

The compression we describe is identical to the one used in \hyperref[case2pointone]{Case $2.1$}, with minor modifications.

\textbf{Compression}. Let \(val(\pmb{\mathrm{tar}}_{\le 9s})\) denote the integer value of the first \(9s\) bits of \(\pmb{\mathrm{tar}}\), incremented by \(1\), which is an integer in \([2^{9s}]\). Let \(\confl{\pi}{u}{l}\) be the configuration in \(S'_l\) with index \(val(\pmb{\mathrm{tar}}_{\le 9s})\). Using \(\indexslprime\), we can visit \(\confl{\pi}{u}{l}\). Similar to Case 2.1, we define the set TAU with minor changes. Now, TAU is the set of all \(c\)-bit catalytic settings such that \(\pi' \in \text{TAU}\) iff, for at least \(\frac{2^m}{4}\) seeds (instead of \(\frac{2^m}{2^{30s}}\)), the \(y\)-DFS from \(\confl{\pi}{u}{l}\) visits either \(\conf{\acc[\pi']}{0}\) or \(\conf{\rej[\pi']}{0}\) within \(d_{M}H\) steps. From the definition of TAU, it follows that \(|\text{TAU}| \leq 4\). Additionally, as in Case 2.1, we can argue that \(\tau\) (the string representing the first \(c\) bits of our catalytic tape) is in TAU. We can also determine the index \(idx\) of \(\tau\) in the set TAU. To achieve compression, we replace \(\tau\) with \(\pi\) (by visiting \(\confl{\pi}{u}{l}\)) and \(\pmb{\mathrm{tar}}_{\le 9s}\) with \(11 \circ idx \circ u \circ 0^{8s-4}\), freeing \(8s-4\) bits (where \(11\) denotes this compression type). This is feasible since \(idx\) takes up \(2\) bits and \(u\) requires \(s\) bits.

\textbf{Decompression}. The decompression is also the same as Case 2.1. Given the compressed tape, we are effectively standing at \(\confl{\pi}{u}{l}\) and know \(idx\). As in Case 2.1, using \(idx\), we can recover \(\tau\), such that after recovering \(\tau\), we can move back and forth between \(\accz\)/\(\rejz\) and \(\confl{\pi}{u}{l}\). Finally, as in Case 2.1, we can  figure out \(val(\pmb{\mathrm{tar}}_{\le 9s})\) (and hence recover \(\pmb{\mathrm{tar}}\)), by using \(\countslprime\) and \(\indexslprime\) to compute the index of $\confl{\pi}{u}{l}$ in the set $S'_l$.

Note that, as before, we can always copy the last $8s-4$ bits of $\pmb{\mathrm{tar}}$ into the freed-up bits so that we have free bits at the end of the catalytic tape.

The correctness of the compression scheme and the fact that the compression/decompression takes $O(s)$ space and runs in time $2^{O(s)}$ can be argued similarly to Case $2.1$.


\subsection{Proof of  \autoref{lemmvianisansprgforcl}}\label{lemmapart2}

Let $L \in \BPepsCdeltL{\delta}
{\epsilon}\PT$ where $\epsilon,\delta$ are constants such that $\delta < 2\epsilon$. Without loss of generality, let the machine for $L$ use $s= d\log{n}$ bits of work space, $c=n^d$ bits of catalytic space; and let it run in time $n^{d}$ for constants $d$ and input length $n$. Then, the proof of  \autoref{lemmvianisansprgforcl} follows from the proof of \autoref{lemmvianisansprggeneral}. This is because any path from $\start$ to $\acc$/$\rej$ in $\gmxtau$ has length $\le n^d = 2^s < 2^{9s}$. Thus, following our discussion from \autoref{correctness}, $\mathcal{F}^L$ never outputs $\perp$.

\section{Proof of  \autoref{exponetiallysmall}}\label{appexposmall}
We present a proof sketch here, assuming that the reader is already familiar with the de-randomization proof of the class $\BPCL$, as described by Cook et al. \cite{CookLiMertzPyne25}. Their proof uses a variant of the Nisan-Wigderson pseudorandom generator that enables random walks on the configuration graph. A key observation they make is that, after running a $\BPCL$ machine from some starting configuration $\start$, using the output of the PRG as random bits, we can return to $\start$ by performing a $y$-DFS (see \autoref{ydfs}) from the resulting configuration (which can be assumed to lie in layer $0$), for $y$ being all possible strings (and their prefixes) produced by the PRG. One of these $y$-DFS is guaranteed to encounter $\start$. Moreover, because the configuration graphs corresponding to different initial catalytic settings are vertex-disjoint, we can be confident that no other starting configuration apart from $\start$ is encountered in any of these $y$-DFS. However, this assurance fails when dealing with machines that do not always reset the catalytic tape, as the configuration graphs may no longer be disjoint.

Let $M$ be a $\BPepsCdeltSpace{d\log{n}}{n^d}{\frac{1}{2^{4n^d+3}}}{\epsilon}$ machine (for some constant $d$). The claim is that we can still use the same approach as in \cite{CookLiMertzPyne25}, where they use walks of length at most $4n^d$. 

\begin{claim}
Let \( v \) be a configuration obtained from \( \start[\tau] \) by following a walk of length \(4n^d\) (in $\gmxtau$). Then, while performing a \(y\)-DFS starting from \(\conf{v}{0}\), for some \(y\) with \(|y| \le 4n^d\), we cannot encounter a different starting configuration \(\conf{\start[\tau']}{*}\) with \(\tau' \ne \tau\).
\end{claim}
\begin{proof}
 Since \(v\) is reachable from \(\start[\tau]\) within \(4n^d\) steps, we assert that \(v\) must be a \(\taubeta[\frac{3}{4}]\)-node (see \autoref{taunodedefn}). Suppose this is not the case. Then, beginning from \(\start\), we would be able to destroy the catalytic tape with a probability of at least \(\frac{1}{2^{4n^d}} \cdot \frac{1}{4} = \frac{1}{2^{4n^d + 2}}\). This, however, would violate the definition of \(M\), which permits destruction of the catalytic tape only with a probability of at most \(\frac{1}{2^{4n^d + 3}}\).

Furthermore, if performing a $y$-DFS from \(\conf{v}{0}\) can lead us to \(\conf{\start[\tau']}{*}\), it suggests that \(v\) is reachable from \(\start[\tau']\) within \(4n^d\) steps. This means that \(v\) is a \(\tauprimebeta[\frac{3}{4}]\)-node as well. However, this implies that starting from \(v\), we could reach a halt configuration with catalytic tape contents \(\tau\), with a probability of \(\frac{3}{4}\)—and similarly for \(\tau'\)—which is not possible.
\end{proof}

\end{document}